\documentclass[11pt, oneside]{article}
\usepackage{etoolbox}

\usepackage{graphicx}	
\usepackage{amssymb}
\usepackage{amsfonts,latexsym,amsthm,amssymb,amsmath,amscd,euscript}
\usepackage{bbm}
\usepackage{framed}
\usepackage{fullpage}
\usepackage{mathrsfs}
\usepackage{physics}
\usepackage[hypertexnames=false]{hyperref}
\usepackage{enumitem}

\hypersetup{
  colorlinks=true,
  linkcolor=blue,
  filecolor=blue,
  citecolor = black,      
  urlcolor=cyan,
}
\usepackage{accents}
\usepackage{setspace}
\hypersetup{colorlinks=true,citecolor=blue,urlcolor=black,linkbordercolor={1 0 0}}
\usepackage{tikz-cd}
\newcount\Comments  %
\ifnum\Comments=1
\usepackage[inline]{showlabels}

\fi
\usepackage{turnstile}
\usepackage{mathpartir}
\usepackage{tikz}
\usepackage{xspace}

\usepackage{algorithm}
\usepackage{verbatim}
\usepackage[noend]{algpseudocode}

\newcommand{\nc}{\newcommand}

\usepackage[nameinlink,capitalize]{cleveref}
\crefformat{equation}{#2(#1)#3}
\Crefformat{equation}{#2(#1)#3}

\Crefformat{figure}{#2Figure #1#3}
\Crefname{assumption}{Assumption}{Assumptions}
\Crefformat{assumption}{#2Assumption #1#3}
   \Crefname{question}{Question}{Questions}
   \Crefformat{question}{#2Question #1#3}
   \Crefname{claim}{Claim}{Claims}
   \Crefformat{claim}{#2Claim #1#3}
   \Crefname{problem}{Problem}{Problems}
  \Crefformat{problem}{#2Problem #1#3}
\Crefname{subsubsection}{Section}{Sections}
\crefformat{subsubsection}{#2Section #1#3}
\Crefformat{subsubsection}{#2Section #1#3}

\nc{\sups}[1]{^{\scriptscriptstyle{#1}}}
\nc{\subs}[1]{_{\scriptscriptstyle{#1}}}

\newcommand{\wb}{\widebar}

\nc{\Critic}{\texttt{Critic}\xspace}
\nc{\PSDPUCB}{\texttt{PSDP-UCB}\xspace}
\nc{\LSVIUCB}{\texttt{LSVI-UCB}\xspace}
\nc{\Actor}{\texttt{Actor}\xspace}
\nc{\EstFeature}{\texttt{EstFeature}\xspace}
\nc{\ExpFTPL}{\texttt{ExpFTPL}\xspace}
\nc{\dist}{\mathrm{dist}}
\nc{\Bquad}{B^{\mathsf{quad}}}

\newcommand{\KLparam}{\mathsf{K}}

\makeatletter
\newtheorem*{rep@theorem}{\rep@title}
\newcommand{\newreptheorem}[2]{%
\newenvironment{rep#1}[1]{%
 \def\rep@title{#2 \ref{##1}}%
 \begin{rep@theorem}}%
 {\end{rep@theorem}}}
\makeatother

\makeatletter
\newcommand\xlabel[2][]{\phantomsection\def\@currentlabelname{#1}\label{#2}}
\makeatother

\theoremstyle{plain}
\newtheorem{theorem}{Theorem}
\newtheorem{lemma}[theorem]{Lemma}

\newtheorem{proposition}[theorem]{Proposition}

\newtheorem{claim}[theorem]{Claim}
\newtheorem{assumption}[theorem]{Assumption}
\theoremstyle{definition}
\newtheorem{definition}{Definition}

\newtheorem{question}[definition]{Question}
\newtheorem{problem}[definition]{Problem}

\numberwithin{theorem}{section}
\numberwithin{definition}{section}

\AddToHook{env/theorem/begin}{\crefalias{theorem}{theorem}}
\AddToHook{env/lemma/begin}{\crefalias{theorem}{lemma}}
\AddToHook{env/corollary/begin}{\crefalias{theorem}{corollary}}
\AddToHook{env/conjecture/begin}{\crefalias{theorem}{conjecture}}
\AddToHook{env/proposition/begin}{\crefalias{theorem}{proposition}}
\AddToHook{env/fact/begin}{\crefalias{theorem}{fact}}
\AddToHook{env/postulate/begin}{\crefalias{theorem}{postulate}}
\AddToHook{env/claim/begin}{\crefalias{theorem}{claim}}
\AddToHook{env/assumption/begin}{\crefalias{theorem}{assumption}}
\AddToHook{env/definition/begin}{\crefalias{definition}{definition}}
\AddToHook{env/defn/begin}{\crefalias{definition}{defn}}
\AddToHook{env/example/begin}{\crefalias{definition}{example}}
\AddToHook{env/remark/begin}{\crefalias{definition}{remark}}
\AddToHook{env/question/begin}{\crefalias{definition}{question}}
\AddToHook{env/problem/begin}{\crefalias{definition}{problem}}

\nc{\DMO}{\DeclareMathOperator}

\DMO{\prox}{prox}
\DMO{\UCB}{UCB}
\DMO{\LCB}{LCB}
\nc{\phidiff}{\phi\sups{\Delta}}
\nc{\pexp}{q_{\mathrm{exp}}}
\nc{\nn}{\nonumber}
\nc{\rk}{\mathrm{rk}}
\nc{\brk}[3]{{\rm br}_{#1}^{#2}({#3})}
\nc{\co}{{\rm co}}
\nc{\br}[2]{{\rm br}^{#1}({#2})}
\nc{\tA}{\textsc{A}}
\nc{\child}[2]{{\rm ch}_{#1}({#2})}
\nc{\parent}{\mathsf{pa}}
\nc{\dg}{\dagger}
\nc{\bB}{\mathbf{B}}
\nc{\Span}{\mathsf{span}}
\nc{\unif}{\mathsf{unif}}
\nc{\indsig}[2]{\mathcal{I}_{#1}({#2})}
\nc{\early}{{\rm pre}}
\nc{\zsink}{z_{\rm sink}}
\nc{\lowv}{{\rm low}}
\nc{\ol}{\overline}
\nc{\ul}{\underline}
\nc{\madec}[3]{\texttt{ma-dec}_{#1}({#2}, {#3})}
\nc{\madeco}[1]{\texttt{ma-dec}_{#1}}
\nc{\madecd}[3]{\texttt{ma-dec}^{\texttt{d}}_{#1}({#2}, {#3})}
\nc{\SF}{\mathscr{F}}
\nc{\SH}{\mathscr{H}}
\nc{\SP}{\mathscr{P}}
\nc{\SPc}{\wb{\mathscr{P}}}
\nc{\SB}{\mathscr{B}}
\nc{\SC}{\mathscr{C}}
\nc{\BS}{\mathbb{S}}
\nc{\PiMarkov}{\Pi^{\rm markov}}
\nc{\trunc}[2]{\mathsf{trunc}_{#2}({#1})}
\nc{\sbl}{of strong Bellman type\xspace}
\nc{\inormal}[1][\Phi, u,v]{\til{N}_{{#1}}}

\nc{\gamvec}{\gamma}
\nc{\til}{\widetilde}
\nc{\td}{\tilde}
\nc{\wh}{\widehat}
\nc{\old}[1]{\ifnum\Comments=1 {\color{brown}  [OLD: #1]}\fi}
\nc{\noah}[1]{\ifnum\Comments=1 {\color{purple} [ng: #1]}\fi}
\nc{\dhruv}[1]{\ifnum\Comments=1 {\color{red} [dr: #1]}\fi}
\nc{\mir}[1]{\ifnum\Comments=1 {\color{teal} [mc: #1]}\fi}
\nc{\sam}[1]{\ifnum\Comments=1 {\color{red} [sg: #1]}\fi}
\nc{\BP}{\mathbb{P}}
\nc{\BI}{\mathbb{I}}
\nc{\midpoint}[1][\Phi,\phi_1,\phi_2]{\mu^{\star}_{{#1}}}

\nc{\fools}[3]{\MF_{#3}({#1}, {#2})}
\nc{\fool}[2]{\MF({#1},{#2})}
\nc{\clip}[2]{{\rm clip}\left[ \left. {#1} \right| {#2} \right]}
\nc{\imax}{\omega}
\DMO{\conv}{conv}
\nc{\MH}{\mathcal{H}}
\nc{\MV}{\mathcal{V}}
\nc{\MC}{\mathcal{C}}
\nc{\MI}{\mathcal{I}}
\nc{\st}{\star}
\nc{\lng}{\langle}
\nc{\rng}{\rangle}
\DMO{\OOPT}{opt}
\nc{\dopt}[2]{\ell_{\OOPT}({#1},{#2})}
\nc{\MG}{\mathcal{G}}
\nc{\MP}{\mathcal{P}}
\nc{\PP}{\mathbb{P}}
\nc{\TT}{\mathbb{T}}
\nc{\TTmax}{\TT_{\max}}
\DMO{\REG}{Reg}
\DMO{\WREG}{wReg}
\nc{\reg}[2]{{\Delta}_{{#1}}({#2})}
\nc{\wreg}[2]{{\Delta}^{\rm w}_{{#1}}({#2})}
\nc{\Reg}[2]{{\REG}_{{#1}}({#2})}
\nc{\wReg}[2]{{\WREG}_{{#1}}({#2})}
\DMO{\Gap}{Gap}
\DMO{\GD}{GD}
\DMO{\GDA}{GDA}
\DMO{\EG}{EG}
\nc{\TE}{\til{\E}}
\nc{\Var}{\mathbf{Var}}
\DMO{\Cov}{Cov}
\DMO{\OGDA}{OGDA}
\DMO{\Unif}{Unif}
\nc{\Qu}{\ul{Q}}
\nc{\Qo}{\ol{Q}}
\nc{\Ro}{\ol{R}}
\nc{\Vu}{\ul{V}}
\nc{\Vo}{\ol{V}}
\nc{\RanQ}{\Delta Q}
\nc{\RanV}{\Delta V}
\nc{\clipQ}{\Delta \breve{Q}}
\nc{\frzQ}{\Delta \mathring{Q}}
\nc{\clipV}{\Delta \breve{V}}
\nc{\clipdelta}{\breve{\delta}}
\nc{\cliptheta}{\breve{\theta}}
\nc{\delmin}{\Delta_{{\rm min}}}
\nc{\delmins}[1]{\Delta_{{\rm min},{#1}}}
\nc{\gapfinal}[1]{\max \left\{ \frac{\frzQ_{{#1}}^{k^\st}(x,a)}{2H}, \frac{\delmin}{4H} \right\}}
\nc{\post}[2]{R({#1}; {#2})}
\nc{\posts}[3]{R_{#3}({#1}; {#2})}

\nc{\algnst}[1]{\begin{align*}#1\end{align*}}
\nc{\algn}[1]{\begin{align}#1\end{align}}
\nc{\matx}[1]{\left(\begin{matrix}#1\end{matrix}\right)}
\renewcommand{\^}[1]{^{(#1)}}

\nc{\nuu}{\nu}

\nc{\bel}[1]{\mathbf{b}({#1})}
\nc{\nbel}[1]{\bar{\mathbf{b}}({#1})}
\nc{\sbel}[2]{\mathbf{b}'_{#1}({#2})}
\nc{\nsbel}[2]{\bar{\mathbf{b}}'_{#1}({#2})}

\nc{\bv}{\mathbf{v}}
\nc{\bone}{\mathbf{1}}
\nc{\bX}{\mathbf{X}}
\nc{\bY}{\mathbf{Y}}
\nc{\bG}{\mathbf{G}}
\nc{\bz}{\mathbf{z}}
\nc{\bw}{\mathbf{w}}
\nc{\bA}{\mathbf{A}}
\nc{\bJ}{\mathbf{J}}
\nc{\bK}{\mathbf{K}}
\nc{\bb}{\mathbf{b}}
\nc{\ba}{\mathbf{a}}
\nc{\bc}{\mathbf{c}}
\nc{\bC}{\mathbf{C}}
\nc{\BR}{\mathbb R}
\nc{\BA}{\mathbb{A}}
\nc{\BC}{\mathbb C}
\nc{\bx}{\mathbf{x}}
\nc{\bS}{\mathbf{S}}
\nc{\bM}{\mathbf{M}}
\nc{\bR}{\mathbf{R}}
\nc{\bN}{\mathbf{N}}
\nc{\NN}{\mathbb{N}}
\nc{\by}{\mathbf{y}}
\nc{\sy}{y}
\nc{\sx}{x}

\nc{\MO}{\mathcal O}
\nc{\MU}{\mathcal{U}}
\nc{\ME}{\mathcal{E}}
\nc{\MN}{\mathcal{N}}
\nc{\MK}{\mathcal{K}}
\nc{\MM}{\mathcal{M}}
\nc{\MS}{\mathcal{S}}
\nc{\MT}{\mathcal{T}}
\nc{\BF}{\mathbb F}
\nc{\BQ}{\mathbb Q}
\nc{\MX}{\mathcal{X}}
\nc{\MA}{\mathcal{A}}
\nc{\MD}{\mathcal{D}}
\nc{\MB}{\mathcal{B}}
\nc{\MZ}{\mathcal{Z}}
\nc{\MJ}{\mathcal{J}}
\nc{\MW}{\mathcal{W}}
\nc{\MF}{\mathcal{F}}
\nc{\CF}{\mathcal{F}}
\nc{\MR}{\mathcal{R}}
\nc{\MY}{\mathcal{Y}}
\nc{\BZ}{\mathbb Z}
\nc{\BN}{\mathbb N}
\nc{\ep}{\epsilon}
\nc{\epbe}{\varepsilon_{\mathsf{BE}}}
\nc{\epout}{\varepsilon_{\mathsf{outlier}}}
\nc{\bellc}[1][h]{\MT_{#1}^\circ}
\nc{\vep}{\varepsilon}
\nc{\gapfn}[1]{\varepsilon_{#1}}
\nc{\ggapfn}[2]{\varphi_{#1}({#2})}
\nc{\epsahk}{\gapfn{0}}
\nc{\BH}{\mathbb H}
\nc{\BG}{\mathbb{G}}
\nc{\D}{\Delta}
\nc{\One}[1]{\mathbbm{1}\left\{{#1}\right\}}
\nc{\bOne}{\mathbf{1}}
\nc{\Aopt}{\mathcal{A}^{\rm opt}}
\nc{\Amul}{\mathcal{A}^{\rm mul}}

\nc{\SQ}{\mathsf Q}

\nc{\DO}{\accentset{\circ}{\D}}
\nc{\mf}{\mathfrak}
\nc{\mfp}{\mathfrak{p}}
\nc{\mfq}{\mf{q}}
\nc{\mfx}{\mf{s}}
\nc{\Sp}{\mbox{Spec}}
\nc{\Spm}{\mbox{Specm}}
\nc{\hookuparrow}{\mathrel{\rotatebox[origin=c]{90}{$\hookrightarrow$}}}
\nc{\hookdownarrow}{\mathrel{\rotatebox[origin=c]{-90}{$\hookrightarrow$}}}
\nc{\hra}{\hookrightarrow}
\nc{\tra}{\twoheadrightarrow}
\nc{\sgn}{{\rm sgn}}
\nc{\muideal}{\mu_{\mathsf{ideal}}}
\nc{\aut}{{\rm Aut}}
\nc{\Hom}{{\rm Hom}}
\nc{\img}{{\rm Im}}
\DMO{\id}{Id}
\DMO{\supp}{supp}
\DMO{\KL}{KL}
\nc{\kld}[2]{D_{\mathsf{KL}}({#1}||{#2})}
\nc{\ren}[2]{D_2({#1}||{#2})}
\nc{\chisq}[2]{\chi^2({#1}||{#2})}
\nc{\tvd}[2]{D_{\mathsf{TV}}({#1}, {#2})}
\nc{\hell}[2]{d_{\mathsf{H}}^2({#1}, {#2})}
\nc{\dbi}[3][\pi]{D_{\mathsf{bi}}^{#1}({#2} \| {#3})}
\DMO{\BSS}{BSS}
\DMO{\BES}{BES}
\DMO{\BGS}{BGS}
\DMO{\poly}{poly}
\nc{\indep}{\perp}
\DMO{\sink}{sink}
\nc{\fp}[1]{\MP_1({#1})}
\nc{\BO}{\mathbb{O}}
\nc{\BT}{\mathbb{T}}

\nc{\RR}{\mathbb{R}}
\nc{\Gradient}{\nabla}
\DMO{\diag}{diag}
\nc{\EE}{\mathbb{E}}
\nc{\MQ}{\mathcal{Q}}
\nc{\ML}{\mathcal{L}}
\nc{\cPhi}{\bar \Phi}

\nc{\E}{\mathbb{E}}
\nc{\ra}{\rightarrow}

\nc{\pmhc}[1]{\{-1,1\}^{#1}}
\nc{\Dbnd}{D}
\nc{\Bbnd}{B}

\nc{\Key}{\mathsf{KeyGen}}
\nc{\Enc}{\mathsf{Encode}}
\nc{\Encemb}{\mathsf{EncodeEmb}}
\nc{\Dec}{\mathsf{Decode}}
\nc{\sk}{\mathsf{sk}}
\nc{\pk}{\mathsf{pk}}
\nc{\lpk}{\ell_{\mathsf{pk}}}
\nc{\lsk}{\ell_{\mathsf{sk}}}
\nc{\msg}{\mathsf{m}}
\nc{\Adv}{\mathsf{Adv}}
\nc{\Red}{\mathsf{Red}}
\nc{\negl}{\mathsf{negl}}
\nc{\Ber}{\mathrm{Ber}}
\nc{\PRFPRC}{\mathsf{PRF\text{-}PRC}}
\nc{\wt}{\mathrm{wt}}
\nc{\res}[2]{{#1}_{#2}}
\nc{\bzero}{\mathbf{0}}
\nc{\Bin}{\mathrm{Bin}}
\nc{\Hyp}{\mathrm{Hyp}}

\nc{\Nrho}[1][\rho]{{N}_{#1}}
\nc{\Trho}[1][\rho]{\mathsf{T}_{#1}}
\nc{\hc}[1][n]{\{0,1\}^{#1}}
\nc{\Stab}{\mathbf{Stab}}
\nc{\bW}{\mathbf{W}}
\nc{\NS}{{\mathbf{NS}}}

\nc{\KeyS}{\mathsf{KeyGen_{Sub}}}
\nc{\EncS}{\mathsf{Encode_{Sub}}}
\nc{\DecS}{\mathsf{Decode_{Sub}}}
\nc{\WeightPerturb}{\mathsf{WeightPerturb}}
\nc{\Unique}{\mathsf{Unique}}
\nc{\PRCS}{\mathsf{PRC_{Sub}}}
\nc{\PRC}{\mathsf{PRC}}
\nc{\PRCI}{\mathsf{PRC_{Idx}}}
\nc{\SampleUnique}{\mathsf{SampleUnique}}
\nc{\PerturbDifference}{\mathsf{PerturbDifference}}

\nc{\Model}{\mathsf{Model}}
\nc{\Modelo}{\overline{\Model}}
\nc{\prompt}{\mathtt{PROMPT}}
\nc{\Setup}{\mathsf{Setup}}
\nc{\Detect}{\mathsf{Detect}}
\nc{\Sigprc}{\Sigma_{\mathsf{PRC}}}
\nc{\Wat}{\mathsf{Wat}}
\nc{\term}{\mathtt{END}}
\nc{\tok}{\mathsf{t}}
\nc{\True}{\textsf{True}}
\nc{\False}{\textsf{False}}
\nc{\Eemb}{\ME_{\mathsf{Emb}}}
\nc{\hist}{\mathsf{hist}}
\nc{\hh}{\mathsf{h}}
\nc{\freq}{\mathsf{freq}}
\nc{\ff}{\mathsf{f}}

\nc{\Hemp}[1]{H_{\mathsf{e}}^{#1}}
\nc{\Hempt}[1]{\bar{H}_{\mathsf{e}}^{#1}}
\nc{\Hemptil}[1]{\tilde{H}_{\mathsf{e}}^{#1}}
\nc{\Hmean}[1]{H_{\mathsf{m}}^{#1}}
\nc{\partition}[1][n,q]{P^{\mathsf{ptn}}_{#1}}
\nc{\Crob}{C_{\mathsf{rob}}}
\nc{\Lmax}{L_{\mathsf{max}}}
\nc{\skwat}{\sk_{\mathsf{Wat}}}
\nc{\EmbedToken}{\mathsf{EmbedChar}}
\nc{\len}{\mathrm{len}}
\nc{\Esub}{\ME_{\mathsf{sub}}}
\nc{\Ecomp}{\ME_{\mathsf{comp}}}
\nc{\comp}{\mathsf{c}}
\nc{\SE}{\mathscr{E}}
\nc{\alphb}{q}
\nc{\tAdv}{\widetilde{\Adv}}
\nc{\Funif}{{F_{\mathsf{Unif}}}}
\nc{\Alg}{\mathsf{Alg}}
\nc{\Majority}{\mathsf{Maj}}
\nc{\Dist}{\mathsf{Dist}}
\nc{\edit}{edit\xspace}
\nc{\Edit}{Edit\xspace}
\nc{\Wcomp}{\MW^{\mathsf{comp}}}

\nc{\INS}{\mathsf{INS}}
\nc{\CNS}{\mathsf{CNS}}
\nc{\cdist}{\stackrel{\mathrm{c}}{\sim}}
\nc{\SU}{\mathscr{U}}
\nc{\rr}{\bar{n}}

\nc{\KeyGen}{\mathsf{KeyGen}}
\nc{\ED}{D_{\mathsf{ED}}}
\nc{\Ham}{D_{\mathsf{Ham}}}
\nc{\bin}{\mathsf{bin}}
\nc{\EDball}{\mathcal{B}_{\mathsf{ED}}}
\nc{\SEDball}{\mathcal{B}_{\mathsf{Ham,ED}}}
\nc{\LEDball}{\mathcal{B}_{\mathsf{len,ED}}}
\nc{\epED}{\varepsilon_{\mathsf{ED}}}
\nc{\epDec}{\varepsilon_{\mathsf{Dec}}}
\nc{\Eedit}{\mathscr{E}^{\mathsf{edit}}}
\nc{\Egood}{\ME^{\mathsf{good}}}
\nc{\Ebad}{\ME^{\mathsf{bad}}}
\nc{\PermEnc}{\mathsf{PermEncode}}
\nc{\dham}{d_{\mathsf{H}}}
\nc{\dedit}{d_{\mathsf{E}}}
\nc{\pDec}{p_{\mathsf{Dec}}}
\nc{\Rclean}{R_{\mathsf{Clean}}}
\nc{\adv}{\mathcal{A}}
\nc{\bi}{\mathbf{i}}

\nc{\bj}{\mathbf{j}}
\nc{\bp}{\mathbf{p}}
\nc{\Ologit}{\MO_{\mathsf{logit}}}
\nc{\Osamp}{\MO_{\mathsf{samp}}}
\nc{\epapx}{\varepsilon_{\mathsf{apx}}}
\nc{\DistSpanner}{\textsf{DistSpanner}\xspace}
\nc{\epbase}{\vep_{\mathsf{base}}}
\nc{\cnorm}{\beta}
\nc{\occ}{\mathsf{occ}}

\nc{\phiprev}{\phi^{\mathsf{prev}}}
\nc{\depth}{\mathsf{depth}}
\nc{\range}{\mathsf{range}}
\nc{\Spread}{\mathsf{Spread}}

\nc{\Ebase}{\ME^{\mathsf{base}}}
\nc{\Eincr}{\ME^{\mathsf{incr}}}
\nc{\MIN}{\mathsf{MIN}}
\nc{\MAX}{\mathsf{MAX}}

\nc{\thrup}{\tau_{\mathsf{upp}}}
\nc{\thrlo}{\tau_{\mathsf{low}}}

\nc{\Sact}{\MS^{\mathsf{act}}}
\nc{\Sterm}{\MS^{\mathsf{term}}}
\nc{\mualg}{\mu_{\mathsf{alg}}}
\nc{\Erej}{\ME^{\mathsf{rs}}}
\nc{\Eterm}{\ME^{\mathsf{term}}}
\nc{\N}{\mathsf{N}}
\nc{\Lconst}{\mathsf{L}}
\nc{\tilmualg}{\tilde{\mu}^{\mathsf{alg}}}
\nc{\Gret}{\MG^{\mathsf{ret}}}
\nc{\Nint}{\MN^{\mathsf{int}}}
\nc{\Nappend}{\MN^{\mathsf{append}}}
\nc{\Nfinal}{\MN^{\mathsf{fin}}}
\nc{\prev}{\mathfrak{p}}
\nc{\Erestrict}{\ME^{\mathsf{restrict}}}
\nc{\Vinit}{V^{\mathsf{init}}}
\nc{\vinit}{v^{\mathsf{init}}}
\nc{\query}{\mathsf{query}}
\nc{\key}{\mathsf{key}}
\nc{\base}{\mathsf{base}}
\nc{\Nvalid}{N_{\mathsf{valid}}}
\nc{\Vvalid}{\MV^{\mathsf{valid}}}
\nc{\Equery}{\ME^{\mathsf{query}}}
\nc{\tAlg}{\widetilde{\Alg}}
\nc{\piref}{\pi_{\mathsf{ref}}}

\title{The Power of Test-Time Training for Approximate Sampling

}
\author{
Noah Golowich\\
Microsoft Research NYC\\
\texttt{nzg@cs.utexas.edu}
\and
Ankur Moitra\\
MIT\\
\texttt{moitra@mit.edu}
\and
Dhruv Rohatgi\\
MIT\\
\texttt{drohatgi@mit.edu}
}
\date{\today}

\begin{document}
\pagenumbering{gobble}
\maketitle
\begin{abstract}
Efficiently sampling from a complex probability distribution is a fundamental problem across machine learning and theoretical computer science. It has become increasingly pertinent in recent years with the rise of generative AI, as sophisticated sampling procedures from large language models (LLMs) have been proposed to solve challenging reasoning problems spanning domains such as mathematics and coding. The efficacy of such sampling algorithms is limited, however, by the relationship between the LLM and the particular sampling task at hand, which has motivated the framework of \emph{test-time training (TTT)}. TTT works by updating a model's weights in response to partial generations and reward feedback received \emph{at inference time}, thus adapting to the particular problem. In this work, we propose a formalization for TTT as the problem of producing a sample from a given probability measure $\mu^\star$ belonging to a known class $\MF$ of distributions, given an oracle $\hat \mu$ which yields approximate density estimates for $\mu^\star$. As it turns out, this is closely related to the problem of reducing sampling to approximate counting studied in seminal works of Jerrum, Valiant \& Vazirani (1986) and Jerrum \& Sinclair (1989): namely, when $\MF$ is the class of \emph{all distributions}, it coincides \emph{exactly} with the aforementioned counting-to-sampling reduction.

In this paper, we first show a quadratic lower bound on the query complexity of sampling from $\mu^\star$ given query access to $\hat \mu$ (which holds for sufficiently large classes $\MF$), thus showing that the random walk approach proposed by Jerrum \& Sinclair (1989), and later refined by Hayes \& Sinclair (2010), is optimal. This result answers an open question explicitly posed by Hayes \& Sinclair. We then show that this lower bound can be circumvented if the size of $\MF$ is bounded appropriately. As we discuss, this latter result can be viewed as an abstraction of TTT, and thus represents a starting point for the development of a principled theoretical framework for TTT.
\end{abstract}

\newpage\tableofcontents
\newpage
\pagenumbering{arabic}
\setcounter{page}{1}
\section{Introduction}
Efficiently sampling from a given probability distribution is a central problem in a number of areas spanning machine learning, artificial intelligence, and theoretical computer science. Recently, generative AI has enabled sampling from natural yet highly complex distributions. As some examples, diffusion models~\cite{ho2020denoising,song2020sde,sohl2015deep} enable sampling from a distribution over images via an iterative denoising process; and autoregressive large language models (LLMs) \cite{brown2020fewshot,openai2023gpt4,google2023gemini} enable generating text by repeatedly sampling the next token conditioned on all previous tokens. Each of these architectures was designed to be used with a specific sampling procedure. However, recent efforts have focused on more complex sampling procedures --- e.g., generating multiple strings of text in parallel, or going back and revising generated text --- which can produce samples from the model which are superior according to some metric. For instance, one might desire text written in a certain style, or to sample from the distribution of solutions to a math problem \emph{conditioned on them being correct}.  

\paragraph{From inference to test-time training.} 
A number of heuristic approaches have been proposed for these problems: %
they typically make many queries to a base (or \emph{reference})  model $\piref$, as well as to an additional \emph{reward model} $\hat V$, which at a high level provides information about the quality of \emph{complete or partial generations}. For instance, in settings such as math or coding where answers can be verified, $\hat V$ may directly assign a reward to proposed solutions \cite{yuksekgonul2026ttd,zuo2025ttrl,yang2026ttcs}, or more generally $\hat V$ itself may be a model which is trained to score partial generations (in this context, it is often referred to as a \emph{process reward model (PRM)}) \cite{wang2025value,puri2025probabilistic,rohatgi2025taming,lightman2023let,zhang2025lessons,golowich2026reject}. In this paper we focus on the case where $\piref$ is a \emph{language model}, though similar frameworks can apply for other settings such as diffusion models. 

While such approaches have shown promising performance, their efficacy is limited by the quality of the reference model $\piref$. If the distribution of generations from $\piref$ is far from the desired target distribution, then the number of queries to $\piref$ and $\hat V$ needed to obtain satisfactory samples may be prohibitively large. For instance, $\piref$ may have been trained to solve general math problems, but if the desired target distribution is over solutions to a particular combinatorics conjecture, then $\piref$ may suggest many approaches which are unlikely to work in that area of combinatorics. 

To help to mitigate this issue, a recent line of work has studied \emph{test-time training (TTT)}, which proposes to \emph{update the weights of the base model} $\piref$ at inference time, using partial generations produced by $\piref$ as well as feedback from $\hat V$ \cite{hardt2024testtime,sun2023learn,sun2019tt,sun2025ttlm}.\footnote{We remark that several of the initial works on TTT for LLMs (e.g., \cite{sun2025ttlm,hardt2024testtime}) do not make use of any external rewards (modeled here by $\hat V$) to update $\piref$, and instead only rely on the prompt, and potentially partial generations under $\piref$. However, the larger-scale developments, such as those achieving IMO-level performance for the first time \cite{trinh2025olympiad} or advancing the state-of-the-art for open mathematics questions \cite{yuksekgonul2026ttd}, do make use of external reward signals (typically by using reinforcement learning algorithms at test time \cite{zuo2025ttrl}).} The hope is that this feedback will push $\piref$ to be ``closer'' to the particular target distribution. In many such cases, TTT achieves impressive performance, beating the performance of existing sampling approaches which held the weights of the reference model $\piref$ fixed. For instance, reinforcement learning at test time \cite{zuo2025ttrl} was used in AlphaProof \cite{trinh2025olympiad}, which was the first model to achieve silver-level performance at the IMO; \cite{yuksekgonul2026ttd} achieved state-of-the art performance using TTT for a number of open mathematics and engineering problems.

In the coming years it is likely that language models will be used to solve even more challenging problems by investing large amounts of compute, and test-time training is at the forefront of techniques which will drive success in this endeavor. That said, for the most part the literature lacks a principled theoretical framework for reasoning about how test-time training works and in what manner it can lead to improved or more efficient sampling. In this work, we aim to fill this gap. %

\subsection{A framework for TTT in the context of LLM inference}
\label{sec:inference-llms}
We develop a formalization for  TTT by viewing the sampling problem discussed above as that of \emph{tilting} the output distribution of the reference language model, a common setting in the literature \cite{rohatgi2025taming,geuter2025guided,foster2025foundation,xiong2024iterative}.\footnote{While we focus on the setting of language models, similar logic applies to other settings, such as diffusion models (see \cite{singhal2025general}).}
In particular, we suppose we are given a \emph{reference} model $\piref \in \Delta(\Sigma^n)$, which is an \emph{autoregressive language model} on sequences of length $n$: this means that for any partial sequence $x_{1:i} = (x_1, \ldots, x_i) \in \Sigma^i$ (for $i \leq n$), we can query $\piref$ for the conditional distribution of the last token given the previous ones, i.e., $\piref(x_i \mid x_{1:i-1})$.\footnote{In fact, we assume that $\piref(x)$ can be determined using a single query; this is the case for typical autoregressive language models, which give the log-probabilities of \emph{all} tokens in a sequence using a single forward pass.} We are also given a \emph{value function} $V^\star : \Sigma^n \to \BR$, which assigns a ``quality score'' to each full sequence: for instance, $V^\star(x)$ might score the style of $x$ or its mathematical correctness. 
Our goal is to sample from the \emph{tilted distribution} $\mu^\star \in \Delta(\Sigma^n)$ obtained by reweighting $\piref$ according to the exponentiated value function, defined by
\begin{align}
\mu^\star(x) := \frac{1}{Z} \piref(x) \cdot \exp( V^\star(x)), \qquad Z = \E_{x \sim \piref}[\exp(V^\star(x))].\label{eq:mustar-tilted-piref}
\end{align}

\paragraph{Sampling from $\mu^\st$ via a PRM.} We extend the domain of $V^\star$ to all partial sequences (namely, to $\Sigma^{\leq n}$, which we let denote the set of sequences of elements in $\Sigma$ of length at most $n$) by letting $V^\star(x_{1:i})$ be the log-expectation of the score of a full sequence starting with $x_{1:i}$: namely, for $x_{1:i} \in \Sigma^i$, we define 
\begin{align}
V^\star(x_{1:i}) = \log \E_{x_{i+1:n} \sim \piref(\cdot \mid x_{1:i})}[\exp(V^\star(x_{1:n}))]\nonumber.
\end{align}
Roughly speaking, $V^\star(x_{1:i})$ should be interpreted as a measure of how likely we are to achieve a high overall value if we generate a full sequence under $\piref$, starting from $x_{1:i}$. 
While sampling from the tilted measure $\mu^\star$ without any guidance is typically intractable, it becomes tractable (specifically, can be solved in linear time in $n$) if we are given exact query access to $V^\star$ on all partial sequences \cite{yang2021fudge,rohatgi2025taming}. 

Of course, exactly querying $V^\star$ on a partial sequence is itself typically an intractable problem. A more reasonable assumption is that we are given \emph{approximate} query access to $V^\star$ on all partial sequences. A great deal of effort has gone into finding such $\hat V$; depending on the context, they may be referred to as \emph{KL-regularized value functions} \cite{rohatgi2025taming,foster2025foundation} or \emph{process reward models (PRMs)} \cite{golowich2026reject,wang2025value,puri2025probabilistic,rohatgi2025taming,lightman2023let,zhang2025lessons}. In particular, recent theoretical literature \cite{rohatgi2025taming,zhu2026power,golowich2026reject} assumes that we can query a mapping $\hat V : \Sigma^{\leq n} \to \BR$, so that, for some error $\rho > 0$, we have $|\hat V(x_{1:i}) - V^\star(x_{1:i})| \leq \rho$ for all sequences $x_{1:i} \in \Sigma^i$ ($i \in [n]$).\footnote{It is also assumed that $\hat V$ and $V^\star$ are equal on length-$n$ sequences, which is without loss since query access to $V^\star$ on full-length sequences is allowed.}\footnote{\cite{rohatgi2025taming,golowich2026reject} also study relaxations of this assumption where the error bound holds in a certain average-case sense, rather than uniformly over all sequences.} 

The below question  captures the problem of sampling from the tilted measure $\mu^\star$, using query access to $\piref$ and $\hat V$, \emph{without doing any TTT}:
\begin{question}[Sampling without TTT]
\label{ques:tilted-nottt}
In the above setting, what is the minimum number of queries to $\piref$ and $\hat V$ which are needed to produce a single sample from $\mu^\star$ (perhaps approximately so)?
\end{question}

\paragraph{Sampling with TTT.} \emph{How can we effectively model TTT approaches which update the weights of the model $\piref$ at inference time?} In general, such approaches should be seen as attempting to \emph{learn} the distribution $\mu^\star \propto \piref \cdot V^\star$ --- not in its entirety, but at least well enough so as to generate a single sample from it. Any potential benefit of TTT must stem from the inductive bias of the model architecture, i.e. the fact that updating the weights of the model is more structured than directly updating the distribution, and therefore may enable more generalizable learning. 

Inspired by statistical learning theory, we therefore formalize test-time training as making an additional assumption that $\mu^\star$ lies in a restricted class of models --- or equivalently (since $\piref$ is fixed), that $V^\star$ lies in some known class $\MV$ consisting of functions $V : \Sigma^n \to \BR$. Thus, the question of sampling from $\mu^\star$ \emph{with TTT} may be phrased as follows:
\begin{question}[Sampling with TTT]
\label{ques:tilted-ttt}
Suppose we are given a class $\MV$ of mappings $V : \Sigma^{\leq n} \to \BR$, with the promise that $V^\star \in \MV$. Using knowledge of $\MV$, what is the minimum number of queries to $\piref$ and $\hat V$ which are needed to produce a single (approximate) sample from $\mu^\star$?
\end{question}
Comparing \cref{ques:tilted-nottt} and \cref{ques:tilted-ttt}, the problem of determining when TTT is beneficial is as follows: \emph{\textbf{How do structural conditions on the class $\MV$ affect the query complexity of sampling from $\mu^\star$? In particular, for what choices of $\MV$ is the query complexity of \cref{ques:tilted-ttt} smaller than that of \cref{ques:tilted-nottt}?}} In this paper we show a positive answer to this second question when the cardinality of the class $\MV$ is not too large. There are a number of further structural conditions which may be placed on $\MV$, analogous to the vast literature in statistical learning theory which characterizes the sample complexity of learning \emph{infinite} hypothesis classes satisfying various capacity constraints. We expect these will lead to a plethora of fruitful directions for follow-up work; see \cref{sec:discussion}. 

\subsection{Interpretation in the context of counting-to-sampling reductions}
Beyond the motivation from recent progress in generative AI, our setting can be cast in the framework of a rich body of historical work in theoretical computer science on sampling algorithms. To explain this connection, we consider a slight rephrasing of the setting of \cref{sec:inference-llms}. Namely, for $x_{1:i} \in \Sigma^i$, we define
\begin{align}
\hat \mu(x_{1:i}) := \frac{1}{Z} \cdot \piref(x_{1:i}) \cdot \exp (\hat V(x_{1:i})) %
\label{eq:define-hatmu}.
\end{align}
It is straightforward to see that the marginal of $\mu^\star$ on the first $i$ coordinates, denoted $\mu^\star(x_{1:i})$, satisfies $\mu^\star(x_{1:i}) = \frac{1}{Z} \cdot \piref(x_{1:i}) \cdot \exp(V^\star(x_{1:i}))$, meaning that $\hat \mu(x_{1:i})$ should be interpreted as an approximation of $\mu^\star(x_{1:i})$ induced by $\hat V$. 
Moreover, given a class $\MV$ of mappings $V : \Sigma^n \to \BR$ as above, let $\MF = \MF_{\MV} \subset \Delta(\Sigma^n)$ denote the class of distributions 
\begin{align}
\MF_{\MV}: =\{ \piref \cdot \exp(V) \cdot Z_V^{-1} \mid V \in \MV \}, \qquad Z_V = \E_{x \sim \piref}[\exp(V(x))]\label{eq:mf-mv}.
\end{align}
We now observe that the sampling problem in either \cref{ques:tilted-nottt} or \cref{ques:tilted-ttt} with respect to the target value function $V^\star$ and the oracles $\piref, \hat V$ is precisely an instance of the problem of sampling from $\mu^\star$ given the ability to query the oracle $\hat \mu$ (and in the case of \cref{ques:tilted-ttt}, the knowledge that $\mu^\star \in \MF_{\MV}$). This observation is immediate from the definitions of $\hat \mu$ and $\MF_{\MV}$, save for one subtlety: per \cref{eq:define-hatmu}, computing $\hat \mu(x_{1:i})$ requires knowledge of $Z$ in addition to the ability to query $\piref$ and $\hat V$, so it is not immediately evident that the two query models are equivalent. However, a key observation is that knowledge of $Z$ is not required if the sampling algorithm which uses query access to $\hat \mu$ is invariant under any rescaling of $\hat \mu$. This will be the case for our algorithms (\cref{sec:ub-proof}), which treat the setting of having query access to $\hat \mu$;\footnote{To be more precise, our algorithms work in the setting where we are given a class $\MV$ of mappings $V : \Sigma^n \to \BR$ and an oracle $\hat V : \Sigma^n \to \BR$ so that $|V(x_{1:i}) - \hat V(x_{1:i})| \leq \log R$ for all $x_{1:i}$.} thus they indeed can be simulated using query access to $\piref$ and $\hat V$, even without knowledge of $Z$.\footnote{Furthermore, in our lower bounds, the constant $Z$ is known.} See \cref{prop:muhat-vhat} for a formal statement of this equivalence.

\emph{In what sense is $\hat \mu$ a good approximation of $\mu^\star$?} The assumption that $\| \hat V - V^\star\|_\infty \leq \rho$ and that $\hat V(x) = V^\star(x)$ for $x \in \Sigma^n$ from \cref{sec:inference-llms} implies that $ \hat \mu$ is an $R$-multiplicative approximation of $\mu^\star$ with $R = \exp(\rho)$, i.e., %
\begin{align}
\frac{1}{R} \leq \frac{\mu^\star(x_{1:i})}{\hat\mu(x_{1:i})} \leq R \quad \text{for all } x_{1:i} = (x_1, \ldots, x_i) \in \Sigma^{i},\ i \leq n, \qquad \mu^\star(x) = \hat\mu(x) \text{ for all } x \in \Sigma^n.\label{eq:r-mult-intro}
\end{align}
For simplicity, we assume that $R$ is a constant in this exposition. 
Thus, \cref{ques:tilted-nottt} may be rephrased as follows: \emph{Assuming \cref{eq:r-mult-intro}, how many queries to the oracle $\hat\mu(\cdot)$ are needed to output a sample from a distribution $\hat\nu$ which is close to $\mu^\star$?}  %
This question has a deep history pertaining to the algorithmic equivalence of sampling and approximate counting, as developed in seminal works of Jerrum, Valiant, and Vazirani \cite{jerrum1986random} and Jerrum and Sinclair \cite{jerrum1989approximate}.\footnote{Technically, this equivalence holds only for self-reducible problems, though a great number of concrete problems of interest satisfy this property.} This equivalence has been instrumental in the design of efficient randomized algorithms for a wide range of problems arising in combinatorics and statistical physics, including estimating the permanent of a nonnegative matrix~\cite{jerrum2004permanent}, sampling proper colorings of graphs~\cite{jerrum1995kcolorings,vigoda1999improved}, simulating the Ising and Potts models~\cite{jerrum1993ising,chen2021spectral}, counting matchings~\cite{sinclair1989approximating}, and sampling bases of matroids~\cite{anari2020logconcave}. %

In more detail, under the assumption \cref{eq:r-mult-intro}, the mapping $\hat \mu : \Sigma^{\leq n} \to \BR_{\geq 0}$ may be viewed as an $R$-approximate \emph{counting oracle} for the distribution $\mu^\star$.\footnote{In general, the interpretation of $\hat \mu$ as an approximate counting oracle for the ``class of problems'' captured by $\mu^\star$ makes sense as long as the problem in question is self-reducible, which in turn is a ubiquitous assumption across a wide range of algorithmic problems.}  The seminal work of Jerrum \& Sinclair \cite{jerrum1989approximate}, which was later refined by Hayes \& Sinclair \cite{hayes2014liftings}, showed that $\tilde O(n^2)$ queries to the counting oracle $\hat \mu$ suffice to generate a sample from a distribution which is $o(1)$-close to $\mu^\star$. The approach of \cite{jerrum1989approximate,hayes2014liftings} relies on the clever idea of \emph{backtracking}: we can try to sample a sequence $x_1, x_2, x_3, \ldots$, character-by-character, by ``pretending'' that $\hat \mu(\cdot)$ denotes the actual marginal probabilities of $\mu^\star$. If at some point, we ``make a mistake'', as evidenced by the fact that we have generated some sequence $x_{1:i}$ for which $\hat \mu(x_{1:i})$ is particularly small, then we can remove the most recently generated character, to yield the sequence $x_{1:i-1}$, and continue. Formally, this process may be viewed as a random walk on the complete $|\Sigma|$-ary tree of depth $n$, where edge weights are given by $\hat \mu(\cdot)$. 

In more recent years, a number of papers \cite{schweizer2012non,zhu2026power,golowich2026reject} have recovered the guarantee of \cite{jerrum1989approximate,hayes2014liftings} that $\tilde O(n^2)$ queries to $\hat \mu$ suffice to approximately sample from $\nu^\star$, using instead approaches from the family of of \emph{Sequential Monte Carlo (SMC)} (i.e., \emph{particle filtering}) algorithms. These approaches have the benefit that they can be parallelized, i.e., the queries can be arranged so that only $\tilde O(n)$ rounds are needed where each round performs $\tilde O(n)$ queries in parallel. Nevertheless, the total work remained as $\tilde O(n^2)$, thus motivating the following question, which was explicitly posed by Hayes \& Sinclair in 2010 \cite{hayes2014liftings}: 
\begin{question}[\cite{hayes2014liftings}]
\label{ques:main}
Is it possible to approximately sample from $\mu^\star$ using $o(n^2)$ queries to $\hat\mu(\cdot)$ under the assumption \cref{eq:r-mult-intro}?
\end{question}
In this paper, our first main result resolves \cref{ques:main} in the negative, giving a near-quadratic lower bound on the number of queries. As discussed above, such a lower bound also serves to show limits to what is possible for \cref{ques:tilted-nottt}. This serves as motivation for \emph{test-time training}: We proceed to show that this lower bound can be circumvented if the target distribution $\mu^\star$ is known to belong to a structured class $\MF$ of distributions of bounded size, thus giving a positive answer to \cref{ques:tilted-ttt}.

\subsection{Results}
\cref{thm:main-lb} below gives a negative answer to \cref{ques:main}, showing that that the $O(n^2)$ query complexity of the Jerrum-Sinclair algorithm \cite{jerrum1989approximate,hayes2014liftings} is essentially optimal:
\begin{theorem}[Main lower bound; informal version of \cref{thm:weak-lb}]
\label{thm:main-lb}
There is no algorithm which receives access to an oracle $\hat \mu : \{0,1\}^{\leq n} \to \BR_{\geq 0}$ satisfying \cref{eq:r-mult-intro} with $R = O(1)$ with respect to some unknown distribution $\mu^\star \in \Delta(\{0,1\}^n)$, makes $o(n^2 / \log^2 n)$ queries to $\hat \mu(\cdot)$, and outputs a sample from a distribution $\hat \nu$ satisfying $\tvd{\mu^\star}{\hat \nu} \leq 1/10$. 

Moreover, the same result holds even if $\mu^\star$ is promised to be in a \emph{known} class of densities of size at most $\exp(\tilde O(n^4))$. 
\end{theorem}
It is straightforward to see that, even in the case $R = 1$ so that $\hat \mu = \mu^\star$, any sampling algorithm requires $\Omega(n)$ queries to $\hat \mu$: this is evident by considering the case where $\mu^\star$ is an unknown singleton in $\{0,1\}^n$. Of course, when $R = 1$, a sample from $\mu^\star$ may be generated with exactly $n$ queries to $\hat \mu$ (simply by querying $\hat \mu = \mu^\star$ coordinate-by-coordinate). Thus, one way to interpret \cref{thm:main-lb} is that \emph{the price of approximation error in $\hat \mu$ (or $\hat V$ in the context of LLM inference) is that an additional $\tilde \Omega(n)$ work must be paid ``per-coordinate''.}

Summarizing, \cref{thm:main-lb} gives a limitation on what is possible in the context of \cref{ques:tilted-nottt}, i.e., sampling from $\mu^\star \propto \piref \cdot V^\star$ \emph{without TTT}. \emph{What can we do with TTT, as captured by \cref{ques:tilted-ttt}?} Here we assume $V^\star \in \MV$ for some class $\MV$, which is equivalent to the knowledge that $\mu^\star \in \MF$ for some known class $\MF$ (see \cref{eq:mf-mv}).\footnote{To be clear, when we say that $\MF$ is known, we mean that for any $\mu \in \MF$ and $x_{1:i} \in \Sigma^i$, the algorithm can determine the value of $\mu(x_{1:i})$ without paying for any queries to $\hat \mu(\cdot)$.} %

Of course, the assumption that $\mu^\star \in \MF$ is only useful if we control the complexity of $\MF$: %
below, we show that as long as $|\MF| \leq \exp(o(n^2))$, then we can sample from $\mu^\star$ using sub-quadratically many queries to $\hat \mu$:
\begin{theorem}[Main upper bound; simplified version of \cref{thm:main-ub-formal}]
  \label{thm:main-ub}
Fix $n \in \BN$, $\delta \in (0,1)$, an alphabet $\Sigma$, and a (known) class of probability distributions $\MF$ on $\Sigma^n$. Also fix any $\mu^\star \in \MF$ (unknown to the algorithm) as well as any mapping $\hat \mu$ satisfying \cref{eq:r-mult-intro} with respect to $\mu^\st$, for some parameter $R \geq 1$. Then there is an algorithm (\cref{alg:pf-learning}) which makes at most $|\Sigma|n \cdot \sqrt{\log|\MF|} \cdot \poly(R, \log \frac{1}{\delta}, \log\log |\MF|)$ queries to $\hat \mu(\cdot)$ and outputs a sample $X \in \Sigma^n$  according to some distribution $\hat \nu$, which satisfies $\tvd{\hat \nu}{\mu^\star} \leq \delta$. 
\end{theorem} 
We remark that there is a gap between \cref{thm:main-lb,thm:main-ub} in terms of the quantitative dependence on $\log |\MF|$: for a class $\MF$ of size $\exp(n^2)$, \cref{thm:main-ub} gives a quadratic upper bound, whereas \cref{thm:main-lb} requires a class $\MF$ of size $\exp(n^4)$ in order to obtain a quadratic lower bound. We view closing this gap as an intruiging direction for future work. In \cref{thm:size-sensitive-lb}, we establish a generalization of \cref{thm:main-lb}, showing that for any ``size parameter'' $S \leq \exp(O(n^4))$, there is a class $\MF$ of size $|\MF| \leq S$ so that any sampling algorithm requires $\tilde \Omega(\sqrt{\log S})$ queries to $\hat \mu$. As discussed below, finding the ``right'' dependence of the query complexity on $|\MF|$ and $n$ remains open.

\paragraph{Takeaways for test-time training.} 
\cref{thm:main-ub} implies that as long as the class $\MF$ is not too large, then knowledge of it may effectively be leveraged to sample more efficiently from $\mu^\star$ than what would be possible without knowledge of $\MF$. This gives an answer to our discussion following \cref{ques:tilted-ttt}, on understanding under what conditions TTT is beneficial. Furthermore, we remark that the algorithm establishing \cref{thm:main-ub} repeatedly generates partial samples from a ``reference'' distribution over $\Delta(\Sigma^n)$, and periodically updates the ``reference'' distribution in response to feedback from $\hat \mu$. This process is analogous to how TTT approaches update the reference model's weights at inference time; see \cref{sec:tech-upperbound} for further details. 

\subsection{Discussion \& Future Work}
\label{sec:discussion}
We view our results as only the beginning for a richer theoretical development of TTT. We elaborate on some directions for follow-up work below.

\paragraph{Tight quantitative bounds.} As discussed above, there is a quantitative gap between \cref{thm:main-lb} and \cref{thm:main-ub} in terms of the dependence on $|\MF|$. In general, it would be interesting to determine the following: 
\begin{problem}
Supposing $R, |\Sigma|, \delta$ are constants, what is the optimal function $f(n, |\MF|)$ so that $f(n, |\MF|)$ queries to an oracle $\hat \mu$ (satisfying \cref{eq:r-mult-intro}) are sufficient to approximately sample from $\mu^\star$, given that $\mu^\star \in \MF$?
\end{problem} We have shown that $f(n, |\MF|) = \tilde O(n \sqrt{\log |\MF|})$ suffices, and that $f(n, |\MF|) = \tilde \Omega(\sqrt{\log |\MF|})$ is necessary. 

\paragraph{Richer complexity measures for $\MF$.} Viewed from the lens of statistical learning theory, our results, which obtain bounds for test-time training based on the \emph{cardinality} of the class of distributions $\MF$, are only scratching the surface. A number of papers have developed combinatorial, geometric, and information-theoretic measures which capture learnability in a range of settings, such as i.i.d.~learning, online learning, and reinforcement learning \cite{rakhlin2014statisticallecture,shalev2014understanding,foster2023foundationsreinforcementlearninginteractive}. Thus, a pertinent question is:
\begin{problem}
What more fine-grained complexity measures of the class $\MF$ can be used to obtain refined bounds for the query complexity of sampling from $\mu^\star \in \MF$? 
\end{problem}

\paragraph{Weaker assumptions on $\hat \mu$.} The recent work of \cite{golowich2026reject} showed that the $\tilde O(n^2)$ query complexity upper bounds for SMC-based sampling approaches extend to the setting where one relaxes the $L_\infty$ type assumption of \cref{eq:r-mult-intro} which requires that $\hat \mu$ and $\mu^\star$ are close on all partial sequences. In particular, it suffices that $\hat \mu$ is close to $\mu^\star$ on average over partial sequences, in an appropriate sense.  \emph{Can the upper bound of \cref{thm:main-ub} be extended to such a setting?}

\paragraph{Computational efficiency and empirical implications.} Finally, we remark that the test-time training algorithm developed for the proof of \cref{thm:main-ub} (namely, \cref{alg:pf-learning,alg:ts-rs}) is query-efficient but computationally inefficient, since it requires explicit enumeration over the class $\MF$. It would be interesting to empirically investigate whether the algorithmic ideas (e.g. the pruning or importance-weighting steps in \cref{alg:ts-rs}) can be translated into practical algorithms. Towards this end, an important stepping stone may be developing algorithms in our framework that are also oracle-efficient, i.e. computationally efficient assuming access to certain optimization oracles for $\MF$.

\section{Technical Overview}\label{sec:tech-overview}
In \cref{sec:tech-upperbound}, we discuss the proof of \cref{thm:main-ub}, and in \cref{sec:tech-lowerbound}, we discuss the proof of \cref{thm:main-lb}. 
\subsection{Upper bound}
\label{sec:tech-upperbound}
\begin{algorithm}[H]
\caption{Basic Sampling}
\label{alg:basic-sampling}
\begin{algorithmic}[1]
\Require Alphabet $\Sigma$, density ratio parameter $R \geq 1$, oracle $\hat \mu: \Sigma^{\leq n} \to \BR_{\geq 0}$, class of densities $\MF$
\State Set $\bar \eta := \frac{1}{|\MF|} \sum_{\mu \in \MF}\mu$.
\State Initialize $x \gets \emptyset$.
\For{$1 \leq i \leq n$}
\State Sample $y \sim \bar \eta(x_i = \cdot \mid x_1, \ldots, x_{i-1})$.
\If{$\frac{\hat \mu(x \circ y)}{\bar \eta(x \circ y)} \geq 2R$}\label{line:hit-xy-intro}
\State \Return $x$.\label{line:return-early-intro}
\Else
\State $x \gets x \circ y$.
\EndIf
\EndFor
\State \Return $x$. 
\end{algorithmic}
\end{algorithm}

We begin by giving a very simple algorithm which proves a quantitatively weaker version of \cref{thm:main-ub}; we will then discuss the full algorithm and its analysis in more detail. 

\paragraph{Warmup: a simple sampling algorithm.}
A natural idea to prove \cref{thm:main-ub} is to use the mean of the densities in $\MF$ (which we denote by $\bar\eta$ in the present discussion) as a ``proxy'' for the true density $\mu^\st$. If it so happens that the density ratio $\frac{\mu^\st(x)}{\bar\eta(x)}=\frac{\hat\mu(x)}{\bar\eta(x)}$\footnote{Recall that \cref{asm:linfty-close-gen} gives that $\hat \mu$ and $\mu^\st$ are equal on length-$n$ sequences.} is bounded (say, by $\poly(R)$) for all $x \in \Sigma^n$, then it is very straightforward to produce samples from $\mu^\st$, using rejection sampling. Namely, we simply sample $x$ from $\bar \eta(x = \cdot)$ and accept $x$ with probability $\frac{\hat\mu(x)}{\bar\eta(x)}$. 

What about the (typical) case where the density ratio is not bounded as above? In this case, a natural strategy is to sample a sequence $x$ from $\bar\eta$, character-by-character, for as many steps as possible until we come to some sequence $x$ and potential next character $y$ where $\frac{\hat \mu(x\circ y)}{\bar \eta(x\circ y)}$ becomes too large, say $\geq 2 R$: this procedure is formalized in \cref{alg:basic-sampling}. Suppose that we get to such an $x \circ y$ (namely, Line \ref{line:hit-xy-intro}). At this point, not all is lost: we can restrict the class $\MF$ to densities that ``are consistent with'' $\hat \mu$ at the input $x$, i.e., suppose we define $\MF' := \{ \mu \in \MF \mid \mu(x\circ y) \geq \hat\mu(x\circ y) / R \}$. Certainly $\mu^\st \in \MF'$ by \cref{asm:linfty-close-gen}. Moreover, for any $\mu \in \MF'$, we have
\begin{align}
\frac{\mu(x\circ y)}{\bar \eta(x\circ y)} \geq \frac{\hat \mu(x\circ y)}{R \cdot \bar \eta(x\circ y)} \geq 2\nonumber,
\end{align}
where the final inequality follows by the assumption that $\frac{\hat \mu(x\circ y)}{\bar \eta(x\circ y)} \geq 2R$. 
But by Markov's inequality, the number of such densities $\mu \in \MF$ is bounded above by $|\MF|/2$. Thus, \emph{we have halved the number of densities under consideration!} We can now repeat \cref{alg:basic-sampling}, except using the class $\MF'$ as input instead of $\MF$. %
After at most $\log |\MF|$ repetitions of this procedure, we will either have narrowed $\MF$ down to $\mu^\st$ (in which case outputting a sample from $\mu^\st$ is trivial), or else we will have converged on some class $\MF''$ which has the following property: for some small value of $\ep > 0$, the probability that \cref{alg:basic-sampling} given the class $\MF''$ as input returns  some $x$ of length less than $n$, is at most $\ep$. Note that, in order to check if the aforementioned probability is bounded above by $\ep$, we need to run \cref{alg:basic-sampling} for $\tilde O(1/\ep)$ iterations given the class $\MF''$ as input. 

Let $\bar\eta''$ denote the mean of the densities in $\MF''$, and let $\MT$ denote the set of $x \in \Sigma^{<n}$ of length less than $n$ which \cref{alg:basic-sampling} can possibly output when its input is $\MF''$. Formally, $\MT$ is the set of $x \in \Sigma^{<n}$ so that $\frac{\hat \mu(x \circ y)}{\bar \eta''(x \circ y)} \geq 2R$ for some $y \in \Sigma$ but no nodes $x' \prec x$ satisfy this property. The choice of $\MF''$ gives that $\sum_{x \in \MT} \bar \eta(x) \leq \ep$. Furthermore, since $\frac{\hat\mu(x)}{\bar \eta(x)} \leq 2R$ for all $x \in \MT$, we have
\begin{align}
\sum_{x \in \MT} \mu^\st(x) \leq \sum_{x \in \MT} R \cdot \hat \mu(x) \leq 2R^2 \sum_{x \in \MT} \bar \eta(x) \leq 2R^2 \cdot \ep\nonumber.
\end{align}
It then follows from a straightforward analysis of rejection sampling that by using rejection sampling on the output of \cref{alg:basic-sampling} (conditioned on its output having length $n$, i.e., not returning on Line \ref{line:return-early-intro}), we obtain a sample from a distribution $\hat \nu$ satisfying $\tvd{\mu^\st}{\hat \nu} \leq 2R^2 \cdot \ep$. Furthermore, the number of queries to $\hat \nu(\cdot)$, over all rounds, is bounded above by $\tilde O( n\log |\MF| \cdot \frac{1}{\ep})$: there are at most $\log |\MF|$ rounds on which we halve the size of $\MF$, and each round takes up to $\tilde O(n/\ep)$ queries to run \cref{alg:basic-sampling} $\tilde O(1/\ep)$ times.\footnote{Technically, we need an additional $\tilde O(Rn)$ queries to perform rejection sampling at the end, though after rescaling $\ep$ down by a factor of $R^2$ this term is dominated by the one we have written above.}

This guarantee falls short of the one stated in \cref{thm:main-ub} in two respects:
\begin{enumerate}
  \item First, in \cref{thm:main-ub}, we only need $\poly\log(\frac{1}{\ep})$ iterations to achieve total variation distance at most $\ep$, as opposed to the (exponentially worse) scaling of $1/\ep$ for the above procedure.
  \item Second, the above guarantee scales linearly with $\log |\MF|$, as opposed to the scaling of $\sqrt{\log |\MF|}$ in \cref{thm:main-ub}. 
\end{enumerate}

\paragraph{Towards a more efficient sampling procedure.} To address both of the above shortcomings, the following high-level idea will be useful: if we reach some $x$ with large ratio as in Line \ref{line:hit-xy-intro} of \cref{alg:basic-sampling}, there is no need to ``start over'' again, sampling a new sequence from scratch. Instead, we can try to continue sampling starting from the sequence $x$ at which we previously ``got stuck'', but using the mean density of the new (smaller) class $\MF'$ to sample continuations. There are a number of possible complications with this procedure: 
\begin{enumerate}
\item[\textbf{(C1)}] First, since in general it will be the case that $\bar \eta(x) \neq \mu^\st(x)$, there will be some bias in the probability of choosing $x$ in the first place. While this bias can be corrected using rejection sampling, we need to take care to ensure that such a bias does not accumulate multiplicatively over multiple rounds of this procedure. 
\item [\textbf{(C2)}] A second and more subtle complication is that the basic sampling procedure in \cref{alg:basic-sampling} might sometimes output long sequences $x$ (e.g., of length $|x| = n$) for which $\hat \mu(x)$ is incredibly small (e.g., suppose $\hat \mu(x) < 2^{-n} \cdot \bar \eta(x)$), so that the subsequent rejection sampling step will nearly always reject such sequences. If such $x$ are output by \cref{alg:basic-sampling} with reasonably high probability (e.g., some constant), then it is possible that between each pair of consecutive steps of the overall algorithm where we halve the size of the class of densities, we will need to ``pay'' $O(n)$ oracle calls to $\hat \mu(\cdot)$  in the process of sampling these long sequences $x$ where $\hat \mu(x)$ is small. Overall, this amounts to $O(n \log |\MF|)$ oracle calls, which falls short of the desired $O(n \sqrt{\log |\MF|})$ guarantee in \cref{thm:main-ub}. To get around this complication, a key idea will be to employ \emph{truncation}: if we are in the process of sampling a sequence $x \in \Sigma^{<n}$ and we detect that $\hat \mu(x) < \bar \eta(x)$, then we can simply prematurely reject $x$ with high probability, since it would be very unlikely for $x$ to be ultimately accepted in the subsequent rejection sampling step. 
\end{enumerate}
Below, we describe in further detail how we overcome the above complications. 

\paragraph{Algorithm overview.}
We are now ready to present \cref{alg:pf-learning}, which establishes \cref{thm:main-ub}. The outermost loop of \cref{alg:pf-learning} iterates over multiple steps $k$: each step implements a rejection sampling step (\cref{line:k-rejection-sampling}) on a uniformly random element of a set $\Sterm$ constructed in the course of that step. The bulk of the work goes into producing this set $\Sterm$, which we proceed to describe in detail.

In each step $k$, the algorithm operates over multiple iterations $j$, maintaining a subclass $\MF^j \subset \MF$ at each iteration. The class $\MF^j$ contains densities $\mu$ which are ``consistent with'' the algorithm's previous queries to the oracle $\hat \mu$.  Additionally, at each iteration $j$, the algorithm maintains a set $\MS^{j-1}$ consisting of sequences $z \in \Sigma^{\leq n}$: these sequences correspond to ``partial samples'' which the algorithm is considering, roughly speaking, as ``candidates'' for inclusion in $\Sterm$. Over the course of the algorithm, sequences $z$ in $\MS^{j-1}$ will be extended by adding characters; these sequences may further be removed from $\MS^{j-1}$ (if, e.g., it turns out that many extensions $z' \succ z$ have very small values of $\hat \mu(z')$), or duplicated (if, e.g., it turns out that many extensions $z' \succ z$ have very large values of $\hat \mu(z')$). 

Initially, $\MF^1$ is set to be the set of densities consistent with $\hat \mu$ at all single-character sequences, i.e., ``$y$'' for $y \in \Sigma$ (Line \ref{line:f1-define}). At each iteration $j$, the algorithm chooses an element $z_j$ of $\MS^{j-1}$ which is incomplete (i.e., has length less than $n$; Line \ref{line:choose-zj}). Next, it calls the sub-procedure \cref{alg:ts-rs} with the node $z_j$ and the density class $\MF^j$ as input.  \cref{alg:ts-rs} returns a set $\MT_j \subset \Sigma^{\leq n}$ and a subset $\MG_j \subset \MF^j$; we add all of $\MT_j$ to $\MS_j$ (and remove $z_j$ from $\MS_j$) and restrict $\MF^j$ to only contain densities in $\MG_j$, and then repeat the procedure.

\cref{alg:ts-rs} performs the work  of extending the sequence $z_j$ (and possibly pruning or duplicating it). Below, we denote the class passed to \cref{alg:ts-rs} as input by $\MG$, and the node passed to it as input by $w$; thus, we will have $\MG = \MF^j$ and $w = z_j$ in the context of the above discussion. 

\cref{alg:ts-rs} executes the following procedure over multiple trials $i$ (the number of trials, denoted by $M$ in \cref{alg:ts-rs}, will be taken to be $\poly(R)$). In each trial $i$, \cref{alg:ts-rs} defines a density $\bar\eta_{i,0}$ to be the mean of the densities in $\MG$ (Line \ref{line:initialize-ts-rs}) and lets $w_{i,0}$ take the value of the input sequence $z_j$ from above. It then repeatedly samples a character $y \in \Sigma$ to append to the input sequence, according to the density $\bar \eta_{i,0}$ (Lines \ref{line:extend-1} and \ref{line:extend-2}); the sequence which is being extended in this manner is denoed by $w_{i,t}$ in \cref{alg:ts-rs} (initially, $t = 1$). This procedure is repeated until one of the following 3 events occur:
\begin{itemize}
\item The first event (Line~\ref{line:if-big}) is that for some character $y \in \Sigma$, the mass of $\hat\mu(w_{i,t }\circ y)$, normalized by $\hat \mu(w_{i,t-1})$, is much larger than $\bar\eta_{i,t-1}(w_{i,t} \mid w_{i,t-1})$ (i.e., by a factor of at least $M$). 
This case is analogous to the ``ratio becomes too large'' case in Line \ref{line:hit-xy-intro} of \cref{alg:basic-sampling}. In this case, the current sequence $w_{i,t}$ is added to the multiset $\MT$ which will be output by \cref{alg:ts-rs} with an appropriate probability. Furthermore, the class $\Gret$, which will be returned by \cref{alg:ts-rs} (along with $\MT$), is restricted to contain only densities $\eta$ satisfying $\hat \mu(w_{i,t} \circ y)/\eta(w_{i,t} \circ y) \leq R$, thus eliminating densities inconsistent with $\hat\mu$ at $w_{i,t} \circ y$. As we will show, this step ensures that the returned class $\Gret$ contains at most half the number of densities as the input class $\MG$. The trial $i$ then terminates.

\item The second event (Line~\ref{line:else-lengthn}) is that $w_{i,t}$ reaches full length $|w_{i,t}| = n$, i.e., we have sampled a complete sequence. In this case, $w_{i,t}$ is added to $\MT$ with appropriate probability (Lines~\ref{line:xi-else}--\ref{line:add-to-T-end}), and the trial $i$ terminates. 

\item The third event (Line~\ref{line:if-small}) is that the mass of $\hat \mu(w_{i,t})$, normalized by $\hat \mu(w_{i,t-1})$ is much \emph{smaller} than $\bar \eta_{i,t-1}(w_{i,t} \mid w_{i,t-1})$ (i.e., by a factor of at least $M$). In this case, with an appropriate probability, the trial $i$ simply ends (Line \ref{line:break-m-trial-end}). Otherwise, we increment $t$ and continue on sampling from $w_{i,t}$.  

We refer to this as the \emph{truncation} step, addressing complication \textbf{(C2)} above: rather than continuing to extend a sequence whose density under $\hat\mu$ is very small (and which would therefore be nearly always rejected in a final rejection sampling step), the algorithm truncates the trial early with high probability.

\end{itemize}

\paragraph{Unbiasedness: addressing complication \textbf{(C1)}.}
The acceptance probabilities in \cref{alg:ts-rs} are calibrated to ensure the following key \emph{unbiasedness property}: for any node $w'$ in a certain ``final'' set of nodes descending from the input node $w$, the expected number of occurrences of $w'$ in the output multiset $\MT$ is exactly $\hat\mu(w')/\hat\mu(w)$ (see \cref{lem:final-probability}). %
This unbiasedness property addresses complication \textbf{(C1)}: even though the proposal distribution $\bar\eta$ may differ substantially from $\mu^\star$, the importance-weighting corrections prevent any accumulation of bias.

In order to show that the algorithm samples from the right distribution $\mu^\st$, by the final rejection sampling step on Line \ref{line:k-rejection-sampling} of \cref{alg:pf-learning}, the key point is to control the extent to which $|\Sterm|$ can exceed $2RN \asymp \sqrt{\log |\MF|} \cdot \poly(R)$. To do this, we reason as follows: a consequence of unbiasedness as discussed above (namely, \cref{lem:final-probability})  is that $M_j := \sum_{x \in \MS^j} \mu^\star(x)/\hat\mu(x)$ is a {martingale} over the iteration index $j$ (\cref{lem:mi-mtgl}). Since $\mu^\star(x) = \hat\mu(x)$ for complete sequences $x \in \Sigma^n$, at termination ($j = J$) this martingale equals $|\Sterm|$. The desired control on the truncation in the final rejection sampling step then follows from a standard martingale concentration argument, together with the corresponding upper-tail overshoot bound (\cref{lem:mtgl-concentration}). We remark that this martingale concentration is crucial for improving the naive upper bound which was linear in $\log |\MF|$ to one which grows with $\sqrt{\log |\MF|}$: the number of iterations $J$ in \cref{alg:pf-learning} is upper bounded by $\log_2 |\MF|$ (\cref{lem:iteration-bound}) since each iteration halves the size of $\MF$, and the error term in martingale concentration grows only as $\sqrt{J} \asymp \sqrt{\log |\MF|}$.

\paragraph{Bounding the query cost of \cref{alg:pf-learning}.}
A key step in bounding the query cost of \cref{alg:pf-learning} is \cref{lem:query-cost}, which bounds the expected number of oracle queries made by the invocation of \cref{alg:ts-rs} on iteration $j$ in \cref{alg:pf-learning} by the quantity
\begin{align}
\E \left[\left( \sum_{w' \in \MT_j} \frac{\mu^\st(w')}{\hat \mu(w')} \cdot |w'| \right) - \frac{\mu^\st(w)}{\hat \mu(w)} \cdot |w|\right] \label{eq:query-cost-key-step}.
\end{align}
Recall that $\MT_j$ denotes the output set of \cref{alg:ts-rs};  further, the above expectation is taken over the execution of \cref{alg:ts-rs}, for the fixed input node $w = z_j$. A key step in deriving this bound is the use of the \emph{truncation} step discussed above: roughly speaking, the truncation step ensures that we only query nodes $w'$ which have a decent probability of ending up in the set $\MT_j$ (thus addressing complication \textbf{(C2)}). 

Thus, using Azuma's inequality, in order to obtain a high-probability bound on the \emph{overall} number of oracle calls, it suffices to bound the sum of the quantity \cref{eq:query-cost-key-step} over all iterations $j$ of \cref{alg:pf-learning}. This argument is carried out in \cref{lem:oracle-calls}: the key step is the observation that the quantity inside the expectation in \cref{eq:query-cost-key-step} is the difference $W_j - W_{j-1}$, where $W_j := \sum_{x \in \MS^j} \frac{\mu^\star(x)}{\hat\mu(x)} \cdot |x|$ is a ``weighted'' variant of the martingale $M_j$, where each sequence in $\MS^j$ is weighted by its length. Namely, the positive term in \cref{eq:query-cost-key-step} corresponds to the contribution of the new sequences in $\MT_j$, while the negative term corresponds to the effect of removing the node $z_j$ from $\MS_j$. Finally, we note that our high-probability bound on $|\Sterm|$ from \cref{lem:mtgl-concentration} allows us to bound $W_J$ by $\sqrt{\log |\MF|} \cdot n \cdot \poly(R, \log \frac{1}{\delta}, \log\log |\MF|)$, and a further application of Azuma's inequality then gives the desired bound on the sum of the quantities in \cref{eq:query-cost-key-step} over all iterations $j$.

\subsection{Lower bound}
\label{sec:tech-lowerbound}
To prove \cref{thm:main-lb}, we consider the case that $\Sigma = \{0,1\}$. We will construct a family of distributions $\MF$ over $\{0,1\}^n$, which has the following property: each $\mu^\st \in \MF$ is associated with some value $\vinit \in \{0,1\}$, so that $\mu^\st(x_1 = \vinit) \approx \frac{1+\gamma}{2+\gamma} > \frac 12 $, for some constant $\gamma \in (0,1)$. In other words, sampling from $\mu^\st$ to accuracy $\ll \gamma$ requires one to learn $\vinit$: formally, for any distribution $\hat \nu$ we have $\mu^\st(x_1 = \vinit) \gtrsim \frac{1+\gamma}{2+\gamma} - \tvd{\mu^\st}{\hat \nu}$, so if $\tvd{\mu^\st}{\hat \nu} \ll \gamma$, then the first bit of a sample from $\hat \nu$ is equal to $\vinit$ with probability at least $\frac{1}{2} + \Omega(\gamma)$. 

We will show that, for an appropriate choice of $\MF$, 
with high probability over a random draw of $\mu^\st \sim \MF$ and an appropriately chosen $\hat \mu$ satisfying \cref{asm:linfty-close-gen} with respect to $\mu^\st$, any algorithm $\Alg$ making $o(n^2/\log^2 n)$ queries to $\hat \mu(\cdot)$ cannot output $\vinit$ with probability $\frac{1}{2} + \Omega(1)$, which will yield the desired result as $\gamma$ was chosen to be constant. Below we give an overview of the construction of $\MF$, and then describe why subquadratically many queries are insufficient to learn $\vinit$ for a random draw of $\mu^\st \sim \MF$.

\paragraph{Construction of $\MF$: overview.} Each element of $\MF$ is parametrized by a tuple $\tau = (\vinit, \phi, \psi)$, where $\vinit$ is as described above, and $\phi, \psi$ are parameters which will be described below. Associated to each tuple $\tau$, we define a density $\mu_\tau \in \Delta(\{0,1\}^n)$, as well as a mapping $\hat \mu_\tau : \{0,1\}^{\leq n} \to \BR_{\geq 0}$. To construct $\hat \mu_\tau$ and $\mu_\tau$, we will interpret $\{0,1\}^n$ as the complete binary tree, so that its edges are in one-to-one correspondence with $\{0,1\}^{\leq n}$. We will define edge weights for each edge of this tree, formalized by mappings $ \hat\omega_{\tau} : \{0,1\}^{\leq n} \to \BR_{\geq 0}$. We then define $\hat \mu_\tau(x)$ to be proportional to the product of the edge weights on the root-to-leaf path corresponding to $x$. Formally, for $x \in \{0,1\}^j$ for some $j \in [n]$, we define
\begin{align}
\hat \mu_\tau( x) = \frac{1}{\hat Z_j}\prod_{i=1}^{j} \hat \omega_\tau(x_{1:i}), \qquad \hat Z_j = \sum_{x' \in \{0,1\}^j} \prod_{i=1}^j \hat \omega_\tau(x'_{1:i})\label{eq:hat-mu-Z-intro}.
\end{align}
Our choice of $\hat \omega_\tau$ will ensure that the normalizing constants $\hat Z_j$ do not depend on $\tau$ (see \cref{lem:zhat-invariant}). 
Finally, we take $\mu_\tau(x) = \hat \mu_\tau(x)$ for $x \in \{0,1\}^n$; the choice of $\hat Z_n$ above ensures that $\mu_\tau(x)$ is indeed a distribution on $\{0,1\}^n$. For $j \in [n]$ and $x \in \{0,1\}^j$, with slight abuse of notation we let $\mu_\tau(x)$ denote the marginal distribution of $\mu_\tau$ on $x$, i.e., $\mu_\tau(x) = \sum_{x' \in \{0,1\}^{n-j}} \mu_\tau(x \circ x')$; then we have $\mu_\tau(x) = \mu_\tau(x \circ 0) + \mu_\tau(x \circ 1)$ for all $x \in \{0,1\}^{<n}$. Since the sum of the edge weights on the edges leaving different nodes may differ, it can be the case that $\mu_\tau(x) \neq \hat \mu_\tau(x)$ and $\hat \mu_\tau(x) \neq \hat \mu_\tau(x \circ 0) + \hat \mu_\tau(x \circ 1)$ for $x \in \{0,1\}^{<n}$.

The edge weights $\hat \omega_\tau$ at step $1$ are given as follows: for all $\tau$, $\hat \omega_\tau(1) = \hat \omega_\tau(0) = 1 + \gamma$. %
Recall, however, our desideratum above that $\mu_\tau(x_1 = \vinit) \geq \frac{1}{2} + \Omega(\gamma)$ for each $\tau$. If the edge weights $\hat \omega_\tau(x)$ are ``evenly balanced'' for $x$ whose first bit is $0$ and for $x$ whose first bit is $1$, then it will hold that $\mu_\tau(x_1 = 0) \approx \mu_\tau(x_1 = 1)$, which violates this lower bound on $\mu_\tau(x_1 = \vinit)$. Thus, we need the edge weights $\hat \omega_\tau(x)$ for $x$ whose first bit is $1-\vinit$ to be noticeably smaller than those for $x$ whose first bit is $\vinit$. 

On the other hand, this ``bias'' in the edge weights needs to be undetectable to any algorithm $\Alg$ which can query sub-quadratically many values of $\hat \omega_\tau(\cdot)$. Indeed, for any $x \in \{0,1\}^{j}$ ($j \in [n]$) $\Alg$ can determine the value of $\hat \omega_\tau(x)$ by querying $\hat \mu_\tau(x)$ and $\hat \mu_\tau(x_{1:j-1})$. How do we achieve this goal? At a high level, we will divide the set of coordinates $\{ 2, 3, \ldots, n\}$ into groups of size $k \approx 2 \log n$: in each group, we will choose a fraction $\ep = \Theta(\log(n)/n)$ of the $2^k$ values $\{0,1\}^k$ corresponding to the coordinates in that group for which we make the edge weights $\hat \omega_\tau(\cdot)$ \emph{smaller} by roughly a $(1+\gamma)$-factor in the case that $x_1 = 1-\vinit$, as compared to the case that $x_1 = \vinit$. The precise choice of these $\ep \cdot 2^k$ values is determined by $\tau$, and will be uniformly random for a uniform choice of $\tau$. Since $\ep = \Theta(\log(n)/n)$ and $2^k \approx n^2$, it is straightforward to see that any $\Alg$ making $o(n^2 /\log^2 n)$ queries to $\hat \omega_\tau(\cdot)$ cannot tell whether the edge weights are decreased in this manner (which translates into showing that $\Alg$ cannot determine the value of $\vinit$). Moreover, since there are $n/\log(n)$ groups and for each group a fraction of $\ep \sim \log(n)/n$ values of $x$ have their edge weight decreased, with high probability over a full sample $x \sim \mu_\tau$ ($x \in \{0,1\}^n$) there will be \emph{some} group whose edge weight is decreased.

The astute reader will notice a potential flaw in the above line of reasoning: while, as we noted above, one can simulate a single query to $\hat \omega_\tau(\cdot)$ using two queries to $\hat \mu_\tau(\cdot)$, in fact this is not necessarily the best that $\Alg$ can possibly do. Indeed, querying $\hat \mu_\tau(x)$ for $x \in \{0,1\}^j$ divulges information about \emph{up to the $j$ values} $\hat \omega_\tau(x_1), \hat \omega_\tau(x_{1:2}), \ldots, \hat \omega_\tau(x)$. To avoid the possibility that $\Alg$ can ``batch query'' $\hat \omega_\tau(\cdot)$ in this manner, we will make it so that the vast majority of $x \in \{0,1\}^j$ in fact have $0$ edge weight, i.e., $\hat \omega_\tau(x) = 0$. Thus, to find $x$ with nonzero edge weight, $\Alg$ needs to ``walk'' on a path starting from the root of the tree, and thus it can only learn information about a single value of $\hat \omega_\tau(\cdot)$ with each query to $\hat \mu(\cdot)$. 
Below, we describe these ideas in more detail.

\paragraph{Construction of the edge weights $\hat \omega_\tau$.} We break the coordinates $\{2, 3, \ldots, n\}$ into $r$ groups of size $2k$ (for $k \approx 2 \log n$), indexed by a value $a \in \{0,1, \ldots, r-1 \}$. For each $a$, and $x \in \{0,1\}^{2(a+1)k+1}$, we define $\hat \omega_\tau(x_{1:2ak+2}), \hat \omega_\tau(x_{1:2ak+3}), \ldots, \hat \omega_\tau(x)$, in a way that depends on the tuple $\tau = (\vinit, \phi, \psi)$, as follows. Let us write 
\begin{align}
\base_a(x) = x_{1:2ak+1}, \quad \query_a(x) = x_{2ak+2:2ak+k+1}, \quad \key_a(x) = x_{2ak+k+2:2(a+1)k+1}, \nonumber
\end{align}
so that $x = \base_a(x) \circ \query_a(x) \circ \key_a(x)$. 
Then: 
\begin{itemize}
\item We set $\hat \omega_\tau(\base_a(x) \circ \query_a(x)_{1:i}) = 1$ for each $i \in [k]$. 
\item To define the remaining values of $\hat \omega_\tau$, we need to specify the data contained in $\phi, \psi$. First, $\psi$ is a mapping $\psi : \{0,1\}^{\leq n} \to (\{0,1\}^k)^{\{0,1\}^k}$ (in words, the range of $\psi$ consists of vectors with entries in $\{0,1\}^k$ and whose entries are also indexed by $\{0,1\}^k$). Moreover, $\phi$ is a mapping which takes each $y \in \{0,1\}^n$ to a partition $\{0,1\}^k = A(y) \sqcup B(y) \sqcup C(y)$, with $|A(y)| = 2^{k} - \ep \cdot 2^{k-1}, |B(y)| = \ep \cdot 2^{k-1}, |C(y)| = 2^{k-1}$. We denote the set of partitions by $\Vvalid_{k,\ep}$. We will be interested in the value of $\phi, \psi$ evaluated at $\base_a(x)$ for $x \in \{0,1\}^{2(a+1)k+1}$ as above. 
\item For $1 \leq b \leq k$, we set $\hat \omega_\tau(\base_a(x) \circ \query_a(x) \circ \key_a(x)_{1:b}) = 0$ if the $b$ bits $\key_a(x)_{1:b}$ do not agree with the first $b$ bits of $\psi(\base_a(x))_{\query_a(x)}$; this requirement accomplishes the goal mentioned above of forcing the algorithm to ``walk from the root downwards'' to uncover new values of $x'$ for which $\hat \omega_\tau(x') \neq 0$. Indeed, it is very unlikely that the algorithm will correctly guess the value of $\psi(\base_a(x))_{\query_a(x)} \in \{0,1\}^k$, without first querying (most of) the values of $\hat \mu_\tau(\base_a(x) \circ \query_a(x) \circ \key_a(x)_{1:b})$ (for $b \in [k]$), as $\psi$ is unknown (and will be chosen uniformly at random). %

\item If further $b < k$ and the $b$ bits in the previous bullet point do agree, then we set $\hat\omega_\tau(\base_a(x) \circ \query_a(x) \circ \key_a(x)_{1:b}) = 1$.

\item Finally, in the case that $b = k$, we choose $\hat \omega_\tau(x) \in \{1,1+\gamma\}$ as follows: first, if $x_1 = \vinit$, then we set
\begin{align}
\hat \omega_\tau(x) = \begin{cases}
1+\gamma & \text{if $\query_a(x) \in A(x) \cup B(x)$} \\
1 & \text{if $\query_a(x) \in C(x)$}.
\end{cases}\label{eq:intro-vinit}
\end{align}
If $x_1 = 1-\vinit$, then \emph{typically} we will set
\begin{align}
\hat \omega_\tau(x) = \begin{cases}1+\gamma & \text{if $\query_a(x) \in A(x)$} \\
1 & \text{if $\query_a(x) \in B(x) \cup C(x)$}.
\end{cases}\label{eq:intro-not-vinit}
\end{align}
In other words, for those $x$ with $\query_a(x) \in B(x)$, the value of $\hat \omega_\tau(x)$ is smaller by a factor of $1+\gamma$ in the event that $x_1 = 1-\vinit$ as compared to the event that $x_1 = \vinit$. We write ``typically'' above due to the following exception: if there is some previous block $a' < a$ for which $\query_{a'}(x) \in B(\base_{a'}(x))$, then we set $\hat \omega_\tau(x)$ as in \cref{eq:intro-vinit}. This exception is needed to ensure that we do not ``overshoot'' and decrease the probability of $x$ by more than a $1+\gamma$ factor. 
\end{itemize}

\paragraph{Establishing the lower bound.} To establish that the above construction proves \cref{thm:main-lb}, we first need to verify that the ratio $\hat \mu_\tau(x) / \mu_\tau(x)$ is a constant for all $x \in \{0,1\}^{\leq n}$. We will in fact be able to show that this ratio lies in $[(1+\gamma)^{-2}, (1+\gamma)^2]$ (see \cref{lem:mu-muhat-ratio}). Roughly speaking, this holds because the ``exception'' described following \cref{eq:intro-not-vinit} ensures that we only modify the ratio by at most a factor of $\approx 1+\gamma$ in the aggregate over all levels of the tree.

Next, we need to verify that $\mu_\tau(x_1 = \vinit) \approx \frac{1+\gamma}{2+\gamma}$ (which ensures that accurately sampling from $\mu_\tau$, namely to error $\ll \gamma$, enables predicting the value of $\vinit$ correctly with probability $\frac 12 + \Omega(\gamma)$, when $\tau$ is drawn uniformly at random).  This holds due to the following reasoning: for any fixed $\tau$, for $x \sim \mu_\tau$, conditioned on $x_1 = 1-\vinit$, for each $0 \leq a < r$, the probability that $\query_a(x) \in B(\base_a(x))$ is roughly $\ep = \Theta(\log(n) \cdot \log(1/\gamma)/n)$. Since there are $r$ values of $a$ and $r = \Theta(n / \log(n))$, the probability that there is some $a$ with $\query_a(x) \in B(\base_a(x))$ is roughly $1 - (1-\ep)^r \approx 1 - e^{-\ep r} \approx 1 - c\gamma$, for a small constant $c$. This means that, for the vast majority of $x$ with $x_1 = 1-\vinit$, there will be some block $a$ where the weight $\hat \omega_\tau(x_{1:2(a+1)k+1})$ is decreased by a factor of $1+\gamma$ (per \cref{eq:intro-not-vinit}). Since $\hat \mu(x)$ is proportional to the product of $\hat \omega_\tau(x_{1:i})$ for $i = 1, 2, \ldots, n$, this leads to the desired claim. 

The above reasoning is carried out formally in \cref{lem:mu-muideal}: formally, in \cref{sec:lower-bound} we define an alternative measure $\muideal$ which has the property that $\muideal(x_1 = \vinit) = \frac{1+\gamma}{2+\gamma}$. Then \cref{lem:mu-muideal} shows that $\muideal$ and $\mu_\tau$ are $e^{-\ep r}$-close in total variation distance. 

The final step in the proof, which is the most technical, is to show that no algorithm can output $\vinit$ with probability $\frac 12 + \Omega(\gamma)$ after making only $o(n^2 / \log^2 n)$ queries to $\hat \mu_\tau(\cdot)$, over a uniformly random draw of $\tau$. The main technical lemma to show this is given in \cref{lem:close-consistent-query}, which is stated informally below:
\begin{lemma}[Informal version of \cref{lem:close-consistent-query}]
  \label{lem:close-consistent-intro}
Fix any sequence of past queries of $\Alg$, denoted $X_1, \ldots, X_{j-1} \in \{0,1\}^{\leq n}$, as well as some response values for these queries, $\mu_1', \ldots, \mu_{j-1}' \geq 0$. For $b \in \{0,1\}$, let $\BP_b$ denote the distribution of $\tau$ (which is uniform) conditioned on $\vinit = b$. Let $\ME$ denote the event that $\hat \mu_\tau(X_\ell) = \hat \mu_\ell'$ for each $\ell \in [j-1]$. Then, assuming $\BP_0(\ME), \BP_1(\ME) > 0$, for any ``current'' query $X_j$, we have
\begin{align}
\kld{\BP_0(\hat \mu_\tau(X_j) = \cdot \mid \ME)}{\BP_1(\hat \mu_\tau(X_j) = \cdot \mid \ME)}\leq O(\ep^2)\nonumber.
\end{align}
\end{lemma}
Roughly speaking, \cref{lem:close-consistent-intro} gives that the conditional distributions of the output of the $j$th query to $\hat \mu_\tau$, \emph{for an arbitrary sequence of previous queries to $\hat \mu_\tau$ and their answers}, is close in the cases where we condition $\vinit$ to equal $1$ or $0$. While the full proof of the lemma involves several complexities due to the necessity of considering which part of each ``block'' $X_j$ is querying, as well as conditioning on past queries, the main idea boils down to the fact that the distributions $\BP_0, \BP_1$ differ only in their values of $\hat \omega_\tau(x)\in \{1,1+\gamma\}$ for $x$ in the case of \cref{eq:intro-vinit,eq:intro-not-vinit} above. Namely, the probability that $\hat \omega_\tau(x) = 1+\gamma$ is larger by roughly $|B(x)|/2^k = \ep$ in the event that $x_1 = \vinit$ (as compared to the event that $x_1 = 1-\vinit$). Here we use that  the random choice of $\tau$ induces a uniformly random value for the set $B(x) \subset \{0,1\}^k$. Further, the probability that $\hat \omega_\tau(x) = 1+\gamma$ is bounded in $[1/3, 2/3]$. Then the claim of \cref{lem:close-consistent-intro} follows because the KL divergence between Bernoulli random variables whose means are separated by $\gamma$ and bounded in $[1/3, 2/3]$ is $O(\ep^2)$. 

Having established \cref{lem:close-consistent-intro}, by applying the chain rule for KL divergence we obtain that, comparing the event when $\vinit = 0$ and the event when $\vinit = 1$, the distribution of the transcript of $\Alg$'s computation can differ in these two cases by at most $O(\ep^2 \cdot n^2 / \log^2 n) \ll 1$ in KL divergence. This establishes the desired claim.

\paragraph{On the role of $\psi$.}
The above argument has neglected to mention one key aspect of the proof, which motivates the inclusion of $\psi$ in the tuple $\tau$ as well as the ``key'' blocks $\key_a(x)$ defined above. In particular, \cref{lem:close-consistent-intro} implicitly assumes that each query that $\Alg$ makes reveals information about $\hat\omega_\tau(x)$ for a \emph{single} value of $x \in \{0,1\}^{\leq n}$. However, since $\hat \mu_\tau(x)$ is proportional to the \emph{product} $\prod_{j=1}^{|x|} \hat \omega_\tau(x_{1:j})$ (see \cref{eq:hat-mu-Z-intro}), $\Alg$ can in principle learn information about up to $n$ values of $\hat \omega_\tau(\cdot)$ with a single query to $\hat \mu_\tau(\cdot)$. To overcome this obstacle, we introduce  a modified algorithm $\tAlg$, which simulates $\Alg$ in the following manner: if, at any point, $\Alg$ ``skips ahead'' in the sense that it queries $\hat \mu_\tau(x)$ for some $x$ for which it has never queried a ``previous block'' (i.e., there is some $a < \lfloor (|x| - 1) / (2k) \rfloor$ for which $\Alg$ has never queried $\hat \mu_\tau(\cdot)$ at $x_{1:2ak+1}, x_{1:2ak+2}, \ldots, x_{1:2ak+2k}$), then $\tAlg$ does not actually query $\hat \mu_\tau(\cdot)$ at such $x$ \emph{and proceeds as if the answer to the query is $0$}. 

Roughly speaking, \cref{lem:close-consistent-intro} above applies to $\tAlg$ as stated, since the definition of $\tAlg$ ensures that it can only ``learn information about'' a single block $0 \leq a < r$ with each query. 
The main remaining step in the proof of \cref{thm:main-lb} is to show that the execution of $\tAlg$ is identical to that of $\Alg$, with high probability. This task is accomplished in the proof of \cref{clm:hatb-hatb0}. The proof uses the fact that for a uniform choice of $\tau$, if $\Alg$ does not query the value of $\hat \mu_\tau(\cdot)$ at any of $x_{1:2ak+1}, x_{1:2ak+2}, \ldots, x_{1:2ak+2k}$ prior to querying $x$, then $\psi(\base_a(x))_{\query_a(x)} \neq \key_a(x)$ with high probability (namely, $1-2^{-k}$); this holds because conditioned on the history of execution, $\psi(\base_a(x))_{\query_a(x)}$ is uniformly random in $\{0,1\}^k$, even conditioned on $\key_a(x)$.

\paragraph{Making the class $\MF$ small via $q$-wise independence.} 
In the above construction, if we naively choose $\phi : \{0,1\}^{\leq n} \to \Vvalid_{k,\ep}$ and $\psi : \{0,1\}^{\leq n} \to (\{0,1\}^k)^{\{0,1\}^k}$ to be uniformly random functions on their respective domains, then the number of possible values for the tuple $(\vinit, \phi, \psi)$ grows \emph{doubly exponentially} with $n$, which falls far short of the claim in \cref{thm:main-lb} that the number of possibilities for $\mu_\tau$  grows only as $\exp(\tilde O(n^4))$. To correct for this, we in fact choose $\phi$ and $\psi$ to be drawn from \emph{$q$-wise independent families} of functions, where $q = O(n^2 / \log^2 n)$ is an upper bound on the number of queries allocated to $\Alg$. By a standard upper bound on the size of $q$-wise independent families, this leads to the desired upper bound and, roughly speaking, lets us treat $\phi, \psi$ in the proof above as if they are uniformly random functions (as the algorithm can only ever access their values on at most $q$ different inputs).

\section{Proof of the main upper bound}
\label{sec:ub-proof}
In this section, we prove our main upper bound, stated formally as \cref{thm:main-ub-formal} below. We begin by reviewing the setup and terminology. We fix an alphabet $\Sigma$ and a positive integer $n$ denoting the length of sequences we will sample from.  We are given a class $\MF \subset \Delta(\Sigma^n)$ of probability measures (i.e., $\MF$ is \emph{known} to the algorithm). Suppose some unknown measure $\mu^\st \in \MF$ is fixed. With a slight abuse of notation, we write, for any $\mu \in \Delta(\Sigma^n)$ and $x \in \Sigma^{i}$, $i < n$, $\mu(x) = \mu(\{ x \circ x' \mid x' \in \Sigma^{n-i} \})$; in words, we use $\mu(\cdot)$ to denote both the probability measure on $\Sigma^n$ as well as its marginals on $\Sigma^{< n}$.

We aim to come up with an algorithm that has knowledge of $\MF$ (but not $\mu^\star$), and which can make queries to a mapping $\hat \mu: \Sigma^{\leq n} \to \BR_{\geq 0}$ which satisfies the following assumption: 
\begin{assumption}[$\ell_\infty$-closeness]
  \label{asm:linfty-close-gen}
  We assume that $\hat \mu : \Sigma^{\leq n} \to \BR_{\geq 0}$ is a mapping so that, for some $R > 0$, for all $x \in \Sigma^{\leq n}$,
  \begin{align}
\frac{1}{R} \leq \frac{\mu^\star(x)}{\hat \mu(x)} \leq R\label{eq:R-ratio-asm}.
  \end{align}
  Moreover, we assume that $\mu^\star(x) = \hat \mu(x)$ for all $x \in \Sigma^n$. We also adopt the convention that $\mu(\emptyset) = 1$.\footnote{Note that $\frac{0}{0}$ in the context of \cref{eq:R-ratio-asm} is to be interpreted as $1$.}
\end{assumption}

Recall, as we discussed in \cref{sec:inference-llms}, \cref{ques:tilted-ttt} asks whether \emph{test-time training} can speed up LLM inference. In our present notation, this question becomes the following: \emph{does the knowledge that $\mu^\star \in \MF$ allow us to sample more efficiently from $\mu^\star$?} Recall that $\MF$ corresponds to the class of value functions which the test-time training algorithm is optimizing over. \cref{thm:main-ub-formal} below gives an affirmative answer to this question when $\MF$ is not too large; in particular, as long as $|\MF| \leq \exp(o(n^2))$, then \cref{thm:main-ub-formal} gives that we can sample using $o(n^2)$ queries to $\hat \mu(\cdot)$, which beats the quadratic lower bound of \cref{thm:main-lb} (i.e., \cref{thm:weak-lb}). 
\begin{theorem}[Main upper bound]
  \label{thm:main-ub-formal}
Fix $n \in \BN$, $\delta \in (0,1)$, a finite alphabet $\Sigma$, and a known class of probability distributions $\MF \subseteq \Delta(\Sigma^n)$. Let $\mu^\star \in \MF$ be an unknown distribution, and let $\hat \mu : \Sigma^{\leq n} \to \BR_{\geq 0}$ satisfy \cref{asm:linfty-close-gen} with respect to $\mu^\star$ for some parameter $R \geq 1$. Then there is a randomized algorithm (\cref{alg:pf-learning}) which outputs a sample from a distribution $\hat\nu \in \Delta(\Sigma^n)$ satisfying $\tvd{\mu^\star}{\hat\nu} \leq \delta$  
and which makes at most 
\begin{align}
O\left(n \cdot |\Sigma| \cdot R^{12} \cdot \sqrt{\log |\MF|} \cdot \log \frac{1}{\delta} \cdot \left( \log \frac{R \log |\MF|}{\delta}\right) \right)\nonumber
\end{align}
queries to $\hat \mu(\cdot)$. %
\end{theorem}
Recall the structure of \cref{alg:pf-learning} (discussed in depth in \cref{sec:tech-upperbound}): for each outer iteration $k$, the algorithm maintains a set $\MS^j$ of partial sequences. At each step $j$, the algorithm removes an (arbitrary) element $z_j$ of $\MS^{j-1}$ and replaces it with set $\MT_j$ of extensions of $z_j$, returned by \cref{alg:ts-rs}. In \cref{sec:ts-fs-analysis}, we prove a number of lemmas that allow us to analyze \cref{alg:ts-rs}; then, in \cref{sec:pf-learning-analysis}, we put these pieces together and prove \cref{thm:main-ub-formal}.

\subsection{Analysis for \cref{alg:ts-rs}} 
\label{sec:ts-fs-analysis}

\begin{algorithm}[H]
\caption{Tree-structured Rejection Sampling}
\label{alg:ts-rs}
\begin{algorithmic}[1]
\Require Alphabet $\Sigma$, density ratio parameter $M \geq 1$, oracle $\hat \mu: \Sigma^{\leq n} \to \BR_{\geq 0}$, node $w \in \Sigma^{\leq n}$, class of densities $\MG \subset \Delta(\Sigma^n)$.
\State Initialize $\Gret \gets \MG$, $\MT \gets \emptyset$. 
\For{$1 \leq i \leq M$}
\State Initialize $w_{i,0} \gets w$, $\MG_{i,0} \gets \{ \eta \in \MG \mid \frac{1}{R} \leq \frac{\hat \mu(w_{i,0})}{\eta(w_{i,0})} \leq R \}$, $\bar \eta_{i,0} \gets \frac{1}{|\MG_{i,0}|} \sum_{\eta \in \MG_{i,0}} \eta$.\label{line:initialize-ts-rs}
\For{$t \geq 1$}\label{line:for-t}
\State Set $w_{i,t} \gets w_{i,t-1}$.
\While{$|w_{i,t}| < n$, $ \frac{\hat \mu(w_{i,t} \circ y)}{\hat \mu(w_{i,t-1})  \bar \eta_{i,t-1}(w_{i,t} \circ y \mid w_{i,t-1})} \leq M$ $\forall y \in \Sigma$ and $\frac{1}{M} \leq \frac{\hat \mu(w_{i,t})}{\hat \mu(w_{i,t-1})  \bar \eta_{i,t-1}(w_{i,t} \mid w_{i,t-1})}$}\label{line:while-wit}
\State Sample $y \sim \bar \eta_{i,t-1}(\cdot \mid w_{i,t})$.\label{line:extend-1}
\State Set $w_{i,t} \gets w_{i,t} \circ y$. \label{line:extend-2}
\EndWhile
\If{$|w_{i,t}| < n$ and $ \frac{\hat \mu(w_{i,t} \circ y)}{\hat \mu(w_{i,t-1})  \bar \eta_{i,t-1}(w_{i,t} \circ y \mid w_{i,t-1})} > M$ for some $y \in \Sigma$}\label{line:if-big}
\State Sample $\xi \sim \mathrm{Ber}(\frac{\hat \mu(w_{i,t} )}{\hat \mu(w_{i,t-1})  \bar \eta_{i,t-1}(w_{i,t} \mid w_{i,t-1})}  \cdot \frac{1}{M})$.\label{line:xi-big}
\State If $\xi = 1$, add $w_{i,t}$ to $\MT$; set $\Gret \gets \Gret \cap \{ \eta \in \MG_{i,t-1} \mid \hat \mu(w_{i,t} \circ y) / \eta (w_{i,t} \circ y) \leq R \}$.\label{line:add-to-T-mid}
\State \textbf{Break} (i.e., proceed to the next value of $i$). \label{line:break-m-trial-end}
\ElsIf{$|w_{i,t}| < n$ and $\frac{1}{M} > \frac{\hat \mu(w_{i,t})}{\hat \mu(w_{i,t-1})  \bar \eta_{i,t-1}(w_{i,t} \mid w_{i,t-1})}$}\label{line:if-small}
\State Sample $\xi \sim \mathrm{Ber}(\frac{\hat \mu(w_{i,t})}{\hat \mu(w_{i,t-1})  \bar \eta_{i,t-1}(w_{i,t} \mid w_{i,t-1})})$. 
\State If $\xi = 0$, \textbf{Break} (i.e., proceed to the next value of $i$).\label{line:xi-small}
\State If $\xi = 1$, set $\MG_{i,t} \gets \{ \eta \in \MG_{i,t-1} \mid \frac{1}{R} \leq \frac{\hat \mu(w_{i,t})}{\eta(w_{i,t})} \leq R \}$ and $\bar \eta_{i,t} \gets \frac{1}{|\MG_{i,t}|} \sum_{\eta \in \MG_{i,t}} \eta$.\label{line:proceed-next-t}
\Else\Comment{\emph{Must have $|w_{i,t}|=n$}}\label{line:else-lengthn}
\State Sample $\xi \sim \mathrm{Ber}(\frac{\hat \mu(w_{i,t})}{\hat \mu(w_{i,t-1})  \bar \eta_{i,t-1}(w_{i,t} \mid w_{i,t-1})} \cdot \frac{1}{M})$.\label{line:xi-else}
\State If $\xi = 1$, add $w_{i,t}$ to $\MT$ and \textbf{Break} (proceed to the next value of $i$). \label{line:add-to-T-end}
\EndIf
\EndFor
\EndFor
\State \Return $\MT, \Gret$. 
\end{algorithmic}
\end{algorithm}
In our analysis, we often view $\Sigma^{\leq n}$ as a complete $\Sigma$-ary tree, and accordingly refer to elements $x \in \Sigma^{\leq n}$ as \emph{nodes}. For $x = (x_1, \ldots, x_i) = x_{1:i} \in \Sigma^{\leq n}$ and $i < j$, we will write $x_{1:j} \prec x_{1:i}$, and say that $x_{1:j}$ is an \emph{ancestor} of $x_{1:i}$ (which is a \emph{descendent} of $x_{1:j}$). 

To analyze \cref{alg:ts-rs}, we need some additional notation to keep track of the nodes $w_{i,t}$ maintained in \cref{alg:ts-rs}. Given a class of densities $\MG$, a node $w \in \Sigma^{\leq n}$, $M \geq 1$, and $t \geq 0$, we define the set of \emph{$(\MG, M,t)$-intermediate nodes with respect to $w$, denoted $\Nint_{\MG,M,t}(w)$, and the set of $(\MG,M,t)$-final nodes with respect to $w$, denoted $\Nfinal_{\MG,M,t}(w)$, as follows:}
\begin{itemize}
\item First, we define $\Nint_{\MG,M,0}(w) := \{ w \}$ and $\Nfinal_{\MG,M,0}(w) = \emptyset$.
\item We define $\Nint_{\MG,M,t}(w),\Nfinal_{\MG,M,t}(w)$ inductively in terms of $\Nint_{\MG, M,t-1}(w)$, as follows: first, for each $w' \succeq w$, we define
\begin{align}
\MG_{t-1}(w') = \bigcap_{0 \leq s \leq t-1}\left\{ \eta \in \MG \mid \frac{1}{R} \leq \frac{\hat \mu(w'')}{\eta(w'')} \leq R \  \forall w'' \preceq w', \ w'' \in \Nint_{\MG,M,s}(w) \right\}\label{eq:define-gt}.
\end{align}
In words, $\MG_{t-1}(w')$ is the set of densities which is consistent with the value of $\hat \mu$ at all ancestors of $w'$ which are in some $(\MG, M,s)$ intermediate set for $s \leq t-1$. We then set
\begin{align}
\bar \eta_{t-1,w'} := \frac{1}{|\MG_{t-1}(w')|} \sum_{\eta \in \MG_{t-1}(w')} \eta\nonumber
\end{align}
to be the mean of the densities in $\MG_{t-1,w'}$. To define $\Nint_{\MG,M,t}(w)$ and $\Nfinal_{\MG,M,t}(w)$ we consider the following properties of nodes $w' \succ w$:
\begin{enumerate}
\item[\textbf{(P1)}] $w' \succ w''$ for some $w'' \in \Nint_{\MG,M,t-1}(w)$ so that
\begin{align}
\frac{\hat \mu(w')}{\hat \mu(w'') \cdot \bar \eta_{t-1,w''}(w' \mid w'')} < \frac{1}{M}\nonumber
\end{align}
\item[\textbf{(P2)}] $w' \succ w''$ for some $w'' \in \Nint_{\MG,M,t-1}(w)$ so that either $|w'| = n$ or for some $y \in \Sigma$,
\begin{align}
\frac{\hat \mu(w'\circ y)}{\hat \mu(w'') \cdot \bar \eta_{t-1,w''}(w' \circ y \mid w'')} > M\nonumber. 
\end{align}
\end{enumerate}

We then define $\Nint_{\MG,M,t}(w)$ to be the set of $w'$ satisfying \textbf{(P1)} and that neither it nor any of its ancestors satisfy either \textbf{(P1), (P2)}. We define $\Nfinal_{\MG,M,t}(w)$ to be the union of $\Nfinal_{\MG,M,t-1}(w)$ and the set of $w'$ satisfying \textbf{(P2)} and that none of its ancestors satisfy either \textbf{(P1), (P2)}. 
\end{itemize}

For $w \in \Sigma^{\leq n}$ and a subset $\MN \subset \Sigma^{\leq n}$, we say that $\MN$ is a \emph{simple $w$-cut} if for each leaf $v \in \Sigma^n$, $v \succ w$, there is a single node $w' \in \MN$ so that $w \preceq w' \preceq v$ (i.e., a single node in $\MN$ lies on the path from $w$ to $v$).
The next lemma shows that, at every stage of the construction, the intermediate and final nodes still form such a simple cut.
\begin{lemma}
\label{lem:n-cut}
For each $t \geq 0$, $\Nint_{\MG, M, t}(w) \cup \Nfinal_{\MG,M,t}(w)$ is a simple $w$-cut.
\end{lemma}
\begin{proof}
We argue by induction on $t$. For $t=0$ we have $\Nint_{\MG,M,0}(w)=\{w\}$ and $\Nfinal_{\MG,M,0}(w)=\emptyset$, so the claim is immediate.

Now fix $t \geq 1$ and assume that
\begin{align}
\MC_{t-1} := \Nint_{\MG,M,t-1}(w) \cup \Nfinal_{\MG,M,t-1}(w)\nonumber
\end{align}
is a simple $w$-cut. Let $v \in \Sigma^n$ be any leaf with $v \succeq w$. Since $\MC_{t-1}$ is a simple $w$-cut, there is a unique node $u \in \MC_{t-1}$ on the path from $w$ to $v$.

If $u \in \Nfinal_{\MG,M,t-1}(w)$, then $u$ is also in $\Nfinal_{\MG,M,t}(w)$. Moreover, no node of $\Nint_{\MG,M,t}(w)$ or of $\Nfinal_{\MG,M,t}(w) \setminus \Nfinal_{\MG,M,t-1}(w)$ can lie below $u$, because every node added at time $t$ must be a descendant of some element of $\Nint_{\MG,M,t-1}(w)$, whereas $u$ is already the unique cut node on the path from $w$ to $v$. Thus $u$ is the unique node of $\Nint_{\MG,M,t}(w) \cup \Nfinal_{\MG,M,t}(w)$ on this path.

It remains to consider the case $u \in \Nint_{\MG,M,t-1}(w)$. Along the path from $u$ to $v$, the leaf $v$ itself satisfies \textbf{(P2)} relative to $u$, simply because $|v|=n$. Hence there is at least one descendant of $u$ on the path to $v$ satisfying either \textbf{(P1)} or \textbf{(P2)} relative to $u$. Let $z$ be the first such descendant. By definition, $z$ belongs to $\Nint_{\MG,M,t}(w)$ if it satisfies \textbf{(P1)} but not \textbf{(P2)}, and to $\Nfinal_{\MG,M,t}(w)$ if it satisfies \textbf{(P2)}.

By the choice of $z$, no proper ancestor of $z$ which is a descendant of $u$ satisfies either \textbf{(P1)} or \textbf{(P2)}. Also, no proper descendant of $z$ can belong to $\Nint_{\MG,M,t}(w)$ or to $\Nfinal_{\MG,M,t}(w) \setminus \Nfinal_{\MG,M,t-1}(w)$, because membership in those sets requires that neither the node itself nor any of its ancestors satisfy \textbf{(P1)} or \textbf{(P2)}. Finally, there is no element of $\Nfinal_{\MG,M,t-1}(w)$ below $u$ on the path to $v$, since $u$ was the unique element of $\MC_{t-1}$ on that path. Therefore $z$ is the unique node of $\Nint_{\MG,M,t}(w) \cup \Nfinal_{\MG,M,t}(w)$ lying on the path from $w$ to $v$.

Since this holds for every leaf $v \succeq w$, the set $\Nint_{\MG,M,t}(w) \cup \Nfinal_{\MG,M,t}(w)$ is a simple $w$-cut.
\end{proof}

The next lemma shows that the number of steps after which we can define subsequent nonempty intermediate node sets is at most $n - |w|$.
\begin{lemma}
\label{lem:int-bounded}
There is some $t_0 \leq n-|w|$ so that for all $t \geq t_0$, $\Nint_{\MG, M, t}(w) = \emptyset$, and therefore $\Nfinal_{\MG,M,t_0}(w) = \Nfinal_{\MG,M,t_0+1}(w) = \cdots$. 
\end{lemma}
\begin{proof}
We first claim that every node in $\Nint_{\MG,M,t}(w)$ lies at depth at least $|w|+t$. Indeed, for $t=0$ this is immediate since $\Nint_{\MG,M,0}(w)=\{w\}$. For $t \geq 1$, if $u \in \Nint_{\MG,M,t}(w)$, then by definition $u$ satisfies \textbf{(P1)} relative to some node $u' \in \Nint_{\MG,M,t-1}(w)$, and hence $u \succ u'$. By induction, $|u'| \geq |w|+t-1$, so $|u| \geq |u'|+1 \geq |w|+t$.

Next, no intermediate node can be a leaf. Indeed, if $|u|=n$, then $u$ satisfies \textbf{(P2)} by definition. But a node belongs to $\Nint_{\MG,M,t}(w)$ only if neither it nor any of its ancestors satisfy \textbf{(P1)} or \textbf{(P2)}. Therefore every $u \in \Nint_{\MG,M,t}(w)$ must satisfy $|u|<n$.

Combining the two observations, if $t \geq n-|w|$ and $u \in \Nint_{\MG,M,t}(w)$, then $|u| \geq |w|+t \geq n$, which forces $|u|=n$, a contradiction. Hence
\begin{align}
\Nint_{\MG,M,t}(w)=\emptyset \qquad \forall t \geq n-|w|.\nonumber
\end{align}
Thus we may take $t_0 := n-|w|$.

Finally, once $\Nint_{\MG,M,t}(w)$ is empty, the inductive definition adds no further nodes to $\Nfinal_{\MG,M,t}(w)$, so the final sets remain constant for all subsequent times.
\end{proof}

The below lemma relates the values of $w_{i,t}$ and $\MG_{i,t}$ that are computed in \cref{alg:ts-rs} to the sets $\Nint_{\MG,M,t}(w)$, $\Nfinal_{\MG,M,t}(w)$, and $\MG_t(w)$ defined above. 
\begin{lemma}
\label{lem:alg-bookkeeping}
Fix any $i \in [M]$ denoting the outer iteration in \cref{alg:ts-rs}. For any $t \geq 1$, the final value of $w_{i,t}$ computed in \cref{alg:ts-rs} is an element of $\Nint_{\MG,M,t}(w) \cup \Nfinal_{\MG,M,t}(w)$, and the candidate class $\MG_{i,t}$ defined in Line \ref{line:proceed-next-t} is exactly $\MG_t(w_{i,t})$ as defined in \cref{eq:define-gt}. Moreover, the set $\MT$ computed in \cref{alg:ts-rs} satisfies $\MT \subset \Nfinal_{\MG,M}(w)$ almost surely. 
\end{lemma}
\begin{proof}
We prove the first two claims by induction on $t$, restricting attention to those values of $t$ for which the $i$th outer-loop iteration actually reaches stage $t$.

The base case $t=0$ is immediate; now fix $t \geq 1$ and assume the claim holds for $t-1$. Suppose the $i$th outer-loop iteration reaches stage $t$. Then stage $t-1$ cannot have terminated in a final branch, so the only way to proceed is that Line \ref{line:proceed-next-t} was executed at step $t-1$. Consequently,
\begin{align}
w_{i,t-1} \in \Nint_{\MG,M,t-1}(w)
\qquad\text{and}\qquad
\MG_{i,t-1}=\MG_{t-1}(w_{i,t-1}).
\nonumber
\end{align}
In particular, the proposal distribution used at step $t$ is exactly $\bar \eta_{t-1,w_{i,t-1}}$.

At step $t$, the while loop (Line \ref{line:while-wit}) walks down the sampled path until the first descendant $z \succ w_{i,t-1}$ for which either \textbf{(P1)} or \textbf{(P2)} relative to $w_{i,t-1}$ is triggered. By the recursive definition of $\Nint_{\MG,M,t}(w)$ and $\Nfinal_{\MG,M,t}(w)$, this implies
\begin{align}
w_{i,t} \in \Nint_{\MG,M,t}(w) \cup \Nfinal_{\MG,M,t}(w).
\nonumber
\end{align}
If Line \ref{line:proceed-next-t} is reached at stage $t$, then $w_{i,t}$ satisfies \textbf{(P1)} relative to $w_{i,t-1}$, so the algorithm updates
\begin{align}
\MG_{i,t}
= \left\{ \eta \in \MG_{i,t-1} \mid \frac{1}{R} \leq \frac{\hat \mu(w_{i,t})}{\eta(w_{i,t})} \leq R \right\}.
\nonumber
\end{align}
Using the induction hypothesis $\MG_{i,t-1}=\MG_{t-1}(w_{i,t-1})$, and using that exactly one element of each of $\MG_{1}(w_{i,t-1}), \MG_2(w_{i,t-1}), \ldots, \MG_{t-1}(w_{i,t-1})$ lies on the path from $w$ to $w_{i,t}$ (namely $w_{i,1}, \ldots, w_{i,t-1}$), it follows from the definition of $\MG_t(\cdot)$ (see \cref{eq:define-gt}) that $\MG_{i,t} = \MG_t(w_{i,t})$. 
This completes the induction.

Finally, every element added to $\MT$ is added either on Line \ref{line:add-to-T-mid} or on Line \ref{line:add-to-T-end}. In the first case, the node $w_{i,t}$ necessarily satisfies \textbf{(P2)} relative to the intermediate node $w_{i,t-1}$ because some child $w_{i,t}\circ y$ violates the upper threshold; in the second case, $|w_{i,t}|=n$, so $w_{i,t}$ again satisfies \textbf{(P2)}. In both cases, the argument above shows that the added node lies in $\Nfinal_{\MG,M,t}(w) \subseteq \Nfinal_{\MG,M}(w)$. %
\end{proof}

Next, we show a key notion of ``progress'' made by \cref{alg:ts-rs}. Whenever if returns a set $\MT$ with incomplete sequences (i.e., of length less than $n$), then its returned set of densities, $\Gret$, must decrease in size relative to that of $\MG$. 
\begin{lemma}
\label{lem:gret-decrease}
Fix any $w \in \Sigma^{\leq n}$ and $\MG \subset \Delta(\Sigma^n)$. Let $\Erestrict$ denote the event that \cref{alg:ts-rs} returns a set $\MT$ for which $\MT \cap \Sigma^{<n} \neq \emptyset$. Then on the event $\Erestrict$, we have $|\Gret| \leq \frac{R^2|\MG|}{M}$.
\end{lemma}
\begin{proof}
On the event $\Erestrict$, there is at least one outer-loop iteration $i \in [M]$ and one stage $t \geq 1$ at which the algorithm executes Line \ref{line:add-to-T-mid}. Fix any such pair $(i,t)$, and let $z := w_{i,t} \in \Sigma^{<n}$ denote the node added to $\MT$ at that line. By the condition of the branch, there exists some $y \in \Sigma$ such that
\begin{align}
\frac{\hat \mu(z \circ y)}{\hat \mu(w_{i,t-1}) \cdot \bar \eta_{i,t-1}(z \circ y \mid w_{i,t-1})} > M.
\label{eq:gret-upper-trigger}
\end{align}
By definition of the class $\MG_{i,t-1}$ (Line \ref{line:proceed-next-t} of \cref{alg:ts-rs}), every $\eta \in \MG_{i,t-1}$ satisfies $\frac{1}{R} \leq \frac{\hat \mu(w_{i,t-1})}{\eta(w_{i,t-1})} \leq R$. In particular, $\hat \mu(w_{i,t-1}) \leq R \cdot \eta(w_{i,t-1})$ for every $\eta \in \MG_{i,t-1}$. Using that $\bar \eta_{i,t-1} = \frac{1}{|\MG_{i,t-1}|} \sum_{\eta \in \MG_{i,t-1}} \eta$, we see that $\hat \mu(w_{i,t-1}) \leq R \cdot \bar \eta_{i,t-1}(w_{i,t-1})$. Using this fact together with \cref{eq:gret-upper-trigger}, we see that 
\begin{align}
\frac{\hat \mu(z \circ y)}{\bar \eta_{i,t-1}(z \circ y)} \geq \frac{\hat \mu(z \circ y)}{\bar \eta_{i,t-1}(z \circ y)} \cdot \frac{\bar \eta_{i,t-1}(w_{i,t-1})}{\hat \mu(w_{i,t-1})\cdot R} = \frac{\hat \mu(z \circ y)}{R \cdot \hat \mu(w_{i,t-1}) \cdot \bar \eta_{i,t-1}(z \circ y \mid w_{i,t-1})} > \frac{M}{R}\label{eq:hatmu-bareta-ratio}.
\end{align}
Since we have assumed that $z$ is added to $\MT$ on Line \ref{line:add-to-T-mid}, then the value of $\Gret$ eventually returned by \cref{alg:ts-rs} satisfies
\begin{align}
\Gret \subset \left\{ \eta \in \MG_{i,t-1} \mid \frac{\hat \mu(z \circ y)}{\eta(z \circ y)} \leq R \right\}.
\end{align}
Combining this with \cref{eq:hatmu-bareta-ratio} and using Markov's inequality yields $|\Gret| \leq |\MG_{i,t-1}| \cdot \frac{R^2}{M} \leq |\MG| \cdot \frac{R^2}{M}$.

\end{proof}

We let $\Nfinal_{\MG,M}(w)$ denote the set $\Nfinal_{\MG,M,t_0}(w) = \Nfinal_{\MG,M,t_0+1}(w) = \cdots$. Fix any node $w_t \in \Nint_{\MG,M,r}(w)$. To simplify notation, we write $\Nappend_{\MG,M,t}(w) := \Nint_{\MG,M,t}(w) \cup (\Nfinal_{\MG,M,t}(w) \backslash \Nfinal_{\MG,M,t-1}(w))$. Moreover, we define
\begin{align}
\MS_{\MG,M,t}(w) := \{ w' \succeq w \mid \exists w_{t-1} \in \Nint_{\MG,M,t-1}(w),\ w_t \in \Nappend_{\MG,M,t}(w) \text{ so that } w_{t-1} \prec w' \prec w \}\nonumber
\end{align}
to denote the set of nodes sandwiched between some node in $\Nint_{\MG,M,t-1}(w)$ (exclusive) and some node in $\Nappend_{\MG,M,t}(w)$ (exclusive). For each node $w' \in \MS_{\MG,M,t}(w)$ for some $t \geq 1$, we let $\prev(w')$ denote the unique ancestor of $w'$ in $\Nint_{\MG,M,t-1}(w)$ (its existence follows from definition and its uniqueness follows from \cref{lem:n-cut}). (Note that $\prev(w')$ also depends on $w, \MG, M$, though we omit this dependence for clarity.)

Next, fix one outer-loop iteration $i$ of \cref{alg:ts-rs}. For $t \geq 0$ and $w_t \in \Nappend_{\MG,M,t}(w)$, let $A_{i,t}(w_t)$ be the event that at the end of step $t$ in iteration $i$, the algorithm sets $w_{i,t} = w_t$ and \emph{does not break} at step $t$ (i.e., $\xi = 1$ in Line \ref{line:proceed-next-t}). On $A_{i,t}(w_t)$, the candidate class used by the algorithm (denoted by $\MG_{i,t}$ in \cref{alg:ts-rs}) is exactly $\MG_{i,t} = \MG_t(w_t)$, by \cref{lem:alg-bookkeeping}. Hence, on the event $A_{i,t}(w_t)$, the proposal distribution $\bar \eta_{i,t}$ used in \cref{alg:ts-rs} is exactly $\bar \eta_{t,w_t}$. Moreover, for $w_t \in \MS_{\MG,M,t}(w) \cup \Nappend_{\MG,M,t}(w)$, we let $\bar A_{i,t}(w_t)$ be the event that the algorithm sets $w_{i,t} = w_t$ at some point in its execution; in particular, for $w \in \Nappend_{\MG,M,t}(w)$, $\bar A_{i,t}(w_t) \supseteq A_{i,t}(w_t)$, but $\bar A_{i,t}(w_t)$ also includes the event that the algorithm \emph{does break} at step $t$.  The following lemma gives an expression for the probability of $\bar A_{i,t}(w')$ (denoted using $\BP(\cdot)$), over the randomness of \cref{alg:ts-rs}. 
\begin{lemma}
\label{lem:final-probability}
For any $t \geq 1$ and $w' \in \MS_{\MG,M,t}(w)$, we have: %
\begin{align}
\BP(\bar A_{i,t}(w') ) = \bar \eta_{t-1, \prev(w')}(w' \mid \prev(w'))\cdot  \frac{\hat \mu(\prev(w'))}{\hat \mu(w)}.\label{eq:ait-probs}
\end{align}
Moreover, for any $w' \in \Nfinal_{\MG,M}(w)$, the expected number of occurrences of $w'$ in the set $\MT$ constructed in \cref{alg:ts-rs} is given as follows:
\begin{align}
\E \left[ \occ(w', \MT) \right] = \frac{\hat \mu(w')}{\hat \mu(w)}\label{eq:exp-occurrences}.
\end{align}
\end{lemma}
\begin{proof}
We first prove \cref{eq:ait-probs}. Fix $t \geq 1$ and $w' \in \MS_{\MG,M,t}(w)$. For each $0 \leq r < t$, let $w_r$ denote the unique node on the path from $w$ to $w'$ that belongs to the simple cut
$
\Nint_{\MG,M,r}(w) \cup \Nfinal_{\MG,M,r}(w).\nonumber
$
By \cref{lem:alg-bookkeeping}, we must have $w_r \in \Nint_{\MG,M,r}(w)$ for every $0 \leq r < t$. In particular, $w_0=w$, and for each $1 \leq r \leq t$, the node $w_r$ satisfies \textbf{(P1)} relative to $w_{r-1}$.

For each $1 \leq r \leq t$, conditioned on $A_{i,r-1}(w_{r-1})$, the probability that the algorithm next reaches $w_r$ and survives the Bernoulli test in the \textbf{(P1)} branch (on Line \ref{line:proceed-next-t}) is
\begin{align}
\bar \eta_{r-1,w_{r-1}}(w_r \mid w_{r-1}) \cdot \frac{\hat \mu(w_r)}{\hat \mu(w_{r-1}) \cdot \bar \eta_{r-1,w_{r-1}}(w_r \mid w_{r-1})}
= \frac{\hat \mu(w_r)}{\hat \mu(w_{r-1})}.\nonumber
\end{align}
Therefore
\begin{align}
\BP(A_{i,t-1}(w_{t-1})) = \prod_{r=1}^{t-1} \frac{\hat \mu(w_r)}{\hat \mu(w_{r-1})} = \frac{\hat \mu(w_{t-1})}{\hat \mu(w)}.\nonumber
\end{align}
To establish  \cref{eq:ait-probs}, we next note that
\begin{align}
\BP(\bar A_{i,t}(w')) = \BP(A_{i,t-1}(w_{t-1})) \cdot \bar \eta_{t-1,w_{t-1}}(w' \mid w_{t-1}) = \frac{\hat \mu(w_{t-1})}{\hat \mu(w)} \cdot \bar\eta_{t-1,w_{t-1}}(w' \mid w_{t-1})\label{eq:bar-ait-equality},
\end{align}
which yields \cref{eq:ait-probs} by noting that $w_{t-1} = \prev(w')$.

We proceed to prove \cref{eq:exp-occurrences}. Fix $w' \in \Nfinal_{\MG,M}(w)$, and let $t \geq 1$ be the smallest index for which $w' \in \Nfinal_{\MG,M,t}(w)$. For each $0 \leq r < t$, let $w_r$ denote the unique node on the path from $w$ to $w'$ that belongs to the simple cut $\Nint_{\MG,M,r}(w).$ It follows from the same argument used to derive \cref{eq:ait-probs} that
\begin{align}
\BP(A_{i,t-1}(w_{t-1})) = \frac{\hat \mu(w_{t-1})}{\hat \mu(w)}.
\label{eq:ts-rs-intermediate-chain}
\end{align}

Now condition on $A_{i,t-1}(w_{t-1})$. Since $w'=w_t$ satisfies \textbf{(P2)} relative to $w_{t-1}$, the algorithm reaches $w'$ with probability $\bar \eta_{t-1,w_{t-1}}(w' \mid w_{t-1})$, and then adds $w'$ to $\MT$ (on either Line \ref{line:add-to-T-mid} or Line \ref{line:add-to-T-end}) with Bernoulli probability
\begin{align}
\frac{\hat \mu(w')}{\hat \mu(w_{t-1}) \cdot \bar \eta_{t-1,w_{t-1}}(w' \mid w_{t-1})} \cdot \frac{1}{M}.
\nonumber
\end{align}
This Bernoulli parameter is indeed at most $1$: the node $w'$ was obtained by extending from $w_{t-1}$ in the immediately preceding step of the while loop, and that extension could only occur while the while-condition in Line \ref{line:while-wit} held. 
Thus, if $B_i$ denotes the event that the $i$th outer-loop iteration contributes one copy of $w'$ to $\MT$, then
\begin{align}
\BP(B_i \mid A_{i,t-1}(w_{t-1}))
= \bar \eta_{t-1,w_{t-1}}(w' \mid w_{t-1}) \cdot \frac{\hat \mu(w')}{\hat \mu(w_{t-1}) \cdot \bar \eta_{t-1,w_{t-1}}(w' \mid w_{t-1})} \cdot \frac{1}{M}
= \frac{\hat \mu(w')}{M \cdot \hat \mu(w_{t-1})}.\nonumber
\end{align}
Combining this with \cref{eq:ts-rs-intermediate-chain}, we obtain
\begin{align}
\BP(B_i)
= \frac{\hat \mu(w_{t-1})}{\hat \mu(w)} \cdot \frac{\hat \mu(w')}{M \cdot \hat \mu(w_{t-1})}
= \frac{\hat \mu(w')}{M \cdot \hat \mu(w)}.\nonumber
\end{align}

Each outer-loop iteration contributes at most one copy of $w'$ to $\MT$, and the $M$ outer-loop iterations are identically distributed. Therefore
\begin{align}
\E[\occ(w',\MT)]
= \sum_{i=1}^M \BP(B_i)
= M \cdot \frac{\hat \mu(w')}{M \cdot \hat \mu(w)}
= \frac{\hat \mu(w')}{\hat \mu(w)},\nonumber
\end{align}
as claimed.
\end{proof}

The next lemma bounds the query complexity of a single call to \cref{alg:ts-rs} in terms of the weighted depths of the returned nodes.
\begin{lemma}
\label{lem:query-cost}
The expected number of oracle calls to $\hat \mu$ made by \cref{alg:ts-rs} is bounded above by
\begin{align}
O \left( |\Sigma| \cdot R \cdot M^2 \right) \cdot \left( \E \left[\sum_{w' \in \MT}  \frac{\mu^\st(w')}{\hat \mu(w')} \cdot |w'|  \right] - \frac{\mu^\st(w)}{\hat \mu(w)} \cdot |w|\right)\nonumber. %
\end{align}
Moreover, the number of oracle calls is bounded almost surely by $O(|\Sigma| M  n)$. 
\end{lemma}
\begin{proof}
For each $i \in [M]$, let $T_i \in \BN$ denote the final value of $t$ reached in the for loop in Line \ref{line:for-t} of \cref{alg:ts-rs}, in outer iteration $i$. 
The expected number of oracle calls made by \cref{alg:ts-rs}, which we denote by the random variable $V$, may be bounded as follows:
\begin{align}
\E[V] \leq M |\Sigma| \cdot O \left( \sum_{t=1}^n \sum_{w' \in \MS_{\MG,M,t}(w)} \BP(\bar A_{i,t}(w'))+ \E[T_i+1]\right)%
\label{eq:ev-init-bnd}.
\end{align}
Indeed, for each outer loop iteration $i \in [M]$, \cref{alg:ts-rs} queries $\hat \mu$ at each of the nodes along the path $w_{i,0} \to w_{i,1} \to \cdots \to w_{i,T_i}$, as well as at each of the children of each of these nodes (in order to check the condition in the while statement on \cref{line:while-wit}). For each $0 \leq t \leq T_i$, we have $w_t \in \Nappend_{\MG,M,t}(w)$: the first term in \cref{eq:ev-init-bnd} accounts for the queries made at the nodes in-between the $w_{i,t}$, and the second term accounts for the queries made at the nodes $w_{i,t}$. 

First, we note that, since $M \geq 2$, we have $\BP(T_i \geq t+1 \mid T_i \geq t) \leq 1/M \leq 1/2$ for each $t \geq 0$. Thus, $\E[T_i] \leq 4$. We proceed to bound the first term in \cref{eq:ev-init-bnd}.

Next, for any $t \geq 1$, we compute
\begin{align}
&  \sum_{w' \in \MS_{\MG,M,t}(w)} \BP(\bar A_{i,t}(w')) \nonumber\\
= &  \sum_{w' \in \MS_{\MG,M,t}(w)} \frac{\bar \eta_{t-1,\prev(w')}(w') \cdot \hat \mu(\prev(w'))}{\hat \mu(w)}\nonumber\\
\leq &\sum_{w' \in \MS_{\MG,M,t}(w)} \frac{M \cdot \hat \mu(w')}{\hat \mu(w)}\nonumber\\
\leq & \sum_{w' \in \MS_{\MG,M,t}(w)} \frac{MR \cdot  \mu^\st(w')}{\hat \mu(w)}\nonumber\\
= & \sum_{w' \in \Nappend_{\MG,M,t}(w)} \frac{MR \cdot  \mu^\st(w')}{\hat \mu(w)} \cdot (|w' - \prev(w')| - 1)\label{eq:aitbar-bound},
\end{align}
where:
\begin{itemize}
\item The first equality uses \cref{eq:ait-probs} of \cref{lem:final-probability}. 
\item The following inequality uses the fact that $w'$ (not being in $\Nappend_{\MG,M,t}(w)$) does not satisfy \textbf{(P1)}, which gives
\begin{align}
\frac{\hat \mu(w')}{\hat \mu(\prev(w')) \cdot \bar \eta_{t-1,\prev(w')}(w' \mid \prev(w'))} \geq \frac{1}{M}.\label{eq:muhat-parent}
\end{align}
\item The following inequality uses \cref{asm:linfty-close-gen}. 
\item The final equality rearranges the summation, using that for each $w' \in \MS_{\MG,M,t}(w)$ we can write
\begin{align}
\mu^\st(w') = \sum_{\substack{x \in \Nappend_{\MG,M,t}(w) \\ w' \prec x}} \mu^\st(x)\nonumber
\end{align}
as $\mu^\st$ is a distribution; further, each element $w' \in \Nappend_{\MG,M,t}(w)$ is preceded by exactly $|w' - \prev(w')| - 1$ elements of $\MS_{\MG,M,t}(w)$.
\end{itemize} 
Summing \cref{eq:aitbar-bound} over $t \geq 1$, we see that
\begin{align}
\sum_{t=1}^n  \sum_{w' \in \MS_{\MG,M,t}(w)} \BP(\bar A_{i,t}(w'))   \leq  & \sum_{t=1}^n \sum_{w' \in \Nappend_{\MG,M,t}(w)} \frac{MR \cdot  \mu^\st(w')}{\hat \mu(w)} \cdot |w' - \prev(w')|\nonumber\\
= & \sum_{w' \in \Nfinal_{\MG,M}(w)} \frac{MR \cdot \mu^\st(w')}{\hat \mu(w)} \cdot |w' - w|\nonumber\\
=& \sum_{w' \in \Nfinal_{\MG,M}(w)} \frac{MR \cdot \mu^\st(w')}{\hat \mu(w)} \cdot (|w'| - |w|)\nonumber\\
=& \sum_{w' \in \Nfinal_{\MG,M}(w)} \frac{MR \cdot \mu^\st(w')}{\hat \mu(w)} \cdot |w'| - \frac{MR \cdot \mu^\st(w)}{\hat \mu(w)} \cdot |w|\nonumber\\
=& MR \cdot \left( \E \left[\sum_{w' \in \MT}  \frac{\mu^\st(w')}{\hat \mu(w')} \cdot |w'|  \right] - \frac{\mu^\st(w)}{\hat \mu(w)} \cdot |w|\right)\label{eq:barait-final},
\end{align}
where:
\begin{itemize}
\item The first inequality uses \cref{eq:aitbar-bound} for each $t$.
\item The following equality uses again the fact that $\mu^\st$ is a distribution, writing, for each $w' \in \Nint_{\MG,M,t}(w) \subset \Nappend_{\MG,M,t}(w)$,
 \begin{align}
\mu^\star(w') = \sum_{\substack{x \in \Nfinal_{\MG,M}(w) \\ w' \prec x}} \mu^\star(x)\nonumber.
\end{align}
Moreover, for each $w' \in \Nfinal_{\MG,M}(w)$, if we choose $t$ so that $w' \in \Nfinal_{\MG,M,t}(w) \backslash \Nfinal_{\MG,M,t-1}(w)$, there is some $w_s \in \Nint_{\MG,M,s}(w)$ for each $1 \leq s < t$ so that $w_s \prec w'$, and summing $|w_s - \prev(w_s)| = |w_s - w_{s-1}|$ over $1 \leq s < t$ gives $|w' - w|$.
\item The following equality uses the fact that $|w' - w| = |w'| - |w|$ for $w \succ w$. 
\item The following equality uses that $\mu^\st$ is a distribution and $\Nfinal_{\MG,M}(w)$ is a simple $w$-cut (\cref{lem:n-cut}). 
\item The final equality uses \cref{eq:exp-occurrences} of \cref{lem:final-probability}, which gives $\E[\occ(w',\MT)] = \hat \mu(w')/\hat \mu(w)$ for each $w' \in \Nfinal_{\MG,M}(w)$.
\end{itemize}
The claim of the lemma follows by combining \cref{eq:barait-final} with \cref{eq:ev-init-bnd} together with the fact that $|w'| - |w| \geq 1$ for all $w' \in \MT$ and the fact that $\E \left[ \sum_{w' \in \MT} \frac{\mu^\star(w')}{\hat \mu(w')} \right] \geq 1/R$ (which follows from \cref{eq:exp-occurrences} of \cref{lem:final-probability}). 
\end{proof}

\begin{algorithm}[H]
\caption{Particle Filtering with Learning}
\label{alg:pf-learning}
\begin{algorithmic}[1]
\Require Alphabet $\Sigma$, parameter $R \geq 1$, failure tolerance $\delta \in (0,1)$, oracle $\hat \mu : \Sigma^{\leq n} \to \BR_{\geq 0}$.
\State Define $\bar R := 2R^2$, $N := C_{\N} \left\lceil \sqrt{\log |\MF|} \cdot R^2 \bar R^2 \left(\log \frac{R \log |\MF|}{\delta}\right) \right\rceil$ (for an appropriate constant $C_{\N}$, specified in \cref{lem:mtgl-concentration}), and $K := \left\lceil 4R \log_2(1/\delta) \right\rceil$.
\For{$1 \leq k \leq K$}
\State Let $\MS^0$ be the multiset consisting of $N$ copies of the empty sequence $\perp$.
\State Set $\MF^1 \gets  \{ \mu \in \MF \mid \hat \mu(y) / \mu(y) \leq \bar R \ \forall y \in \Sigma\}$, $j \gets 1$.\label{line:f1-define}
\While{$\MS^{j-1} \cap \Sigma^{<n} \neq \emptyset$}\label{line:while-s}
\State Let $z_j$ be the lexicographically first element of $\MS^{j-1} \cap \Sigma^{<n}$.\label{line:choose-zj}
\State Call \cref{alg:ts-rs} with the node $w \gets z_j$, the class $\MG \gets \MF^j$, and $M = \bar R$; suppose it returns $(\MT_j, \MG_j)$, where $\MT_j \subset \Sigma^n$ and $\MG_j \subset \MF^j$. \label{line:rs-call} %
\State Set $\MS^j \gets (\MS^{j-1} \setminus \{z_j\}) \cup \MT_j$.
\State Set $\MF^{j+1} \gets \MG_j$.\label{line:update-fj-gj}
\State Set $j \gets j+1$. 
\EndWhile
\State Let $J$ denote the final value of $j$.
\State Set $\Sterm \gets \MS^J$.
\State With probability $\min \{1, |\Sterm|/(2RN)\}$, return a uniformly random element of $\Sterm$.\label{line:k-rejection-sampling}
\EndFor
\State Return an arbitrary fixed element of $\Sigma^n$.
\end{algorithmic}
\end{algorithm}

\subsection{Analysis of \cref{alg:pf-learning}}
\label{sec:pf-learning-analysis}
We are now ready to analyze \cref{alg:pf-learning}. 
We let $J$ denote the total number of iterations executed by \cref{alg:pf-learning}, i.e., the smallest index such that every element of $\MS^J$ has length $n$. %
For convenience, for all $j > J$ we set $\MS^j := \MS^J$. Finally, for each $j \geq 0$, let $\SF^j$ denote the sigma-algebra generated by all random variables in \cref{alg:pf-learning} up to the end of the $j$th iteration of the loop.
Using the fact that the size of $\MF^j$ is decreased whenever \cref{alg:ts-rs} returns particles of length less than $n$ (\cref{lem:gret-decrease}), we first bound the number of iterations $J$: 
\begin{lemma}
\label{lem:iteration-bound}
With probability $1$, we have $J \leq 2 \bar R \cdot \log_2 |\MF| + N$.
\end{lemma}
\begin{proof}
Consider any iteration $j$ for which $\MT_j \cap \Sigma^{<n} \neq \emptyset$; \cref{lem:gret-decrease} gives that for any such $j$, we have $|\MF^{j+1}| \leq |\MF^j| \cdot \frac{R^2}{\bar R} \leq |\MF^j|/2$. 
Thus, the number of iterations $j$ for which $\MT_j \cap \Sigma^{<n} \neq \emptyset$ is bounded above by $\log_2 |\MF|$. Since $|\MT_j| \leq \bar R$ almost surely for all $j$ by the structure of \cref{alg:ts-rs}, and all elements of $\MS^{j'} \backslash \MS^0$ for any $j' \in [J]$ are elements of some $\MT^j$ for some $j < j'$, it follows that $|\MS^{j'}| \leq N + \bar R \cdot \log_2 |\MF|$ for all $j' \geq 1$. 

Further, in any iteration $j$ for which $\MT^j \cap \Sigma^{<n} = \emptyset$, we have $|\MS^{j} \cap \Sigma^{<n}| = |\MS^{j-1} \cap \Sigma^{<n}| - 1$. It follows that $J \leq \log_2 |\MF| + \max_{j \geq 1} |\MS^j| \leq 2 \bar R \cdot \log_2 |\MF| + N$. 
\end{proof}

Next, using \cref{lem:final-probability}, we give an expression for the expected number of times each ancestor of a given sequence $x \in \Sigma^n$ will occur in one of the multisets $\MS^j$. 
\begin{lemma}
  \label{lem:unbiased-path}
For each $x \in \Sigma^n$ and $j \geq 0$, it holds that $\left( \sum_{x' \in \MS^j} \frac{\One{x' \preceq x}}{\hat \mu(x')} \right)_{j \geq 0}$ is a martingale with respect to $(\SF^j)_{j \geq 0}$, and in particular
\begin{align}
  \E \left[ \sum_{x' \in \MS^j} \frac{\One{x' \preceq x}}{\hat \mu(x')} \right] = N\nonumber.
\end{align}
\end{lemma}
\begin{proof}
For a fixed $x \in \Sigma^n$, define
\begin{align}
\Phi_x(\mathcal T) := \sum_{x' \in \mathcal T} \frac{\One{x' \preceq x}}{\hat \mu(x')}\nonumber.
\end{align}
for any multiset $\mathcal T$ of particles. We claim that $(\Phi_x(\MS^j))_{j \geq 0}$ is a martingale with respect to $(\SF^j)_{j \geq 0}$.

Fix any $j \geq 1$. If $j > J$, then $\MS^j = \MS^{j-1}$ by definition, and there is nothing to prove. Suppose therefore that $1 \leq j \leq J$. The update from $\MS^{j-1}$ to $\MS^j$ removes one copy of $z_j$ from $\MS^{j-1}$ and then adds the multiset $\MT_j$ returned by \cref{alg:ts-rs}. If $z_j\not\preceq x$, then no element of $\MT_j$ can be an ancestor of $x$, and so $\Phi_x(\MS^j) = \Phi_x(\MS^{j-1})$. Otherwise, suppose that $z_j \preceq x$. Conditioned on $\SF^{j-1}$, there is a unique ancestor of $x$ that can belong to $\MT_j$, namely the unique ancestor of $x$ in $\Nfinal_{\MF^{j},\bar R}(z_j)$; denote this element by $w_j$. Thus, we have
\begin{align}
\E \left[ \sum_{w \in \MT_j} \frac{\One{w \preceq x}}{\hat \mu(w)} \mid  \SF^{j-1}\right] =& \E \left[ \frac{\occ(w_j, \MT_j)}{\hat \mu(w_j)} \mid \SF^{j-1}\right] \nonumber\\
=& \frac{\hat \mu(w_j)}{\hat \mu(z_j)} \cdot \frac{1}{\hat \mu(w_j)} = \frac{1}{\hat \mu(z_j)}\nonumber,
\end{align}
where the second equality uses \cref{eq:exp-occurrences} of \cref{lem:final-probability}. 

Hence, on the ($\SF^{j-1}$-measurable) event that $z_j \preceq x$, we have 
\begin{align}
\E\!\left[\Phi_x(\MS^j)-\Phi_x(\MS^{j-1}) \mid \SF^{j-1} \right]
= -\frac{1}{\hat \mu(z_j)} + \E \left[ \sum_{w \in \MT_j} \frac{\One{w \preceq x}}{\hat \mu(w)} \mid \SF^{j-1}\right] = 0.\nonumber
\end{align}

Thus $(\Phi_x(\MS^j))_{j \geq 0}$ is a martingale. At time $0$, the multiset $\MS^0$ consists of $N$ copies of the empty sequence, so by the convention $\hat \mu(\emptyset) = 1$ we have
\begin{align}
\Phi_x(\MS^0) = \frac{N}{\hat \mu(\emptyset)} = N.\nonumber
\end{align}
Combining this base case with the martingale property yields $\E[\Phi_x(\MS^j)] = N$ for every $j \geq 0$, which is exactly the claimed identity.
\end{proof}

As a corollary of \cref{lem:unbiased-path}, we can show that each $x \in \Sigma^n$ occurs in $\Sterm$ a number of times that is proportional to its true density $\mu^\star(x)$. 
\begin{lemma}[Unbiasedness]
  \label{lem:unbiasedness}
For each $x \in \Sigma^n$, it holds that $\E[\occ(x, \Sterm)] = N \cdot \mu^\st(x)$, where the expectation is over the randomness in \cref{alg:pf-learning}.
\end{lemma}
\begin{proof}
Fix $x \in \Sigma^n$. Note that $J$ is a stopping time for the filtration $(\SF^j)_{j\geq 0}$. The optional stopping theorem gives
\begin{align}
\E \left[ \sum_{x' \in \MS^J} \frac{\One{x' \preceq x}}{\hat \mu(x')} \right]
= \E[\Phi_x(\MS^J)] = \E[\Phi_x(\MS^0)] = N\nonumber.
\end{align}
By definition of $J$, every element of $\MS^J = \Sterm$ has length $n$. Since $x \in \Sigma^n$, the condition $x' \preceq x$ for $x' \in \MS^J$ is equivalent to $x' = x$. Together with the assumption that $\mu^\st(x) = \hat \mu(x)$ for $x \in \Sigma^n$, this gives
$
\E \left[ \sum_{x' \in \MS^J} \frac{\One{x' = x}}{\mu^\st(x')} \right] = N,
$
which yields the desired statement after rearranging.
\end{proof}

For each $j \geq 0$, we define $M_j := \sum_{x \in \MS^j} \frac{\mu^\star(x)}{\hat \mu(x)}$. 
The next lemma shows that this total reweighted mass is itself a martingale along the execution of the algorithm.
\begin{lemma}
  \label{lem:mi-mtgl}
$(M_j)_{j \geq 0}$ is a martingale with respect to the filtration $(\SF^j)_{j \geq 0}$, i.e., for every $j \geq 1$ we have $\E[M_j \mid \SF^{j-1}] = M_{j-1}$.
\end{lemma}
\begin{proof}
Fix any $j \geq 1$. If $j > J$, then $M_j = M_{j-1}$ by definition, so the claim is immediate. For each $1 \leq j \leq J$, we have 
\begin{align} 
M_j - M_{j-1} = \sum_{x' \in \MT_j} \frac{\mu^\st(x')}{\hat \mu(x')} - \frac{\mu^\st(z_j)}{\hat \mu(z_j)}\nonumber.
\end{align}
It therefore suffices to show
\begin{align}
\E \left[ \sum_{x' \in \MT_j} \frac{\mu^\st(x')}{\hat \mu(x')} \mid  \SF^{j-1} \right] = \frac{\mu^\st(z_j)}{\hat\mu(z_j)}.\label{eq:ij-mtgl}
\end{align}
Recall that $\MT_j \subset \Nfinal_{\MF^j, \bar R}(z_j)$ by \cref{lem:alg-bookkeeping}. Then we may compute
\begin{align}
\E \left[ \sum_{x' \in \MT_j} \frac{\mu^\st(x')}{\hat \mu(x')} \mid   \SF^{j-1} \right]  =& \sum_{x' \in \Nfinal_{\MF^j, \bar R}(z_j)} \frac{\mu^\st(x')}{\hat \mu(x')} \cdot \E[\occ(x', \MT_j) \mid \SF^{j-1}] \nonumber\\
=& \sum_{x' \in \Nfinal_{\MF^j, \bar R}(z_j)} \frac{\mu^\st(x')}{\hat \mu(x')} \cdot \frac{\hat \mu(x')}{\hat \mu(z_j)} \nonumber\\
=& \frac{1}{\hat \mu(z_j)} \sum_{x' \in \Nfinal_{\MF^j, \bar R}(z_j)} \mu^\st(x') = \frac{\mu^\st(z_j)}{\hat \mu(z_j)}\label{eq:mij-mtgl},
\end{align}
where the second equality uses \cref{eq:exp-occurrences} of \cref{lem:final-probability} and the final equality uses the fact that $\Nfinal_{\MF^j, \bar R}(z_j)$ is a simple $z_j$-cut (\cref{lem:n-cut}).
\end{proof}

Next, we use Azuma's inequality to bound the martingale $M_j$: 
\begin{lemma}
  \label{lem:mtgl-concentration}
Fix any $\delta \in (0,1)$. Then the martingale $(M_j)_{j \geq 0}$ satisfies %
\begin{align}
\BP \left(\max_{1 \leq j \leq J} |M_j - M_0| \leq  \frac{M_0}{2}\right) \geq 1-\frac{\delta}{2R}, \qquad
\E\left[\left(\frac{M_J}{2RN}-1\right)_+\right] \leq \frac{\delta}{2R} \nonumber. %
\end{align}
\end{lemma}
\begin{proof} 
Let $T := 2 \bar R \cdot \log_2 |\MF| + N$, so that \cref{lem:iteration-bound} ensures that $J \leq T$ with probability $1$. 
For $1 \leq j \leq J$, $|\MT_j| \leq \bar R$ almost surely and so $0 \leq \sum_{x' \in \MT_j} \frac{\mu^\st(x')}{\hat \mu(x')} \leq R \cdot \bar R$; thus $|M_j - M_{j-1}| \leq R \bar R$ for all $j$. Using the fact that a random variable $M$ taking values in $[0,A]$ satisfies $\log \E[\exp(\lambda \cdot (A - \E[A]))] \leq \lambda^2 A^2 / 8$, together with the fact that $\E[M_j - M_{j-1} \mid  \SF^{j-1}] = 0$ (which was established in \cref{lem:mi-mtgl}), we may conclude that for any $\lambda \in \BR$,
\begin{align}
\log \E \left[ \exp(\lambda \cdot (M_j - M_{j-1})) \mid \SF^{j-1}\right] \leq \lambda^2 (R \bar R)^2\label{eq:mgf-bound-for-freedman}.
\end{align}
Since $M_j=M_J$ for all $j \geq J$, iterating the preceding conditional moment-generating function bound up to the deterministic time $T$ gives, for every $\lambda \geq 0$,
\begin{align}
\log \E \left[\exp(\lambda \cdot (M_J-M_0))\right] \leq T\lambda^2(R\bar R)^2.\nonumber
\end{align}
Therefore, for every $s > 0$,
$
\BP(M_J-M_0 > s) \leq \exp\left(-\frac{s^2}{4T(R\bar R)^2}\right).
$
Using $M_0=N$ and $R \geq 1$, it follows that
\begin{align}
\E\left[\left(\frac{M_J}{2RN}-1\right)_+\right]
&= \frac{1}{2RN}\int_{(2R-1)N}^{\infty} \BP(M_J-M_0 > s)\,ds \nonumber\\
&\leq \frac{1}{2RN}\int_N^{\infty} \exp\left(-\frac{s^2}{4T(R\bar R)^2}\right)\,ds \nonumber\\
&\leq \frac{T(R\bar R)^2}{RN^2}\exp\left(-\frac{N^2}{4T(R\bar R)^2}\right) \leq \frac{\delta}{2R},\nonumber
\end{align}
where the final inequality follows from the choice of $N$ in \cref{alg:pf-learning}, after increasing the absolute constant $C_{\N}$ if necessary.

Moreover, using \cref{eq:mgf-bound-for-freedman}, by \cref{lem:freedman} and a union bound over $j \in [T]$ it follows that for any $\lambda \in \BR$, with probability at least $1-\delta/(4R)$, for all $1 \leq j \leq J$, we have
\begin{align}
\lambda \cdot (M_j - M_0) \leq j \cdot \lambda^2 (R \bar R)^2 + \log(4RT/\delta)\nonumber.
\end{align}
Rearranging and choosing $\lambda =\pm\frac{\sqrt{\log (4RT/\delta)}}{\sqrt{T} \cdot R \bar R}$ gives that, with probability $1-\delta/2R$,
\begin{align}
\max_{1 \leq j \leq J} |M_j - M_0| \leq & |\lambda |\cdot j (R \bar R)^2 + \frac{\log (4RT/\delta)}{|\lambda|} \leq O \left( R \bar R \cdot \sqrt{T \log (R\log_2(|\MF|) N/\delta)} \right) \nonumber\\
\leq & O \left(R \bar R \sqrt{(\bar R \cdot \log_2 |\MF| + N) \cdot \log (R\log_2(|\MF|) N/\delta)} \right) \leq \frac{N}{2} = \frac{M_0}{2}\nonumber,
\end{align}
where the final inequality uses the choice of $N$ in \cref{alg:pf-learning} for an appropriately large constant $C_{\N}$. 
\end{proof}

Next, we convert the preceding unbiasedness and concentration bounds for the martingale $M_j$ into correctness of the output distribution.
\begin{lemma}
  \label{lem:right-sample-distribution}
For any $\delta \in (0,1/2)$, \cref{alg:pf-learning} with the parameter $\delta$ outputs a sample from a distribution $\mualg$ satisfying $\tvd{\mu^\star}{\mualg} \leq \delta$. 
\end{lemma}
\begin{proof}
Fix one outer-loop iteration (i.e., over $k \in [K]$). Let $X$ denote the random variable which is equal to the sample returned in that iteration if the final Bernoulli step succeeds, and is equal to $\bot$ otherwise. By \cref{lem:unbiasedness}, we have, for any set $\MA \subseteq \Sigma^n$, 
\begin{align}
\BP(X\in \MA ) =  \E \left[ \sum_{x \in \MA} \frac{\occ(x, \Sterm)}{|\Sterm|} \cdot \min \left\{ 1, \frac{|\Sterm|}{2RN} \right\} \right] \leq \frac{1}{2R} \cdot \mu^\star(\MA).\label{eq:single-round-output-law}
\end{align}
For the lower bound, note that the pointwise inequality
\begin{align}
\sum_{x \in \MA} \frac{\occ(x, \Sterm)}{|\Sterm|} \cdot \min \left\{1, \frac{|\Sterm|}{2RN}\right\}
\geq \frac{1}{2RN}\sum_{x \in \MA}\occ(x,\Sterm) - \left(\frac{|\Sterm|}{2RN}-1\right)_+ \nonumber
\end{align}
holds for every realization of $\Sterm$, where the left-hand side is interpreted as $0$ if $|\Sterm|=0$. Therefore,
\begin{align}
\BP(X \in \MA)
&\geq \frac{1}{2RN}\E\left[\sum_{x \in \MA} \occ(x,\Sterm)\right]
- \E\left[\left(\frac{|\Sterm|}{2RN}-1\right)_+\right] \nonumber\\
&\geq \frac{1}{2R} \cdot \mu^\st(\MA) - \frac{\delta}{2R}.\label{eq:single-round-output-lb}
\end{align}
It follows that $\frac{1}{2R} - \frac{\delta}{2R}\leq \BP(X \neq \perp) \leq \frac{1}{2R}$, and so
\begin{align}
\mu^\st(\MA) - \delta\leq \BP(X \in \MA \mid X \neq \perp) \leq \mu^\st(\MA) \cdot \frac{1}{1-\delta} \leq \mu^\st(\MA) \cdot (1 + \delta)\label{eq:mustar-close}.
\end{align}
Let $\tilmualg \in \Delta(\Sigma^n)$ be the distribution defined by $\tilmualg(x) := \BP(X = x \mid X \neq \perp)$. Then \cref{eq:mustar-close} gives that $\tvd{\mu^\st}{\tilmualg} \leq \delta$. By our choice of $K = \lceil 4R \log_2(1/\delta) \rceil$ and the fact that each outer-loop iteration produces a symbol which is not $\perp$ with probability at least $\frac{1}{2R} - \frac{\delta}{2R} \geq \frac{1}{4R}$, with probability at least $1-\delta$, some outer loop iteration will produce a symbol which is not $\perp$, and which is therefore distributed according to $\tilmualg$. It follows that $\tvd{\mu^\st}{\mualg} \leq 2\delta$; the claimed result follows by rescaling $\delta$.
\end{proof}

The final lemma uses the oracle complexity bound of \cref{lem:query-cost} to derive the overall oracle complexity guarantee for \cref{alg:pf-learning}.
\begin{lemma}
  \label{lem:oracle-calls}
The number of oracle calls to $\hat \mu$ made by \cref{alg:pf-learning} is upper bounded by
\begin{align}
O\left(n \cdot |\Sigma| \cdot K \cdot R \bar R^2 \cdot \left\lceil \sqrt{\log |\MF|} \cdot R^2 \bar R^2 \left( \log \frac{R \log |\MF|}{\delta}\right) \right\rceil \right)
= O\left(n \cdot |\Sigma| \cdot R^{12} \cdot \sqrt{\log |\MF|} \cdot \log \frac{1}{\delta} \cdot \left(\log \frac{R \log |\MF|}{\delta}\right) \right)\nonumber
\end{align}
with probability $1-\delta$.
\end{lemma}

\begin{proof}
For each $j \geq 1$, we let the random variable $V_j$ denote the number of oracle calls to $\hat \mu$ made by \cref{alg:ts-rs} in Line \ref{line:rs-call} of \cref{alg:pf-learning} during step $j$ of the algorithm. By \cref{lem:query-cost}, we have
\begin{align}
\E[V_j \mid \SF^{j-1}] \leq O \left(|\Sigma| \cdot R \cdot \bar R^2 \right) \cdot \E \left[ \sum_{x \in \MT_j} \frac{\mu^\st(x)}{\hat \mu(x)} \cdot |x| - \frac{\mu^\st(z_j)}{\hat \mu(z_j)} \cdot |z_j| \mid \SF^{j-1} \right]\label{eq:vj-update}.
\end{align}
Next, we define a random variable
\begin{align}
W_j = \sum_{x \in \MS_j}\frac{\mu^\st(x)}{\hat \mu(x)} \cdot |x|\nonumber.
\end{align}
Note that
\begin{align}
W_j = W_{j-1} - \frac{\mu^\st(z_j)}{\hat \mu(z_j)} \cdot |z_j| + \sum_{x \in \MT_j} \frac{\mu^\st(x)}{\hat \mu(x)} \cdot |x| \label{eq:wj-update},
\end{align}
so that \cref{eq:vj-update,eq:wj-update} give
\begin{align}
\E[V_j \mid \SF^{j-1}] \leq O \left(|\Sigma| \cdot R \cdot \bar R^2 \right) \cdot \E[W_j - W_{j-1} \mid \SF^{j-1}].\label{eq:vj-wj}
\end{align}
It follows that 
\begin{align}
\E \left[ \sum_{j=1}^J V_j \right] \leq O(|\Sigma| R \bar R^2) \cdot \E[W_J] = O(n \cdot |\Sigma| R \bar R^2) \cdot \E[M_J] =  O(Nn \cdot |\Sigma| R \bar R^2)\nonumber.
\end{align}
Let $T := 2 \bar R \cdot \log_2 |\MF| + N$, so that $J \leq T$ almost surely by \cref{lem:iteration-bound}. Fix one outer-loop iteration of \cref{alg:pf-learning}. For each $j \in [T]$, define
$
\Gamma_j := V_j - \E[V_j \mid \SF^{j-1}], 
$
where for $j > J$ we interpret $V_j := 0$. Then $(\sum_{s=1}^j \Gamma_s)_{j \in [T]}$ is a martingale. By the almost-sure bound in \cref{lem:query-cost},
$
|\Gamma_j|  \leq O(|\Sigma| \bar R n)
$
for every $j$. Applying Azuma's inequality, with probability at least $1-\delta/(4K)$ we have
\begin{align}
\sum_{j=1}^J V_j \leq \sum_{j=1}^J \E[V_j \mid \SF^{j-1}] + O\left(|\Sigma| \bar R n \sqrt{T \log \frac{K}{\delta}}\right).\label{eq:oracle-calls-v-azuma}
\end{align}
Next, define
\begin{align}
\Delta_j := (W_j - W_{j-1}) - \E[W_j - W_{j-1} \mid \SF^{j-1}],\nonumber
\end{align}
where for $j > J$ we interpret $W_j := W_J$. Then $(\sum_{s=1}^j \Delta_s)_{j \in [T]}$ is a martingale. Since $|\MT_j| \leq \bar R$ almost surely and $\mu^\st(x)/\hat \mu(x) \leq R$ for every $x$, we have
\begin{align}
|W_j - W_{j-1}| \leq Rn(\bar R + 1) = O(R \bar R n)\nonumber
\end{align}
for every $j$, and hence $|\Delta_j| \leq O(R \bar R n)$. Applying Azuma's inequality, with probability at least $1-\delta/(4K)$ we therefore have
\begin{align}
\sum_{j=1}^J \E[W_j - W_{j-1} \mid \SF^{j-1}] \leq W_J - W_0 + O\left(R \bar R n \sqrt{T \log \frac{K}{\delta}}\right).\label{eq:oracle-calls-w-azuma}
\end{align}
Since $W_0 = 0$, \cref{eq:vj-wj,eq:oracle-calls-w-azuma} imply that, with probability at least $1-\delta/(4K)$,
\begin{align}
\sum_{j=1}^J \E[V_j \mid \SF^{j-1}] \leq O(|\Sigma| R \bar R^2) \cdot W_J + O\left(|\Sigma| R^2 \bar R^3 n \sqrt{T \log \frac{K}{\delta}}\right).\label{eq:oracle-calls-predictable}
\end{align}
Since
\begin{align}
N = C_{\N} \left\lceil \sqrt{\log |\MF|} \cdot R^2 \bar R^2 \left( \log \frac{R \log |\MF|}{\delta}\right) \right\rceil\nonumber
\end{align}
and by \cref{lem:mtgl-concentration} we have, after increasing $C_{\N}$ if needed, that $M_J \leq 3N/2$ with probability at least $1-\delta/(4K)$, it follows that with the same probability
$
W_J \leq n M_J \leq 3nN/2.
$
Combining this with \cref{eq:oracle-calls-v-azuma,eq:oracle-calls-predictable} yields that, with probability at least $1-3\delta/(4K)$,
\begin{align}
\sum_{j=1}^J V_j \leq O(|\Sigma| R \bar R^2) \cdot nN + O\left(|\Sigma| R^2 \bar R^3 n \sqrt{T \log \frac{K}{\delta}}\right) + O\left(|\Sigma| \bar R n \sqrt{T \log \frac{K}{\delta}}\right).\nonumber
\end{align}
Using $T = O(\bar R \log |\MF| + N)$ together with $K = O(R \log(1/\delta))$, the choice of $N$ ensures that the second and third terms are both absorbed into the main term for a sufficiently large absolute constant $C_{\N}$. Therefore, one outer-loop iteration uses at most
\begin{align}
O\left(n \cdot |\Sigma| \cdot R \bar R^2 \cdot \left\lceil \sqrt{\log |\MF|} \cdot R^2 \bar R^2 \left(\log \frac{R \log |\MF|}{\delta}\right) \right\rceil \right)\nonumber
\end{align}
oracle calls with probability at least $1-\delta/K$. Summing over the $K$ outer-loop iterations and using a union bound proves that the total number of oracle calls made by \cref{alg:pf-learning} is at most
\begin{align}
O\left(n \cdot |\Sigma| \cdot K \cdot R \bar R^2 \cdot \left\lceil \sqrt{\log |\MF|} \cdot R^2 \bar R^2 \left(\log \frac{R \log |\MF|}{\delta}\right) \right\rceil \right)\nonumber,
\end{align}
with probability at least $1-\delta$. Using $\bar R = 2R^2$ and $K = \lceil 4R\log_2(1/\delta) \rceil$ gives the second displayed bound in the lemma statement.
\end{proof}

Finally, we are ready to prove \cref{thm:main-ub-formal}. 
\begin{proof}[Proof of \cref{thm:main-ub-formal}]
Run \cref{alg:pf-learning} with failure tolerance $\delta/2$. By \cref{lem:right-sample-distribution}, its output distribution $\hat\nu'$ satisfies $\tvd{\mu^\star}{\hat\nu'} \leq \delta/2$. By \cref{lem:oracle-calls}, with probability at least $1-\delta/2$, the number of oracle calls made by the algorithm is at most
\begin{align}
O\left(n \cdot |\Sigma| \cdot R^{12} \cdot \sqrt{\log |\MF|} \cdot \log \frac{1}{\delta} \cdot \left(\log \frac{R \log |\MF|}{\delta}\right) \right)\label{eq:final-oracle-ub}.
\end{align}
We can modify \cref{alg:pf-learning} so that, whenever its number of oracle calls reaches the threshold in \cref{eq:final-oracle-ub}, it aborts and outputs an arbitrary symbol. This changes its output distribution by at most $\delta/2$ in total variation distance, meaning that the resulting output distribution, which we denote by $\hat \nu$, satisfies $\tvd{\mu^\star}{\hat\nu} \leq \delta$.
\end{proof}

\section{Lower bound}
\label{sec:lower-bound}
In this section, we prove our main lower bound, \cref{thm:weak-lb}. 
\begin{theorem}
  \label{thm:weak-lb}
There is a constant $c > 0$ so that the following holds, for any sufficiently large $n \in \BN$ and $\gamma \in (1/n,1)$. There is a class of distributions $\MF \subset \Delta(\{0,1\}^n)$ of size $ |\MF| \leq \exp(O(n^4))$ so that no randomized algorithm $\Alg$ enjoys the following guarantee: for any distribution $\mu^\star \in \MF$ and $\hat \mu : \{0,1\}^{\leq n} \to \BR_{\geq 0}$ satisfying \cref{asm:linfty-close-gen} with respect to $\mu^\star$ for $R = (1+\gamma)^2$, $\Alg$ makes $\frac{c \gamma^2 n^2}{ \log^2 (n) \cdot \log^2(2/\gamma)}$ queries to $\hat \mu$ and outputs a (random) string $X \in \{0,1\}^n$, according to some distribution $\hat \nu$, satisfying $\tvd{\mu^\star}{\hat \nu} \leq \frac{\gamma}{8}$.
\end{theorem}
The remainder of the section is organized as follows: in \cref{sec:lb-term}, we introduce some terminology and notation that will be used throughout the section. In \cref{sec:lb-construction}, we describe the construction of the family $\MF$ of distributions. In \cref{sec:lb-proof}, we analyze the construction and prove \cref{thm:weak-lb}. Finally, in \cref{sec:add-lbs}, we describe a couple of variants of \cref{thm:weak-lb} that give stronger or generalized lower bounds. 

\subsection{Terminology and setup}
\label{sec:lb-term}
Fix a positive integer $n$, and positive integers $k,r$ so that $n = 2kr + 1$. (Ultimately, we will extend the result to values of $n$ which cannot be written in this form via a padding argument.) 
Fix some $\ep \in (0, 1/2)$ (we will eventually choose $\ep \sim \poly\log(n)/n$). Write $\MV_k := \{0,1\}^{\{0,1\}^k}$ and $\MU_k = (\{0,1\}^k)^{\{0,1\}^k}$. We will interpret elements of $\MV_k, \MU_k$ as vectors whose elements are indexed by $\{0,1\}^k$; further, for $v \in \MV_k$, we will write $\| v \|_1 = \sum_{x \in \{0,1\}^k} v_x$. 

We will be concerned with tuples $(v, \tilde v) \in \MV_k^2$ which are valid in the sense captured by the below definition:
\begin{definition}
\label{def:validity}
We say that a tuple $(v, \tilde v) \in \MV_k^2$ is \emph{$\ep$-valid} if the following conditions hold: 
\begin{itemize}
\item $v$ satisfies $\| v \|_1= 2^{k-1}$.
\item $\tilde v$ satisfies $\| \tilde v \|_1 = 2^{k-1} - \ep \cdot 2^k$. 
\item $\tilde v \leq v$ entrywise (i.e., $v_x = 0$ implies $\tilde v_x = 0$ for $x \in \{0,1\}^k$). 
\end{itemize}
We let the set of valid tuples be denoted by $\Vvalid_{k,\ep}$. 
\end{definition}
Given a tuple $(v,\tilde v) \in \Vvalid_{k,\ep}$, it corresponds to a partition of $\{0,1\}^k$ into three sets, namely $A = \{ z \mid v_z =\tilde v_z = 1\}, B = \{ z \mid v_z = 1, \tilde v_z = 0\}, C = \{ z \mid v_z = 0\}$. The $\ep$-validity condition ensures that $|A| = 2^{k-1} - \ep \cdot 2^k$, $|B| = \ep \cdot 2^k$, and $|C| = 2^{k-1}$. Such partitions are of the form of those discussed in \cref{sec:tech-lowerbound}. 

Moreover, elements of $\MU_k$ are to be interpreted as elements of the ``key'' blocks discussed in \cref{sec:tech-lowerbound}. In particular, they will be used to ensure that any algorithm does not ``skip ahead'' and query $\hat \mu(x)$ for certain $x$ before querying $\hat \mu(x')$ at various prefixes of $x$. 

Recall that a family $\MH$ of functions $f : \MX \to \MY$ (for finite sets $\MX, \MY$) is defined to be \emph{$q$-wise independent} if for any choice of distinct $x_1, \ldots, x_q \in \MX$, and a uniformly random draw of $h \sim \mathrm{Unif}(\MH)$, the random variables $h(x_1), \ldots, h(x_q)$ are independent and uniformly distributed in $\MY$. 
We will be concerned with $q$-wise independent mappings from $\{0,1\}^{\leq n}$ to valid tuples (and also to elements of $\MU_k$), formalized below: 
\begin{definition}
\label{def:hash-functions}
We consider the following families of functions: 
\begin{itemize}
\item Let $\Phi_{n,k,q,\ep}$ be a $q$-wise independent family of functions $\phi : \{0,1\}^{\leq n} \to \Vvalid_{k,\ep}$, of size $\exp(O(q \cdot \max \{ n, 2^k\}))$. %
\item Let $\Psi_{n,k,q}$ be a $q$-wise independent family of functions $\psi : \{0,1\}^{\leq n} \to \MU_k$, of size $\exp(O(q \cdot \max \{ n, k2^k\}))$. %
\item We set $\MH_{n,k,q,\ep} := \{0,1\} \times \Phi_{n,k,q,\ep} \times \Psi_{n,k,q}$; we will denote elements of $\MH_{n,k,q,\ep}$ as tuples $(\vinit, \phi, \psi)$, where $\vinit \in \{0,1\}$, $\phi \in \Phi_{n,k,q,\ep}$, and $\psi \in \Psi_{n,k,q}$.
\end{itemize}
\end{definition}
For $\phi \in \Phi_{n,k,q,\ep}$ and $z \in \{0,1\}^k$, we will write $\phi(x;z) := (\phi(x)_1)_z \in \{0,1\}$ and $\tilde \phi(x;z) := (\phi(x)_2)_z \in \{0,1\}$ for $x \in \{0,1\}^{\leq n}$. (Recall that $\phi(x)$ is a tuple $(v,\tilde v)$; thus $\phi(x;z)$ uses the first coordinate $v$ and $\tilde \phi(x;z)$ uses the second coordinate $\tilde v$.) For $\psi \in \Psi_{n,k,q}$ and $z\in \{0,1\}^k$, we will write $\psi(x; z) = \psi(x)_z \in \{0,1\}^k$.

The below lemma, which follows from the construction of $q$-wise independent families via polynomials, is a standard fact. 
\begin{lemma}[Corollary 3.34 of \cite{vadhan2011foundations}]
\label{lem:kwise-indep-existence}
There exist families $\Phi_{n,k,q,\ep}, \Psi_{n,k,q}$ satisfying the conditions of \cref{def:hash-functions}.
\end{lemma}

Next, for $\phi \in \Phi_{n,k,q,\ep}$, and $0 \leq a \leq r-1$ and $x \in \{0,1\}^{2ak+1}$, we define the following partitions depending on $\phi(x)$, as alluded to above: 
\begin{align}
A(x) &:= \{z \in \{0,1\}^k :\tilde \phi(x; z)=1\},\label{eq:A}\\
B(x) &:= \{z \in \{0,1\}^k : \phi(x;z) \text{ and } \tilde \phi(x;z)=0\},\label{eq:B}\\
C(x) &:= \{z \in \{0,1\}^k : \phi(x;z)=0\}.\label{eq:C}
\end{align}
By $\ep$-validity of $\phi(x)$, these sets form a partition of $\{0,1\}^k$, and moreover
\begin{align}
|A(x)| = 2^{k-1}-\ep 2^k,\qquad |B(x)| = \ep 2^k,\qquad |C(x)| = 2^{k-1}.\label{eq:ABC-sizes-initial-lemmas}
\end{align}

Finally, we introduce the following notation which ``breaks'' strings $x \in \{0,1\}^{\leq n}$ into chunks: for $x \in \{0,1\}^i$, for $1 \leq i \leq 2rk+1$, we define:
\begin{itemize}
\item For $0 \leq a \leq r-1$, $\query_a(x) := x_{2 + 2ak : 1 + 2ak + k} \in \{0,1\}^{\leq k}$.
\item For $0 \leq a \leq r-1$, $\key_a(x) := x_{2ak + k + 2 : 2ak + 2k + 1} \in \{0,1\}^{\leq k}$.
\item For $0 \leq a \leq r-1$, $\base_a(x) := x_{1:2ak+1} \in \{0,1\}^{\leq 2ak+1}$. We will only consider $\base_a(x)$ for $a$ satisfying $2ak+1\leq |x|$. 
\end{itemize}

\subsection{The lower bound construction}
\label{sec:lb-construction}
Fix integers $k,n,r$ as above, a real number $\gamma \in (0,1)$, and a tuple $\tau := (\vinit, \phi, \psi) \in \MH_{n,k,q,\ep}$. We will define a distribution $\mu = \mu_{\tau} \in \Delta(\{0,1\}^n)$ as well as a mapping $\hat \mu = \hat \mu_{\tau} : \{0,1\}^{\leq n} \to \BR_{\geq 0}$, which depend on $\tau$. (To avoid clutter, we suppress the dependence on $\tau$ in the below.) We will do so as follows: we will first define mappings $\omega = \omega_{\tau}, \hat \omega  = \hat \omega_{\tau}: \{0,1\}^{\leq n} \to \BR_{ \geq 0}$. One should think of $\omega, \hat \omega$ as specifying ``edge weights'' for the complete binary tree of depth $n$, and then we will define $\mu, \hat \mu$ so as to evaluate the (appropriately normalized) product of the edge weights on the path from each node to the root. Our mappings $\omega,\hat\omega$, defined below, will have range equal to $\{0,1, 1+\gamma \}$. 

\paragraph{Defining $\omega$.} First, for $x \in \{0,1\}$, we set $\omega(x) = 1 + \One{\vinit = x} \cdot \gamma$. Now consider $x \in \{0,1\}^i$, for $2 \leq i \leq 2rk+1$. We consider the following cases:
\begin{itemize}
\item If $i \pmod{2k} \in \{2, 3, \ldots, k+1 \}$, then we set $\omega(x) = 1$. 
\item If $i \pmod{2k} =: p \in \{k+2,k+3, \ldots, 2k\}$, then, for $a := \lfloor (i-2)/(2k) \rfloor$, we set $\omega(x) = \One{\key_a(x) = \psi(\base_a(x); \query_a(x))_{1:|\key_a(x)|}}$.
\item If $i =1 \pmod{2k}$, we set 
\begin{align}
\omega(x) = \One{\key_a(x) = \psi(\base_a(x); \query_a(x))}\cdot (1 + \gamma \cdot \phi(\base_a(x);\query_a(x)),\label{eq:define-omega}
\end{align}
 where $a := \lfloor (i-2)/(2k) \rfloor$.
\end{itemize}

\paragraph{Defining $\hat\omega$.} First, we set $\hat\omega(1) = \hat\omega(0) =1+\gamma$. Now consider $x \in \{0,1\}^i$, for $2 \leq i \leq 2rk+1$. We consider the following cases:
\begin{itemize}
\item If $i \pmod{2k} \in \{2, 3, \ldots, k+1 \}$, then we set $\hat\omega(x) = 1$. 
\item If $i \pmod{2k} =: p \in \{k+2,k+3, \ldots, 2k\}$, then, for $a := \lfloor (i-2)/(2k) \rfloor$, we set $\hat\omega(x) = \One{\key_a(x) = \psi(\base_a(x); \query_a(x))_{1:|\key_a(x)|}}$.
\item If $i = 1\pmod{2k}$, write $a = \lfloor (i-2)/(2k) \rfloor$, so that $i=2(a+1)k+1$. If either $\vinit = x_1$, or there is some $0 \leq b < a$ such that
\begin{align}
\phi(\base_b(x);\query_b(x)) \neq \tilde \phi(\base_b(x);\query_b(x)),\label{eq:prior-b-exists}
\end{align}
then we set
\begin{align}
\hat\omega(x) = \One{\key_a(x) = \psi(\base_a(x); \query_a(x))}  \cdot (1 + \gamma \cdot \phi(\base_a(x);\query_a(x))).\label{eq:define-hatomega-1}
\end{align}
Otherwise, we set
\begin{align}
\hat\omega(x) = \One{\key_a(x) = \psi(\base_a(x); \query_a(x))} \cdot (1 + \gamma \cdot \tilde \phi(\base_a(x);\query_a(x))).\label{eq:define-hatomega-2}
\end{align}
\end{itemize}

\paragraph{Defining $\hat \mu$ and $\mu$.} For each $i \in [n]$, define
\begin{align}
\hat Z_i = \sum_{x' \in \{0,1\}^{i}} \prod_{j=1}^{i} \hat \omega(x'_{1:j}).\label{eq:def-zhati}
\end{align}
For $x \in \{0,1\}^{i}$, we set
\begin{align}
\hat \mu(x) = \hat Z_i^{-1} \cdot \prod_{j=1}^{i} \hat \omega(x_{1:j}).\label{eq:def-muhat-zhat}
\end{align}
For convenience, we write $\hat Z := \hat Z_n$.

Next, we define $\mu \in \Delta(\{0,1\}^n)$ by $\mu(x) = \hat \mu(x)$ for $x \in \{0,1\}^n$; by virtue of the division by $\hat Z$ in the definition of $\hat \mu$, we have that $\mu$ is indeed a distribution. 

For analysis purposes, we also define a distribution $\muideal \in \Delta(\{0,1\}^n)$ as follows: for $x \in \{0,1\}^n$,
\begin{align}
\muideal(x) =  Z^{-1} \cdot \prod_{j=1}^n\omega(x_{1:j}), \qquad Z = \sum_{x \in \{0,1\}^n} \prod_{j=1}^n \omega(x_{1:j}) = (2+\gamma) \cdot \left( 2^k + \gamma \cdot 2^{k-1}\right)^r.\nonumber
\end{align}

\subsection{Proof of the lower bound}
\label{sec:lb-proof}
We first show that the distribution $\mu$ is close to $\muideal$ as long as $\ep \cdot r$ is sufficiently large.
\begin{lemma}
  \label{lem:mu-muideal}
It holds that $\tvd{\mu}{\muideal} \leq e^{-\ep r}$.
\end{lemma}
\begin{proof}
For a leaf $x \in \{0,1\}^n$, define
\begin{align}
W(x) := \prod_{j=1}^n \omega(x_{1:j}), \qquad \hat W(x) := \prod_{j=1}^n \hat\omega(x_{1:j}).\nonumber
\end{align}
Then $\mu(x) = \hat\mu(x) = \hat W(x)/\hat Z$ and $\muideal(x) = W(x)/Z$ for $x \in \{0,1\}^n$. 

Let $\mathcal E \subseteq \{0,1\}^n$ be the set of leaves $x$ such that $x_1=1-\vinit$ and
\begin{align}
\query_a(x) \notin B(\base_a(x)) \qquad \text{for every } 0 \leq a \leq r-1.\nonumber
\end{align}
We claim that for every leaf $x \in \{0,1\}^n$,
\begin{align}
\hat W(x) = \bigl(1+\gamma \cdot \One{x \in \mathcal E}\bigr) \cdot W(x).\label{eq:What-over-W}
\end{align}
Indeed, if $x_1=\vinit$, then $x \not \in \ME$ and the defining condition for $\hat\omega$ always selects $\phi$, so $\hat\omega$ and $\omega$ agree at every prefix of $x$, including the root. Hence $\hat W(x)=W(x)$. 
Now suppose that $x_1=1-\vinit$. At the root we have $\hat\omega(x_1)=1+\gamma$ whereas $\omega(x_1)=1$. If $x \in \mathcal E$, then for every index $a < r$ we have $\query_a(x) \in A(\base_a(x)) \cup C(\base_a(x))$, so
\begin{align}
\phi(\base_a(x);\query_a(x)) = \tilde\phi(\base_a(x);\query_a(x)).\nonumber
\end{align}
Thus, by the definitions in \cref{eq:define-hatomega-1,eq:define-hatomega-2,eq:define-omega}, $\hat \omega(x_{1:i}) = \omega(x_{1:i})$ for all $i \geq 2$, and so $\hat W(x) = (1+\gamma) \cdot W(x)$. 
Finally, if $x \notin \mathcal E$, and let $a_\star$ be the least index in $\{0,\ldots,r-1\}$ such that
$
\query_{a_\star}(x) \in B(\base_{a_\star}(x)).
$
The definitions of $\omega,\hat\omega$ give that $\hat \omega(x_{1:i}) = \omega(x_{1:i})$ for all $i \geq 2$ with $i \neq 2(a_\star +1)k+1$, and for $i = 2(a_\star+1)k+1$, we have $\hat \omega(x_{1:i}) \cdot (1+\gamma) = \omega(x_{1:i})$. Then we again have $\hat W(x) = W(x)$. 

Let $p := \muideal(\mathcal E)$. Summing \cref{eq:What-over-W} over all leaves shows that
\begin{align}
\hat Z = \sum_{x \in \{0,1\}^n} \hat W(x) = \sum_{x \in \{0,1\}^n} W(x) + \gamma \sum_{x \in \ME}W(x) = Z \cdot (1 + \gamma p).\label{eq:zhat-z-comparison}
\end{align}
Using that $\mu(x) = \hat\mu(x)$ for $x \in \{0,1\}^n$, it follows that
\begin{align}
\mu(x) = \frac{1+\gamma \cdot \One{x \in \mathcal E}}{1+\gamma p} \cdot \muideal(x).\label{eq:mu-vs-muideal-pointwise}
\end{align}
Therefore
\begin{align}
\tvd{\mu}{\muideal}
&= \frac12 \sum_{x \in \{0,1\}^n} \muideal(x) \cdot \left|\frac{1+\gamma \cdot \One{x \in \mathcal E}}{1+\gamma p} - 1\right|\nonumber\\
&= \frac{\gamma p(1-p)}{1+\gamma p} \leq p.\label{eq:tvd-by-p}
\end{align}

It remains to bound $p$. Under $\muideal$, the probability of taking the root edge $x_1 = 1-\vinit$ is $1/(2+\gamma)$, because the two root subtrees have identical total $W$-mass below the root and the two root weights are $1$ and $1+\gamma$. Next, fix any $0 \leq a \leq r-1$ and any $y \in \{0,1\}^{2ak+1}$ with $y_1 = 1-\vinit$ and so that for all $b < a$, $\phi(\base_b(y);\query_b(y)) = \tilde\phi(\base_b(y);\query_b(y))$. For each $z \in \{0,1\}^k$, there is a unique ``key'' $\kappa(z) \in \{0,1\}^k$ for which $\key_a(y \circ z \circ \kappa(z)) = \psi(y;z)$. Moreover, we have
\begin{align}
\frac{W(y \circ z \circ \kappa(z))}{W(y)} = \prod_{j=2ak+2}^{2(a+1)k+1} \omega((y \circ z \circ \kappa(z))_{1:j})
= 1 + \gamma \cdot \phi(y;z).\nonumber
\end{align}
Therefore, under $x \sim \muideal$, conditioned on $x_{1:2ak+1} = y$, the conditional probability that $\query_a(x) \notin B(y)$ is
\begin{align}
\frac{(1+\gamma)|A(y)| + |C(y)|}{(1+\gamma)(|A(y)|+|B(y)|) + |C(y)|}
&= 1 - \frac{(1+\gamma)|B(y)|}{2^k + \gamma 2^{k-1}}\nonumber\\
&= 1 - \frac{1+\gamma}{1+\gamma/2} \cdot \ep\nonumber\\
&\leq 1-\ep,\label{eq:no-B-probability}
\end{align}
where we used \cref{eq:ABC-sizes-initial-lemmas}. Since the right-hand side of \cref{eq:no-B-probability} is uniform over all such prefixes $y$, we obtain
\begin{align}
p = \muideal(\mathcal E) \leq \frac{1}{2+\gamma} \cdot (1-\ep)^r \leq e^{-\ep r}.\nonumber
\end{align}
Combining this with \cref{eq:tvd-by-p} proves the lemma.
\end{proof}

For $x \in \{0,1\}^{\leq n}$, recall that we let $\mu(x)$ denote the marginal distribution of $\mu$ on $x$ (i.e., on its first $|x|$ coordinates). The next lemma shows that, for an arbitrary choice of $\tau \in \MH_{n,k,q,\ep}$, the resulting values of $\mu = \mu_\tau$ and $\hat \mu = \hat \mu_\tau$ satisfy \cref{asm:linfty-close-gen} with $R = (1+\gamma)^2$. 
\begin{lemma}
  \label{lem:mu-muhat-ratio}
For all $x \in \{0,1\}^{\leq n}$, either $\mu(x) = \hat \mu(x) = 0$, or else it holds that 
\begin{align}
\max\{ \mu(x)/\hat \mu(x), \hat \mu(x)/\mu(x) \} \leq (1 + \gamma)^2.\nonumber
\end{align}
\end{lemma}
\begin{proof}
For $x \in \{0,1\}^{\leq n}$, write
\begin{align}
P(x) := \prod_{j=1}^{|x|} \hat\omega(x_{1:j}),\qquad T(x) := \sum_{z \in \{0,1\}^{m-|x|}} \prod_{j=|x|+1}^{m} \hat\omega((x\circ z)_{1:j}).\nonumber
\end{align}
Then
\begin{align}
\hat\mu(x) = \frac{P(x)}{\hat Z_{|x|}},\qquad \mu(x) = \frac{P(x)\,T(x)}{\hat Z}, \qquad \frac{\mu(x)}{\hat \mu(x)} = \frac{T(x) \hat Z_{|x|}}{\hat Z}.\label{eq:muhat-mu-factorization}
\end{align}
In particular, $\hat\mu(x)=0$ if and only if $P(x)=0$. Moreover, for $x \in \{0,1\}^{\leq n}$, whenever $P(x)>0$ there is at least one leaf descendent of $x$ with positive $\hat\omega$-mass, and therefore $\mu(x)>0$ if and only if $\hat\mu(x)>0$.

We next bound the possible values of $T(x)$ at depths $1 \pmod{2k}$. For $0 \leq a \leq r$ and $x \in \{0,1\}^{2ak+1}$, we say that $x$ is \emph{switched} if either $x_1 = \vinit$, or there is some $0 \leq b < a$ such that
$
\query_b(x) \in B(\base_b(x)).
$
Otherwise we say that $x$ is \emph{unswitched}.
\begin{claim}
  \label{clm:fa-ga} 
For each $0 \leq a \leq r$, there are positive real numbers $F_a, G_a$ satisfying
\begin{align}
G_a \leq F_a \leq (1+\gamma) G_a,\nonumber
\end{align}
such that, for every prefix $x \in \{0,1\}^{2ak+1}$ with $P(x)>0$, we have
\begin{align}
T(x) = \begin{cases}
F_a &: \text{if $x$ is switched,}\\
G_a &: \text{if $x$ is unswitched.}
\end{cases}\nonumber
\end{align}
\end{claim}
\begin{proof}[Proof of \cref{clm:fa-ga}]
We argue by reverse induction on $a$. For $a=r$, every prefix $x \in \{0,1\}^{2rk+1}=\{0,1\}^n$ is a leaf, so $T(x)=1$. Thus we may set $F_r = G_r = 1$. 

Now fix $0 \leq a < r$, and assume that the claim has already been established at step $a+1$ for appropriate values of $F_{a+1},G_{a+1}$. Let $x \in \{0,1\}^{2ak+1}$ with $P(x)>0$.

Suppose first that $x$ is switched. For each $z \in \{0,1\}^k$, let $\kappa(z) \in \{0,1\}^k$ denote the unique string for which
$
\key_a(x \circ z \circ \kappa(z)) = \psi(x;z). 
$
Then $x \circ z \circ \kappa(z)$ is a prefix of length $2(a+1)k+1$, and it is again active because $x$ was already active. Moreover, note that $\prod_{j=2ak+2}^{2(a+1)k+1} \hat \omega((x \circ z \circ \kappa(z))_{1:j}) = \hat \omega(x \circ z \circ \kappa(z)) = 1 + \gamma \cdot \phi(x;z)$, which is $1+\gamma$ when $z \in A(x)\cup B(x)$ and $1$ when $z \in C(x)$. %
Therefore
\begin{align}
T(x) = \bigl((1+\gamma)(|A(x)|+|B(x)|)+|C(x)|\bigr)F_{a+1} =: F_a.\nonumber
\end{align}

Suppose instead that $x$ is unswitched. For each $z \in \{0,1\}^k$, let $\kappa(z)$ be defined as above. Then $x \circ z \circ \kappa(z)$ is unswitched when $z \in A(x) \cup C(x)$, and switched when $z \in B(x)$. Moreover, we have $\prod_{j=2ak+2}^{2(a+1)k+1} \hat \omega((x \circ z \circ \kappa(z))_{1:j}) = \hat\omega(x \circ z \circ \kappa(z)) = 1 + \gamma \cdot \tilde \phi(x;z)$, which is $1+\gamma$ when $z \in A(x)$, and $1$ when $z \in B(x) \cup C(x)$.  Therefore
\begin{align}
T(x) = \bigl((1+\gamma)|A(x)|+|C(x)|\bigr)G_{a+1} + |B(x)|F_{a+1} =: G_a.\nonumber
\end{align}
Using the induction hypothesis $G_{a+1} \leq F_{a+1} \leq (1+\gamma)G_{a+1}$, we obtain
\begin{align}
G_a \leq \bigl((1+\gamma)(|A(x)|+|B(x)|)+|C(x)|\bigr)F_{a+1} = F_a\nonumber
\end{align}
and also
\begin{align}
F_a &= \bigl((1+\gamma)|A(x)|+|C(x)|\bigr)F_{a+1} + (1+\gamma)|B(x)|F_{a+1}\nonumber\\
&\leq (1+\gamma)\bigl((1+\gamma)|A(x)|+|C(x)|\bigr)G_{a+1} + (1+\gamma)|B(x)|F_{a+1}\nonumber\\
&\leq (1+\gamma)G_a.\nonumber
\end{align}
This completes the inductive step.
\end{proof}
For any $i \leq n$, we have $\hat Z = \sum_{x \in \{0,1\}^i} P(x) T(x)$, and $\hat Z_i = \sum_{x \in \{0,1\}^i} P(x)$. Thus, for any $0 \leq a \leq r-1$ and $x \in \{0,1\}^{2ak  + 1}$, we have
\begin{align}
\frac{1}{(1+\gamma) \cdot \hat Z_{2ak+1}} = \frac{G_a}{\sum_{x' \in \{0,1\}^{2ak+1}} P(x') \cdot (1+\gamma) G_a} \leq \frac{T(x)}{\hat Z} \leq \frac{(1+\gamma) \cdot G_a}{\sum_{x' \in \{0,1\}^{2ak+1}} P(x') \cdot G_a} = \frac{1+\gamma}{\hat Z_{2ak+1}}\nonumber,
\end{align}
and it follows from \cref{eq:muhat-mu-factorization} that $1/(1+\gamma) \leq \mu(x) / \hat \mu(x) \leq 1+\gamma$. In the above display, both inequalities use \cref{clm:fa-ga}: the first inequality uses the claim to bound $G_a \leq T(x)$ for all $x \in \{0,1\}^{2ak+1}$ and $(1+\gamma) G_a \geq T(x')$ for all $x' \in \{0,1\}^{2ak+1}$; the second inequality uses the claim to bound $T(x) \leq (1+\gamma) \cdot G_a$ and $T(x') \geq G_a$.

Having established the statement of \cref{lem:mu-muhat-ratio} for $x \in \{0,1\}^{\leq n}$ with $|x| = 2ak+1$ for some $0 \leq a \leq r$, we may now use this to extend the result for the remaining values of $x \in \{0,1\}^{\leq n}$: 

\paragraph{Case 1: $|x| \pmod{2k} \in \{k+1, \ldots, 2k \}$.} For any $0 \leq a \leq r-1$ and $k+1 \leq a' \leq 2k$, we have
\begin{align}
\hat Z_{2ak + a'} = \sum_{x \in \{0,1\}^{2ak + a'}} P(x) \in [(1+\gamma)^{-1},1] \cdot \sum_{x \in \{0,1\}^{2ak + 2k +1}} P(x) = [(1+\gamma)^{-1} \cdot\hat Z_{2ak + 2k +1},  \hat Z_{2ak + 2k +1}]\nonumber,
\end{align}
since for each $x \in \{0,1\}^{2ak + a'}$ with $P(x) > 0$, there is a unique $z \in \{0,1\}^{2k+1-a'}$ so that $P(x \circ z) > 0$, and such $z$ satisfies $P(x \circ z) = P(x) \cdot \hat \omega(x \circ z)$. For such $x,z$, we have $T(x \circ z)\cdot \hat\omega(x \circ z) = T(x)$, which means that 
\begin{align}
\frac{\mu(x)}{\hat \mu(x)} = \frac{T(x) \cdot \hat Z_{2ak+a'}}{\hat Z} = \frac{T(x \circ z) \cdot \hat \omega(x \circ z) \cdot \hat Z_{2ak+a'}}{\hat Z} \in [(1+\gamma)^{-1}, (1+\gamma)] \cdot \frac{\mu(x \circ z)}{\hat \mu(x \circ z)}.
\end{align}
Thus, by \cref{clm:fa-ga}, $\frac{\mu(x)}{\hat \mu(x)} \in [(1+\gamma)^{-2}, (1+\gamma)^2]$. %
As a consequence of this argument and of \cref{clm:fa-ga}, we also see that for all $x \in \{0,1\}^{2ak + k+1}$ (for any $0 \leq a \leq r-1$), we have
\begin{align}
T(x) \in [G_a/(1+\gamma), (1+\gamma)G_a].\label{eq:tx-kp1}
\end{align}

\paragraph{Case 2: $|x| \pmod{2k} \in \{2, \ldots, k \}$.} For any $0 \leq a \leq r-1$ and $1 \leq a' \leq k-1$, we have
\begin{align}
\hat Z_{2ak + 1 + a'} = \sum_{x \in \{0,1\}^{2ak + 1 + a'}} P(x) = 2^{a'-k} \cdot \hat Z_{2ak + k+1} \in 2^{a'-k} \cdot \hat Z_{2ak+2k+1} \cdot [(1+\gamma)^{-1},1]\nonumber,
\end{align}
since for any $x \in \{0,1\}^{2ak + 1 + a'}$, we have $P(x) = P(x_{1:2ak + k+1})$. Moreover, for such $x$, 
\begin{align}
T(x) = \sum_{x' \in \{0,1\}^{2ak+k+1} :\ x' \succ x} T(x') \in \left[ {2^{k-a'} \cdot G_a}/{(1+\gamma)}, 2^{k-a'} \cdot G_a \cdot (1+\gamma) \right],\nonumber
\end{align}
where the final step uses \cref{eq:tx-kp1}. 
Thus 
\begin{align}
\frac{\mu(x)}{\hat \mu(x)} = \frac{T(x) \cdot \hat Z_{2ak+1+a'}}{\hat Z} \in \frac{G_a \cdot \hat Z_{2ak+2k+1}}{\hat Z} \cdot [(1+\gamma)^{-2}, (1+\gamma)]\nonumber,
\end{align}
and it then follows from \cref{clm:fa-ga} that $\frac{\mu(x)}{\hat \mu(x)} \in [(1+\gamma)^{-2}, (1+\gamma)^2]$.
\end{proof}

The next lemma shows that the normalizing constant $\hat Z_i$ used to define $\hat \mu$ in \cref{eq:def-zhati,eq:def-muhat-zhat} does not depend on the choice of tuple $\tau \in \MH_{n,k,q,\ep}$. 
\begin{lemma}
\label{lem:zhat-invariant}
For each $i \in [n]$, there is a value $\hat Z_{\ep,i}$ so that, for any $\tau \in \MH_{n,k,q,\ep}$, the quantity $\hat Z_i = \sum_{x' \in \{0,1\}^i} \prod_{j=1}^i \hat \omega_{\tau}(x_{1:j}')$ appearing in \cref{eq:def-zhati} is equal to $\hat Z_{\ep,i}$. (In particular, it does \emph{not} depend on $\tau$.)
\end{lemma}
\begin{proof}
The proof is immediate by symmetry: consider any tuples $\tau = (\vinit, \phi, \psi), \tau' = ((\vinit)', \phi', \psi') \in \MH_{n,k,q,\ep}$. Let the mappings $\hat\omega$ corresponding to these two tuples be denoted $\hat \omega = \hat \omega_{\tau}, \hat \omega'=\hat\omega_{\tau'} : \{0,1\}^{\leq n} \to \BR_{\geq 0}$. We claim that, for each $i \leq n$, we can find a permutation $\pi_i : \{0,1\}^{i} \to \{0,1\}^{i}$ so that $\prod_{j=1}^i\hat \omega(x_{1:j}) =  \prod_{j=1}^i \hat \omega'(\pi_i(x)_{1:j})$ for all $x \in \{0,1\}^{i}$. To see this, we may construct the permutations $\pi_i$ inductively, in blocks of size $2k$. For the base case, $\pi_1$ is the identity if $\vinit = (\vinit)'$, and swaps the bits otherwise. We will suppose (inductively) that we have constructed $\pi_i$ for $i \leq 2ak+1$ for some $a \geq 0$, which satisfy the desired property above and also that $x \in \{0,1\}^{2ak+1}$ satisfies \cref{eq:prior-b-exists} for $\phi$ if and only if $\pi_{2ak+1}(x)$ satisfies \cref{eq:prior-b-exists} for $\phi'$. 

We now construct $\pi_i$ for the depths in the next block. Fix $x \in \{0,1\}^{2ak+1}$ and write $x' := \pi_{2ak+1}(x)$. Let the sets $A(x), B(x), C(x)$ be defined as in \cref{eq:A,eq:B,eq:C} in terms of $\phi$, and let $A'(x'), B'(x'), C'(x')$ be defined similarly in terms of $\phi'$. Since the sizes of these respective sets match up (per \cref{eq:ABC-sizes-initial-lemmas}), we can find a permutation $\sigma_x : \{0,1\}^k \to \{0,1\}^k$ so that $\sigma_x(A(x)) = A'(x')$, $\sigma_x(B(x)) = B'(x')$, and $\sigma_x(C(x)) = C'(x')$. 

We first define the permutations at depths lying in the query block. For $1 \leq s \leq k$ and $u \in \{0,1\}^s$, set
$
\pi_{2ak+1+s}(x \circ u) := x' \circ u.
$
Since all edge weights in the query block are equal to $1$, this preserves the product of edge weights up to depth $2ak+1+s$.

It remains to define the permutations inside the key block. Fix $z \in \{0,1\}^k$ and $1 \leq s \leq k$. Choose any bijection $\lambda_{x,z,s}:\{0,1\}^s \to \{0,1\}^s$ satisfying
$
\lambda_{x,z,s}\left(\psi(x;z)_{1:s}\right) = \psi'(x';\sigma_x(z))_{1:s}.
$
Then for $v \in \{0,1\}^s$, define
\begin{align}
\pi_{2ak+k+1+s}(x \circ z \circ v) := x' \circ \sigma_x(z) \circ \lambda_{x,z,s}(v).\label{eq:pi-key-block}
\end{align}
For each fixed $s$, this defines a bijection on $\{0,1\}^{2ak+k+1+s}$, since $\pi_{2ak+1}$, $\sigma_x$, and each $\lambda_{x,z,s}$ are bijections. Moreover, \cref{eq:pi-key-block} maps a prefix whose key block agrees with $\psi(x;z)$ for its first $s$ coordinates to a prefix whose key block agrees with $\psi'(x';\sigma_x(z))$ for its first $s$ coordinates, and maps every other key prefix to one which also fails the corresponding agreement condition. Thus the product of the key-indicator factors is preserved. If $s<k$, this proves the desired product identity at depth $2ak+k+1+s$. To deal with depth $i = 2(a+1)k+1$, we simply note that the choice of $\sigma_x$ maps $A(x)$ to $A'(x')$, $B(x)$ to $B'(x')$, and $C(x)$ to $C'(x')$, which completes the inductive construction of the permutations $\pi_i$.

Now, existence of $\pi_i$ as above implies that 
\begin{align}
\sum_{x \in \{0,1\}^i} \prod_{j=1}^i \hat \omega(x_{1:j}) = \sum_{x \in \{0,1\}^i} \prod_{j=1}^i \hat \omega'(\pi_i(x)_{1:j}) = \sum_{x' \in \{0,1\}^i} \prod_{j=1}^i \hat \omega'(x'_{1:j}),\nonumber
\end{align}
as desired.
\end{proof}

\cref{lem:conditional-independence} is a technical ingredient needed in our proof; it shows that, when we draw a tuple $\tau = (\vinit, \phi, \psi)$ uniformly at random from $\MH_{n,k,q,\ep}$, and condition on the values of the resulting mapping $\hat \mu_\tau$ at a number of points $x \in \{0,1\}^{\leq n}$, then conditionally $\psi$ is independent of $(\vinit, \phi)$. 
\begin{lemma}
\label{lem:conditional-independence}
Fix $q \in \BN$, and consider any collection $x_1, \ldots, x_q \in \{0,1\}^{\leq n}$ and $\mu_1, \ldots, \mu_q \geq 0$. Let $\BP$ denote the distribution of a uniformly random tuple $\tau = (\vinit, \phi,\psi) \sim \MH_{n,k,q,\ep}$ and the induced value of $\hat \mu = \hat \mu_\tau$. %
Then for any $v' \in \{0,1\}$, $\phi' \in \Phi_{n,k,q,\ep}$, and $\psi' \in \Psi_{n,k,q}$, we have
\begin{align}
\BP(\psi = \psi'  \mid (\vinit, \phi) = (v',\phi'), \ \ \hat \mu(x_j) = \mu_j\ \forall j \in [q]) = \BP(\psi = \psi' \mid \hat \mu(x_j) = \mu_j\ \forall j \in [q]).\nonumber
\end{align}
\end{lemma}
\begin{proof}
The definition of $\hat \omega = \hat\omega_\tau$ and $\hat \mu = \hat \mu_\tau$ above gives that there is a quantity $\eta_{\vinit,\phi}(x) > 0$ depending only on $(\vinit,\phi)$ and an indicator $\chi_\psi(x) \in \{0,1\}$ depending only on $\psi$ such that
\begin{align}
\hat\mu_{\tau}(x) = \eta_{\vinit,\phi}(x) \cdot \chi_\psi(x)\label{eq:muhat-factor-ci}
\end{align}
for all $x \in \{0,1\}^{\leq n}$. 
Indeed, $\eta_{\vinit,\phi}(x)$ is the product of the positive factors of the form $1+\gamma$ and $1+\gamma \cdot \phi(\base_a(x); \query_a(x))$ or $1+\gamma \cdot \tilde \phi(\base_a(x); \query_a(x))$, divided by $\hat Z_{\ep,|x|}$, while $\chi_\psi(x)=1$ if and only if every key block appearing along the path to $x$ agrees with the corresponding value determined from $\psi$. 

Let
\begin{align}
  \ME := & \{ (\vinit, \phi, \psi) \mid \hat \mu_{\vinit,\phi,\psi}(x_j) = \mu_j\ \forall j \in [q]\}\nonumber\\
\ME_\eta:=& \{(\vinit,\phi,\psi) \mid \eta_{\vinit, \phi}(x_j) = \mu_j \ \forall j \mbox{ s.t. } \mu_j > 0 \}\nonumber\\
\ME_\chi :=& \{(\vinit, \phi,\psi) \mid \chi_\psi(x_j) = 1 \ \forall j \mbox{ s.t. } \mu_j > 0, \ \chi_\psi(x_j) = 0 \ \forall j \mbox{ s.t. } \mu_j = 0\}\nonumber
\end{align}
Then since $\eta_{\vinit, \phi}(x) > 0$ for all $\vinit,\phi, x$, it follows from \cref{eq:muhat-factor-ci} that $\ME = \ME_\eta \cap \ME_\chi$: for each $j$ with $\mu_j>0$ the equality $\hat \mu(x_j)=\mu_j$ is equivalent to $\eta_{\vinit,\phi}(x_j)=\mu_j$ and $\chi_\psi(x_j)=1$, while for each $j$ with $\mu_j=0$ it is equivalent to $\chi_\psi(x_j)=0$. Moreover, note that the event $(\vinit, \phi, \psi) \in \ME_\eta$ depends only on $(\vinit, \phi)$, while the event $(\vinit, \phi, \psi) \in \ME_\chi$ depends only on $\psi$. For any $(v', \phi')$ for which $\BP(\{ (\vinit, \phi) = (v', \phi') \} \cap \ME) > 0$, we have
\begin{align}
\BP(\psi = \psi' \mid \{(\vinit, \phi) = (v',\phi')\} \cap \ME)
&= \BP(\psi = \psi' \mid \{(\vinit, \phi) = (v',\phi')\} \cap \ME_\chi)\nonumber\\
&= \BP(\psi = \psi' \mid \ME_\chi)
= \BP(\psi = \psi' \mid \ME_\chi \cap \ME_\eta)\nonumber,
\end{align}
where the first equality uses that $(\vinit, \phi) = (v', \phi')$ implies that $\ME_\eta$ holds, while the second and third equalities use the independence of $(\vinit, \phi)$ and $\psi$. Since $\ME=\ME_\eta \cap \ME_\chi$, this is the desired conditional independence statement.
\end{proof}

\cref{def:consistent} formalizes the notion of an ``ideal'' sequence of queries which can be made to $\hat \mu$: namely such a sequence satisfies the property that the queries do not ``skip ahead'' in a certain sense. In particular, the sequence has the property that before making any query, all prefixes of that query of length $1\pmod{2k}$ have been previously queried. 
\begin{definition}
\label{def:consistent}
 We say that a sequence $X_1, \ldots, X_j \in \{0,1\}^{\leq n}$ is \emph{ideal} if the following properties hold:
\begin{itemize}
\item $X_1, \ldots, X_j$ are all distinct.
\item For all $j' \in [j]$, $|X_{j'}| \pmod{2k} \in \{1, k+2, k+3, \ldots, 2k\}$.
\item For each $j' \in [j]$ and $0 \leq a < r$ for which $|X_{j'}| > 2ak + 1$, there is some $j'' < j'$ for which $X_{j''} = \base_a(X_{j'}) = (X_{j'})_{1:2ak+1}$.
\end{itemize}
\end{definition}

The main technical lemma in the proof of our lower bound is \cref{lem:close-consistent-query}, which shows the following: under a uniformly random draw of $\tau \sim \MH_{n,k,q,\ep}$, for any fixed ideal sequence of queries (per \cref{def:consistent}), the distribution of the output of the final query in this sequence conditioned on the results of the previous queries in the sequence does not depend much on the value of $\vinit$. Ultimately, we will combine \cref{lem:close-consistent-query} with the chain rule for KL divergence to show that the results of all queries in an ideal sequence do not depend much on the value of $\vinit$, and this will ultimately allow us to show that no query-limited algorithm can guess the value of $\vinit$. 

To state the lemma, we introduce the following notation. Let $\BP$ denote the joint distribution over the random draw of $(\vinit,\phi,\psi) \sim \MH_{n,k,q,\ep}$ and the internal randomness of $\Alg$, and $\E$ denote the corresponding expectation. For each $b \in \{0,1\}$, let $\BP_b(\cdot) := \BP(\cdot \mid \vinit = b)$.

\begin{lemma}
\label{lem:close-consistent-query}
Suppose that $X_1, \ldots, X_j \in \{0,1\}^{\leq n}$ is an ideal sequence for some $j \leq q$. For $b \in \{0,1\}$, let $\BP_{b}$ denote the joint distribution of $(\hat \mu_\tau, \hat \omega_\tau, \tau)$ for $\tau = (\vinit, \phi, \psi) \sim \MH_{n,k,q,\ep}$, conditioned on $\vinit = b$. Fix any values of $ \mu_1', \ldots,  \mu'_{j-1} \geq 0$, and write $\ME$ to denote the event that $\hat \mu(X_{j'}) =  \mu'_{j'}$ for all $j' \in [j-1]$. If $\BP_0(\ME), \BP_1(\ME) > 0$, then we have
\begin{align}
\kld{\BP_0(\hat \mu(X_{j}) = \cdot \mid \ME)}{\BP_1(\hat \mu(X_{j}) = \cdot \mid \ME)} \leq O(\ep^2).\nonumber
\end{align}
\end{lemma}

\begin{proof}
To simplify notation, we will typically write $\hat \mu := \hat \mu_\tau, \hat \omega:= \hat \omega_\tau$. Further, we will make use of the sets $A(x), B(x), C(x)$ indexed by $x \in \{0,1\}^n$ (which are determined by the choice of $\tau$), defined in \cref{eq:A,eq:B,eq:C}, respectively. 
We consider the following two cases for $X_j$:
\paragraph{Case 1 for $X_j$.} Suppose $X_j$ lies strictly inside a key block, so that $|X_j| = 2ak + k + 1 + t$ for some $0 \leq a < r$ and $1 \leq t \leq k-1$. Since the sequence $X_1, \ldots, X_j$ is ideal, there is some $j' < j$ for which $X_{j'} = (X_j)_{1:2ak+1}$. Note that $\hat \mu(X_{j'}) = \mu_{j'}'$ under $\ME$, so that 
\begin{align}
\hat \mu(X_j) = \frac{\mu_{j'}' \cdot \hat Z_{2ak+1}}{\hat Z_{|X_j|}} \cdot \One{\key_a(X_j) = \psi(\base_a(X_j); \query_a(X_j))_{1:t}},\nonumber
\end{align}
since every edge weight between depths $2ak+2$ and $2ak+k+1$ is equal to $1$. In particular, $\hat \mu(X_j)$ is $\psi$-measurable conditioned on $\ME$ (under both $\BP_0, \BP_1$). By \cref{lem:conditional-independence}, $\psi$ and $(\phi, \vinit)$ are conditionally independent conditioned on $\ME$ under both $\BP_0, \BP_1$, so
\begin{align}
\BP_0(\hat \mu(X_j) = \cdot \mid \ME) = \BP_1(\hat \mu(X_j) = \cdot \mid \ME).\label{eq:case1-kl}
\end{align}

\paragraph{Case 2 for $X_j$.} Suppose $|X_j| = 2(a+1)k+1$ for some $0 \leq a < r$. 
Let
\begin{align}
I_j :=& \One{\key_a(X_j) = \psi(\base_a(X_j); \query_a(X_j))}\nonumber\\
 R_j :=& \One{\vinit = (X_j)_1 \text{ or } \exists b<a:\ \phi(\base_b(X_j); \query_b(X_j)) \neq \tilde \phi(\base_b(X_j); \query_b(X_j))}\nonumber\\
 S_j := & R_j \cdot \phi(\base_a(X_j); \query_a(X_j)) + (1-R_j) \cdot \tilde \phi(\base_a(X_j); \query_a(X_j)) \in \{0,1\}.\nonumber
\end{align}
Since the sequence $X_1, \ldots, X_j$ is ideal, there is some $j' < j$ so that $X_{j'} = (X_j)_{1:2ak+1}$. We may then write
\begin{align}
\hat \mu(X_j) = \frac{\hat \mu(X_{j'}) \cdot \hat Z_{2ak+1}}{\hat Z_{|X_j|}} \cdot I_j \cdot \left(1 + \gamma \cdot S_j \right),\label{eq:hatmu-formula}
\end{align}
To characterize the distribution of $\hat \mu(X_j)$ conditioned on $\SF_{j-1}$ under $\BP_0, \BP_1$ we observe that $[j-1]$ may be partitioned into 3 subsets:
\begin{itemize}
\item $\MJ_1$ consists of $\ell \in [j-1]$ for which $\mu_\ell' = 0$.
\item $\MJ_2$ consists of $\ell \not \in \MJ_1$ for which $\base_a(X_\ell) \neq \base_a(X_j)$.
\item $\MJ_3$ consists of $\ell \not \in \MJ_1 \cup \MJ_2$. We claim that all $\ell \in \MJ_3$ satisfy $\query_a(X_\ell)  \neq \query_a(X_j)$. This claim holds since, in order to have $\hat \mu(X_\ell) > 0$ (which must have positive probability) and $\base_a(X_\ell) = \base_a(X_j)$ and $\query_a(X_\ell) = \query_a(X_j)$, we must have $X_j \prec X_\ell$, which is impossible since the sequence $X_1, \ldots, X_j$ (as $|X_j| \pmod{2k} = 1$) is ideal. 
\end{itemize}
Since sequence $X_1, \ldots, X_j$ is ideal, under $\ME$, the values of $\hat \omega(X_\ell)$ for each $\ell \in \MJ_2 \cup \MJ_3$ are uniquely determined (under the distributions $\BP_0, \BP_1$); indeed, for such $\ell$ we have $\hat \mu(X_\ell) = \frac{1}{\hat Z_{|X_\ell|}} \cdot \prod_{b=0}^{\lfloor (|X_\ell| - 2)/(2k) \rfloor} \hat \omega(\base_b(X_\ell))$, and the values of $\hat \omega(\base_b(X_\ell))$ for all but (potentially) the largest value of $b$ in this product are determined by $\mu'_1, \ldots, \mu'_{\ell-1}$on the event $\{ \hat \mu(X_{\ell'}) = \mu_{\ell'}' \ \forall \ell' < \ell \}$. Thus, we may make the following definitions:
\begin{itemize}
\item Let $\tilde A$ denote the set of values of $\query_a(X_\ell)$, for $\ell \in \MJ_3$, for which $\hat \omega(X_\ell) = 1 + \gamma$.
\item Let $\tilde C$ denote the set of values of $\query_a(X_\ell)$, for $\ell \in \MJ_3$, for which $\hat \omega(X_\ell) = 1$.
\end{itemize}
All $\ell \in \MJ_3$ satisfy $\query_a(X_\ell) \in \tilde A \cup \tilde C$ since $\hat \omega(\base_a(X_\ell)) \neq 0$ for $\ell \in \MJ_3$. Let $\Equery$ denote the event that $\hat \omega(x) = 1+\gamma$ for all $x \in \tilde A$ and that $\hat \omega(x) = 1$ for all $x \in \tilde C$.

Finally, let $\MD$ denote the set of all values of $\base_b(X_\ell)$ not equal to $\base_a(X_j)$ for any $0 \leq b< r$ and $\ell \in \MJ_2 \cup \MJ_3$ with $2bk+1 \leq |X_\ell|$. Since the sequence $X_1, \ldots, X_j$ is ideal, each $X_\ell$ can only contribute at most one new value to $\MD$, meaning that $|\MD| \leq q-1$. For each $d \in \MD$, let $(v\^d, \tilde v\^d) \in \Vvalid_{k,\ep}$ be an arbitrary $\ep$-valid tuple; via slight abuse of notation, we let $(v,\tilde v)$ denote the collection of $(v\^d, \tilde v\^d)$ for $d \in \MD$. Let $\Ebase_{v,\tilde v}$ denote the event that $\phi(d) = (v\^d, \tilde v\^d)$ for each $d \in \MD$. 

We proceed to analyze the distribution $\BP_{(X_j)_1}(\hat \mu(X_j) = \cdot \mid \ME)$. 
For any values of $\psi' \in \Psi_{n,k,q}$ and $(v,\tilde v) \in (\Vvalid_{k,\ep})^{\MD}$ as above which are consistent with $\ME$ in the sense that $\BP_{(X_j)_1}(\ME \mid \{ \psi = \psi' \} \cap \Ebase_{v,\tilde v}) > 0$, we have
\begin{align}
& \BP_{(X_j)_1} \left(S_j = 1 \mid \ME \cap \Ebase_{v,\tilde v} \cap \Equery \cap \{\psi = \psi'\} \right) \nonumber\\
=& \BP_{(X_j)_1} \left(S_j = 1 \mid \Ebase_{v,\tilde v} \cap \Equery \cap \{\psi = \psi'\} \right) \nonumber\\
=& \BP_{(X_j)_1} \left(S_j = 1 \mid \Ebase_{v,\tilde v} \cap \Equery\right) \nonumber\\
=& \BP_{(X_j)_1} \left( \query_a(X_j) \in A(\base_a(X_j)) \cup B(\base_a(X_j))\mid \Equery\right) \nonumber\\
=& \BP_{(X_j)_1} \left( \query_a(X_j) \in A(\base_a(X_j)) \cup B(\base_a(X_j)) \mid \tilde A \subset A(\base_a(X_j)) \cup B(\base_a(X_j)),\ \tilde C \subset C(\base_a(X_j)) \right)\nonumber\\
=& \frac{2^{k-1} - |\tilde A|}{2^k - |\tilde A| - |\tilde C|}\label{eq:xj-same},
\end{align}
where the first equality uses that, for $\psi', (v,\tilde v)$ which are consistent with $\ME$, we have $\Ebase_{v,\tilde v} \cap \Equery \cap \{\psi = \psi'\} \subseteq \ME$; the second equality uses that $S_j$ and $\Ebase_{v,\tilde v} \cap \Equery$ depend only on $\phi$, together with the fact that $\phi, \psi$ are independent; and the third equality uses the fact that $\phi$ is $q$-wise independent: $\Ebase_{v,\tilde v}$ conditions on the values of $\phi(d)$ for at most $q-1$ values of $d$, and $S_j = 1$ if and only if $\phi(\base_a(X_j); \query_a(X_j)) = 1$, i.e., if $\query_a(X_j) \in A(\base_a(X_j)) \cup B(\base_a(X_j))$. The fourth equality simply rewrites the definition of $\Equery$, and the final equality follows from the following reasoning: on the event $\tilde A \subset A(\base_a(X_j)) \cup B(\base_a(X_j)),\ \tilde C \subset C(\base_a(X_j))$, the set $A(\base_a(X_j)) \cup B(\base_a(X_j)) \backslash \tilde A$ is a uniformly random set of size $2^{k-1} - |\tilde A|$ in $\{0,1\}^k \backslash (\tilde A \cup \tilde C)$, which has size $2^k - |\tilde A| - |\tilde C|$. 

We next perform a similar computation for $\BP_{1-(X_j)_1}$. To do so, first note that, under $\BP_{1-(X_j)_1}$, conditioned on $\Ebase_{v,\tilde v}$, the values of $v,\tilde v$ uniquely determine the value of $R_j$: indeed, conditioned on $\Ebase_{b,\tilde v}$, $R_j = 1$ if and only if some $b < a$ satisfies $v\^{\base_b(X_j)}_{\query_b(X_j)} \neq \tilde v\^{\base_b(X_j)}_{\query_b(X_j)}$. 
For any $(v,\tilde v) \in (\Vvalid_{k,\ep})^{\MD}$ and $\psi' \in \Psi_{n,k,q}$ which are consistent with $\ME$ under $\BP_{1-(X_j)_1}(\cdot)$, which further satisfy that $R_j = 1$ under $\Ebase_{v,\tilde v}$, we have
\begin{align}
& \BP_{1-(X_j)_1} \left(S_j = 1 \mid\ME \cap \Ebase_{v,\tilde v} \cap \Equery \cap \{\psi = \psi'\} \right)\nonumber\\
 =& \BP_{1-(X_j)_1} \left( \query_a(X_j) \in A(\base_a(X_j)) \cup B(\base_a(X_j)) \mid \tilde A \subset A(\base_a(X_j)) \cup B(\base_a(X_j)),\ \tilde C \subset C(\base_a(X_j)) \right)\nonumber\\
=& \BP_{(X_j)_1} \left( S_j = 1\mid \Ebase_{v,\tilde v} \cap \ME\right)\label{eq:xj-diff1},
\end{align}
using the same reasoning as in \cref{eq:xj-same}. 
In instead $(v,\tilde v)$ is so that $R_j = 0$ under $\Ebase_{v,\tilde v}$,  we have
\begin{align}
& \BP_{1-(X_j)_1} \left(S_j = 1 \mid\ME \cap \Ebase_{v,\tilde v} \cap \Equery \cap \{\psi = \psi'\} \right)\nonumber\\
=& \BP_{1-(X_j)_1} \left( S_j = 1 \mid \Ebase_{v,\tilde v} \cap \Equery \right)\nonumber\\
=& \BP_{1-(X_j)_1} \left( S_j = 1 \mid \Ebase_{v,\tilde v} \cap \{\tilde A \subset A(\base_a(X_j))\} \cap \{\tilde C \subset C(\base_a(X_j))\cup B(\base_a(X_j))\} \right)\nonumber\\
 =& \BP_{1-(X_j)_1} \left( \query_a(X_j) \in A(\base_a(X_j)) \mid \tilde A \subset A(\base_a(X_j)) ,\ \tilde C \subset C(\base_a(X_j))\cup B(\base_a(X_j)) \right)\nonumber\\
=& \frac{2^{k-1} - \ep 2^k - |\tilde A|}{2^k - |\tilde A| - |\tilde C|}\label{eq:xj-diff2},
\end{align}
where the first equality follows the steps as in \cref{eq:xj-same}; the second equality uses that under $\Ebase_{v,\tilde v}$) (so that $R_j = 0$), the event $\Equery$ is equivalent to the event that $\tilde A \subset A(\base_a(X_j))$ and $\tilde C \subset C(\base_a(X_j))\cup B(\base_a(X_j))$; and the third equality uses the $q$-independence of $\phi$, together with the fact that $\Ebase_{v,\tilde v}$ depends on the values of $\phi(d)$ for at most $q-1$ values of $d$, all not equal to $\base_a(X_j)$. The final equality uses the following reasoning: on the event $\tilde A \subset A(\base_a(X_j))$, $\tilde C \subset C(\base_a(X_j)) \cup B(\base_a(X_j))$, the set $A(\base_a(X_j)) \backslash \tilde A$ is a uniformly random set of size $2^{k-1} - \ep 2^k - |\tilde A|$ in $\{0,1\}^k \backslash (\tilde A \cup \tilde C)$, which has size $2^k - |\tilde A| - |\tilde C|$.

Note that $|\tilde A| \cup |\tilde C| \leq |\MJ_3| \leq j-1 \leq q-1$. Since we have assumed that $q \leq  2^k/100$, the right-hand sides of \cref{eq:xj-same,eq:xj-diff1,eq:xj-diff2} all lie in $[1/3, 2/3]$ and differ by at most $10\ep$.  
Since \cref{eq:xj-same,eq:xj-diff1,eq:xj-diff2} hold for every $v,\tilde v$, it follows that
\begin{align}
\BP_0(S_j = 1 \mid \ME), \BP_1(S_j = 1 \mid \ME) \in [1/3,2/3], \qquad \left| \BP_0(S_j = 1 \mid \ME) - \BP_1(S_j = 1 \mid \ME) \right| \leq 10\ep\nonumber,
\end{align}
and so
\begin{align}
\kld{\BP_0(S_j = \cdot \mid \ME)}{\BP_1(S_j = \cdot \mid \ME)} \leq O(\ep^2).\nonumber
\end{align}
Next, \cref{lem:conditional-independence} gives that, conditioned on $\ME$, $I_j$ is independent of $S_j$ under the measures $\BP_0, \BP_1$; and further, the distribution of $I_j$ conditioned on $\ME$ is identical under $\BP_0, \BP_1$. Thus, combining the above display with \cref{eq:hatmu-formula} gives that
\begin{align}
\kld{\BP_0(\hat \mu(X_j) = \cdot \mid \ME)}{\BP_1(\hat \mu(X_j) = \cdot \mid \ME)} \leq O(\ep^2),\nonumber
\end{align}
as desired.
\end{proof}

Next, \cref{lem:technical-random-lb} uses \cref{lem:close-consistent-query} to show that no algorithm which can make $q$ queries to $\hat \mu_\tau$ (for a random choice of $\tau$) can predict the value of $\vinit$, as long as $q$ is not too large.
\begin{lemma}
\label{lem:technical-random-lb}
Suppose $q \leq 2^k/100$, and a tuple $\tau = (\vinit, \phi, \psi) \in \MH_{n,k,q,\ep}$ is drawn uniformly at random, and let $\hat\mu = \hat \mu_{\tau}$ denote the corresponding oracle. Then any randomized algorithm which makes $q$ adaptive queries to $\hat\mu(\cdot)$ and outputs a bit $\hat B \in \{0,1\}$ satisfies
\begin{align}
\BP(\hat B = \vinit) \leq \frac 12 + O(\ep\sqrt{q}) + q\cdot 2^{-k}\nonumber.
\end{align}
\end{lemma}
\begin{proof}
Fix an algorithm $\Alg$, and denote the (adaptive) queries the algorithm makes by $X_1, \ldots, X_q \in \{0,1\}^{\leq n}$. %
We begin by setting up some notation. 

\paragraph{Setup.}  For any $x \in \{0,1\}^{2ak + i}$ for some $0 \leq a < r$ and $1 \leq i \leq 2k$, if $i \in \{2,3,\ldots,k+1\}$, then every edge weight on the path from $x_{1:2ak+1}$ to $x$ is equal to $1$, so $\hat\mu(x)$ is determined by $\hat\mu(x_{1:2ak+1})$ together with the deterministic normalizing constant $\hat Z_{2ak+i}$ from \cref{eq:def-zhati,lem:zhat-invariant}. Hence, it is without loss of generality to assume that for each query $j \in [q]$, $|X_j| \pmod{2k} \in \{1, k+2, k+3, \ldots, 2k\}$. 

For each $j \in [q]$, write $\SF_{j}$ for the sigma-algebra generated by the first $j$ query-answer pairs $(X_1, \hat\mu(X_1)), \ldots, (X_{j}, \hat\mu(X_{j}))$, as well as the internal randomness of $\Alg$; then for each $j \in [q]$, $X_j$ is $\SF_{j-1}$-measurable. Moreover, at a cost of a factor of at most $2$ in $q$, we may assume the following: for each query $X_j$ lying strictly  inside a ``key block'', i.e., with $|X_j| = 2ak + k + 1 + t$ for some $0 \leq a < r$ and $1 \leq t \leq k-1$, there is some $j' < j$ for which $X_{j'} = (X_j)_{1:2ak+1}$. In words, prior to making such a query, the algorithm has already queried the most recent ancestor of $X_j$ whose length is $1 \pmod{2k}$.

Recall that $\BP$ denotes the joint distribution over the random draw of $(\vinit,\phi,\psi) \sim \MH_{n,k,q,\ep}$ and the internal randomness of $\Alg$, and $\E$ denotes the corresponding expectation. For each $b \in \{0,1\}$, recall that $\BP_b(\cdot) := \BP(\cdot \mid \vinit = b)$. We will show that the probability laws of $\BP_0(X_1, \ldots, X_q, \hat \mu(X_1), \ldots, \hat \mu(X_q) = \cdot)$ and $\BP_1(X_1, \ldots, X_q, \hat \mu(X_1), \ldots, \hat \mu(X_q) = \cdot)$ are close, from which the desired result will follow using Bayes rule. %

\paragraph{Alternative algorithm: $\tAlg$.} First, we consider the following further modification of $\Alg$, which we denote by $\tAlg$. $\tAlg$ internally simulates $\Alg$, in the following manner: it maintains a sequence $Y_1, Y_2, \ldots, Y_t$ (intialized to be the empty sequence, so that $t = 0$), as follows. For each $j \in [q]$, when $\Alg$ is about to query $\hat \mu(X_j)$, $\tAlg$ checks if the sequence $Y_1, \ldots, Y_t, X_j$ is ideal (in the sense of \cref{def:consistent}); we let this event be denoted by $\ME_j$. If $\ME_j$ holds, then $\tAlg$ sets $Y_{t+1} = X_j$, increments $t$ by $1$, and continues to execute $\Alg$. Otherwise, $\tAlg$ does \emph{not} query $\hat \mu(X_j)$ and simply proceeds as if the result of the query had returned $\hat \mu(X_j) = 0$. We let $\tilde B$ denote the output bit returned by $\tAlg$. Moreover, for each $j \in [q]$, we let $\tilde \mu_j$ denote the value which $\tAlg$ uses for $\hat \mu(X_j)$ (i.e., $\tilde \mu_j = \hat \mu(X_j)$ if $Y_1, \ldots, Y_t, X_j$ is ideal, and otherwise $\tilde \mu_j = 0$). We may consider $\Alg,\tAlg$ as being defined on the same probability space, using the same internal randomness and the same sampled tuple $(\vinit,\phi,\psi) \sim \MH_{n,k,q,\ep}$. We denote the corresponding probability measure by $\BP$ (and use $\BP_0, \BP_1$ to denote conditioning on $\vinit = 0,1$, respectively). 
\begin{claim}
  \label{clm:hatb-hatb0}
It holds that $\BP(\hat B = \tilde B) \geq 1 - q 2^{-k}$.
\end{claim}

The proof of \cref{clm:hatb-hatb0} is deferred to below. Given \cref{clm:hatb-hatb0}, it suffices to upper bound $\BP(\tilde B = \vinit)$. 

\paragraph{Analyzing $\tAlg$.} Let $t_j$ be the value of $t$ maintained by $\tAlg$ at the beginning of step $j$, so that prior to step $j$, $\tAlg$ has queried $\hat \mu(\cdot)$ only at $Y_1, \ldots, Y_{t_j}$. We denote the ``simulated queries'' of $\Alg$ made by $\tAlg$ by $\tilde X_1, \ldots, \tilde X_q$. Further, we let $\tilde \SF_j$ denote the sigma-algebra generated by the internal randomness of $\tAlg$ together with $(\tilde X_\ell, \tilde \mu_\ell)_{\ell \in [j]}$. If $\ME_j$ holds, then the sequence $Y_1, \ldots, Y_{t_j}, Y_{t_j+1} = \tilde X_j$ is ideal, so we may apply \cref{lem:close-consistent-query} to this sequence to obtain that there is an absolute constant $C$ such that
\begin{align}
D_{\mathrm{KL}}\!\left(
\mathcal \BP_1\left(\tilde \mu_j = \cdot \mid \tilde \SF_{j-1}\right)\,
\middle\|\,
\mathcal \BP_0\left(\tilde \mu_j = \cdot \mid \tilde \SF_{j-1}\right)
\right)
\leq C\ep^2.\label{eq:case2-kl}
\end{align}
If $\ME_j$ does not hold, then $\tilde \mu_j = 0$ and so
\begin{align}
D_{\mathrm{KL}}\!\left(
\mathcal \BP_1\left(\tilde \mu_j = \cdot \mid \tilde \SF_{j-1}\right)\,
\middle\|\,
\mathcal \BP_0\left(\tilde \mu_j = \cdot \mid \tilde \SF_{j-1}\right)
\right) = 0.\label{eq:case2b-kl}
\end{align}

The chain rule for KL divergence yields
\begin{align}
& \kld{\BP_1(\tilde X_{1:q},\tilde \mu_{1:q} = \cdot)}{\BP_0(\tilde X_{1:q}, \tilde \mu_{1:q} = \cdot)}\nonumber\\
=& \sum_{j=1}^q D_{\mathrm{KL}}\!\left(
\mathcal \BP_1\left(\tilde\mu_j = \cdot \mid \tilde \SF_{j-1}\right)\,
\middle\|\,
\mathcal \BP_0\left(\tilde\mu_j = \cdot \mid \tilde \SF_{j-1}\right)
\right) \leq C q \ep^2\nonumber,
\end{align}
where the final inequality uses \cref{eq:case2-kl,eq:case2b-kl}.

Pinsker's inequality yields
\begin{align}
\tvd{\BP_1(\tilde X_{1:q},\tilde \mu_{1:q} = \cdot)}{\BP_0(\tilde X_{1:q}, \tilde \mu_{1:q} = \cdot)}\leq O(\ep\sqrt q).\nonumber
\end{align}
Consequently, we have
\begin{align}
  & \E\left[ \left|\BP(\vinit = 1 \mid \tilde X_{1:q}, \tilde \mu_{1:q})- \BP(\vinit = 0 \mid \tilde X_{1:q}, \tilde \mu_{1:q})\right| \right]\nonumber\\
  =& \tvd{\BP_1(\tilde X_{1:q},\tilde \mu_{1:q} = \cdot)}{\BP_0(\tilde X_{1:q}, \tilde \mu_{1:q} = \cdot)}\leq O(\ep\sqrt q)\nonumber,
\end{align}
and since $\vinit \in \{0,1\}$, it follows that
\begin{align}
\E \left[ \left|\BP(\vinit = 1 \mid \tilde X_{1:q}, \tilde \mu_{1:q}) - \frac 12\right| \right] \leq O(\ep\sqrt q),\nonumber
\end{align}
and thus, since $\tilde B$ is a function of $\tilde X_{1:q}, \tilde \mu_{1:q}$, and the internal randomness of $\Alg$, 
\begin{align}
\left|\BP(\vinit = \tilde B) - \frac 12 \right| \leq \E \left[ \left|\BP(\vinit = \tilde B \mid \tilde X_{1:q}, \tilde \mu_{1:q}) - \frac 12\right| \right] \leq O(\ep\sqrt q).\nonumber
\end{align}
Combining the above with \cref{clm:hatb-hatb0} yields the desired result. 
\end{proof}

Finally, we prove \cref{clm:hatb-hatb0}.
\begin{proof}[Proof of \cref{clm:hatb-hatb0}]
Recall that we couple $\Alg$ and $\tAlg$ using the same internal randomness and the same realization of $(\vinit,\phi, \psi)$. The two algorithms can differ only if, at some query step $j$, the event $\ME_j^c$ holds and yet the true oracle value satisfies $\hat\mu(X_j)\neq 0$; indeed, in all other cases $\tAlg$ feeds the same answer to the algorithm as $\Alg$ does.

For each $j \in [q]$, let $a_j \geq 0$ be chosen as small as possible so that $2k \cdot a_j +1 < |X_j|$ and $\base_{a_j}(X_j)$ has not been previously queried. (If no such $a_j$ exists, then we set $a_j = \perp$.) Let $J \in [q] \cup \{ \infty\}$ be a random variable denoting the first iteration for which $a_J \neq \perp$ and $\hat\mu(X_J) \neq 0$. 

Fix $j \in [q]$; in order for $J = j$ to hold, it is necessary that (a) $a_j \neq \perp$ and (b) $\key_{a_j}(X_j) = \psi(\base_{a_j}(X_j); \query_{a_j}(X_j))$. Note that $a_j$ is $\SF_{j-1}$-measurable. We now make the following claim: %
\begin{claim}
\label{clm:psi-uniform}
For each $j \in [q]$, under the execution of $\tAlg$, conditioned on $a_j \neq \perp$ and $\SF_{j-1}$, the distribution of $\psi(\base_{a_j}(X_j))$ is uniformly random on $\{0,1\}^k$.
\end{claim} 
\begin{proof}[Proof of \cref{clm:psi-uniform}]
First note that at each iteration $\ell$ (in particular, those $\ell < j$), $\hat \mu(X_\ell)$ only depends on $\psi$ through (a) the history $\SF_{\ell-1}$ and (b) potentially $\psi(\base_a(X_\ell))$ for a single value of $a$ (namely, $a = \lfloor (|X_\ell|-2)/k\rfloor$). Indeed, if the values $\base_b(X_\ell)$ for all $b = 0, 1, \ldots, a-1$ had not been queried prior to step $\ell$, then $\tAlg$ would not have actually queried $\hat \mu(X_\ell)$ and would have used the value $0$ as the response. Thus, the history of interaction $\SF_{j-1}$ only depends on $\psi$ through at most $j-1$ values of $\psi(x)$, $x \in \{0,1\}^{\leq n}$. 

Moreover, we claim that $\base_a(X_j)$ is not one of these values: by our assumption on $\Alg$ (and thus $\tAlg$), there is no $\ell < j$ for which $\base_a(X_\ell) = \base_a(X_j)$ and $|X_\ell| = 2ak + t$ for $t \in \{k+2, \ldots, 2k \}$ (as otherwise $\Alg$ would have queried $\hat \mu(\base_a(X_\ell))$ earlier). Further, if $\base_a(X_\ell) = \base_a(X_j)$ and $|X_\ell| \geq 2(a+1)k + 1$, then $\tAlg$ would not have queried $X_\ell$ (and would have instead proceeded by using $0$ as the returned value of the oracle call). 

Summarizing, since $\psi$ is drawn from a $q$-wise independent family, conditioned on $\SF_{j-1}$, the value of $\psi(\base_a(X_j))$ is uniformly random on its domain $(\{0,1\}^k)^{\{0,1\}^k}$. 
\end{proof}
If $a_j \neq \perp$ (so that $\base_{a_j}(X_j)$ has not been previously queried), then we must in fact that $|X_j| = 2(a_j+1)k + 1$ (by our assumption on $\Alg$). Hence, for any realization of the history for which $a_j \neq \perp$, we have
\begin{align}
\BP(J = j \mid \SF_{j-1}) = \BP(\key_{a_j}(X_j) = \psi(\base_{a_j}(X_j); \query_{a_j}(X_j)) \mid \SF_{j-1}) = 2^{-k}\nonumber,
\end{align}
where the final equality uses the fact that $\key_{a_j}(X_j)$, $\base_{a_j}(X_j)$, and $\query_{a_j}(X_j)$ are all $\SF_{j-1}$-measurable, and $\psi(\base_{a_j}(X_j); \query_{a_j}(X_j))$ is uniform on $\{0,1\}^k$, conditioned on $\SF_{j-1}$, by \cref{clm:psi-uniform}. A union bound then yields $\BP\left( J \leq q \right) \leq q/2^k$; in other words, with probability at least $1-q/2^k$, the execution traces of $\Alg, \tAlg$ are identical, and thus the algorithms return the same bit as output.
\end{proof}

We are now ready to prove \cref{thm:weak-lb}; it follows as a straightforward consequence of the previous lemmas, notably \cref{lem:technical-random-lb}.
\begin{proof}[Proof of \cref{thm:weak-lb}]
  Set $q = \frac{c \gamma^2 n^2}{ \log^2 (n) \cdot \log^2(2/\gamma)}$, for a sufficiently small constant $c > 0$ to be chosen later. 
Take  %
\begin{align}
k = \lceil 2 \log_2(n)\rceil, \qquad r =\left\lfloor \frac{n-1}{2k}\right\rfloor, \qquad \ep := \frac{4 \log(48/\gamma) \cdot k}{n}, \qquad n' = 2rk+1.\label{eq:weak-params}
\end{align}

Suppose to the contrary that such $\Alg$ exists, and consider the execution of $\Alg$ with respect to the distribution $\mu^\star = \mu_{\tau}$ and mapping $\hat \mu_{\tau}$, where $\tau = (\vinit, \phi, \psi) \sim \MH_{n',k,q}$ is drawn uniformly at random. First, \cref{lem:mu-muhat-ratio} ensures that we indeed have that $\frac{1}{(1+\gamma)^2} \leq \frac{\mu_{\tau}(x)}{\hat \mu_{\tau}(x)} \leq (1+\gamma)^2$ with probability $1$ for all $x \in \{0,1\}^{\leq n}$; further, by definition we have $\mu_\tau(x) = \hat \mu(x)$ for $x \in \{0,1\}^n$. Thus \cref{asm:linfty-close-gen} is satisfied. For any instantiation of $\tau$, let the distribution of the output $X$ of $\Alg$ be denoted $\hat \nu_{\tau}$. The assumption that $\tvd{\mu_\tau}{\hat \nu_\tau} \leq \frac{\gamma}{8}$ and $\gamma < 1$, as well as the fact that  $e^{-\ep r} \leq \gamma/48$ (from \cref{eq:weak-params}) yields
\begin{align}
& \BP_{X \sim \hat \nu_\tau}(X_1 = \vinit_1) \geq \BP_{X \sim \mu_\tau}(X_1 = \vinit_1) - \frac{\gamma}{8} \geq \BP_{X \sim \muideal}(X_1 = \vinit_1) - \frac{\gamma}{8} - e^{-\ep r} \nonumber\\
& \geq \frac{1+\gamma}{2+\gamma} - \frac{\gamma}{8} - e^{-\ep r} \geq \frac{1}{2} + \frac{\gamma}{24} - e^{-\ep r} \geq \frac{1}{2} + \frac{\gamma}{48}\nonumber.
\end{align}
Above the second inequality uses \cref{lem:mu-muideal}, and the third inequality uses the definition of $\muideal$.
This contradicts \cref{lem:technical-random-lb} as long as $q \leq \frac{c \gamma^2}{\ep^2}$ (which also yields $q \cdot 2^{-k} < \gamma/100$), for a sufficiently small constant $c$. 

Note that we have shown a lower bound for an instance $\tau \in \MH_{n',k,q}$, for which the corresponding measures are defined on $\{0,1\}^{n'}$; but we can extend this to $\{0,1\}^n$ by simply padding the last $n-n'$ bits deterministically.

Finally, we remark that we can take the class $\MF$ to be equal to $\MH_{n',k,q}$. By \cref{def:hash-functions} (which relies on \cref{lem:kwise-indep-existence}), its cardinality satisfies 
\begin{align}
\log |\MF| = 2 + \log |\Phi_{n',k,q}| + \log |\Psi_{n',k,q}| \leq O(q k2^k) \leq O(n^4)\nonumber.
\end{align}
\end{proof}

\subsection{Additional lower bounds}
\label{sec:add-lbs}
Below we show a variant of \cref{thm:weak-lb} where we allow the cardinality of the class $\MF$ to be a free parameter: we derive a query lower bound of $\tilde \Omega(\sqrt{\log |\MF|})$. For simplicity, we take $\gamma$ to be an absolute constant. Moreover, define the function $f : \BN \times \BN \to \BR$ by 
\begin{align}
f(n, S) :=\max \left\{ \min\{ S, n \}, \min \left\{\frac{\sqrt{\log S}}{\log^2(2 + \log S)}, \frac{n^2}{\log^2 n} \right\} \right\}\nonumber.
\end{align}
It is easier to interpret $f(n,S)$ when we condition on the size of $S$, as follows --- namely, note that
\begin{align}
f(n,S) \asymp\begin{cases}
\min\{S, n\} &: S \leq \exp(n^2 \log^4 n) \\
\min \left\{\frac{\sqrt{\log S}}{\log^2(2 + \log S)}, \frac{n^2}{\log^2 n} \right\} &: \exp(n^2 \log^4 n) < S < \exp(n^4) .
\end{cases}\nonumber
\end{align}
\begin{theorem}[Size-sensitive lower bound]
\label{thm:size-sensitive-lb}

There is a constant $c > 0$ so that the following holds for all sufficiently large
$n \in \BN$ and all integers $S\geq 2$. There exists a class of distributions
$\MF \subseteq \Delta(\{0,1\}^n)$ with $|\MF| \leq S$ such that any randomized
algorithm which enjoys the following guarantee must make at least
$
c \cdot f(n,S)
$
queries to $\hat \mu$: for every $\mu^\star \in \MF$ and every mapping
$\hat \mu : \{0,1\}^{\leq n} \to \BR_{\geq 0}$ satisfying \cref{asm:linfty-close-gen} with respect to $\mu^\star$ for $R = 9/4$, 
the algorithm outputs a distribution $\hat \nu$ satisfying
$
\tvd{\mu^\star}{\hat \nu} \leq \frac{1}{16}.
$
\end{theorem}
\begin{proof}
It is without loss of generality to assume that $S \leq \exp(n^4)$ (as the function $f(n,S)$ is constant as a function of $S$ for $S > \exp(n^4)$). 

\paragraph{Case 1: large $S$.} We first show a lower bound of $\Omega\left(\min \left\{\frac{\sqrt{\log S}}{\log^2(2 + \log S)}, \frac{n^2}{\log^2 n} \right\} \right)$ on the query complexity. Write $n' =  c'\log(S)^{1/4}$, for a sufficiently small constant $c' > 0$ to be specified below. Then \cref{thm:weak-lb} with $\gamma = 1/2$ gives that there is a class of distributions $\MF' \subset \Delta(\{0,1\}^{n'})$ of size at most $\exp(O((n')^4))$ so that no randomized algorithm $\Alg$ enjoys the following guarantee: for any $\mu^\star \in \MF'$ and $\hat \mu : \{0,1\}^{\leq n'} \to \BR_{\geq 0}$ satisfying \cref{asm:linfty-close-gen} with $R = 9/4$, $\Alg$ makes $\frac{c \cdot (n')^2}{\log^2(n')}$ queries to $\hat \mu$ and outputs a random string $X' \in \{0,1\}^{n'}$ according to some distribution $\hat \nu$ which satisfies $\tvd{\mu^\star}{\hat \nu} \leq 1/16$, for some constant $c > 0$. Note that we can view $\MF'$ as a subset of $\Delta(\{0,1\}^n)$ by padding the last $n-n'$ coordinates to be $0$ with probability 1 (for all $\mu \in \MF'$). Thus, noting that $(n')^2 = \Theta(\sqrt{\log S})$ and $\log^2(n') \leq \log^2(2 + \log S)$,  as well as $|\MF'| \leq S$ as long as the constant $c'$ is sufficiently small, we obtain that any randomized algorithm satisfying the requirement of \cref{thm:size-sensitive-lb} must make $\Omega\left(\min \left\{\frac{\sqrt{\log S}}{\log^2(2 + \log S)}, \frac{n^2}{\log^2 n} \right\} \right)$ queries to $\hat \mu$. 

\paragraph{Case 2: Small $S$.} We next show a lower bound of $\Omega (\min\{S, n\})$ on the query complexity. It suffices to assume $S \leq n$. For $1 \leq i \leq S$, let $v_i \in \{0,1\}^S$ be defined by $v_{i,j} = \One{j \leq i}$. We define a class $\MF_S \subset \Delta(\{0,1\}^S)$ consisting of the $S-1$ measures $\mu_1, \ldots, \mu_{S-1}$, where for each $i \in [S-1]$,
\begin{align}
\mu_i := \frac{1}{2} \cdot \delta_{v_i} + \frac{1}{2} \cdot \delta_{v_S}.\nonumber
\end{align}
Next, for each $i \in [S-1]$, define $\hat \mu_i : \{0,1\}^{\leq S} \to \BR_{\geq 0}$ by
\begin{align}
\hat \mu_i(x) = \begin{cases}
\frac{1}{2} &: x \preceq v_i \mbox{ or } x \preceq v_S \\
0 &: \mbox{otherwise}.
\end{cases}\nonumber
\end{align}
We claim that $\hat \mu_i$ satisfies \cref{asm:linfty-close-gen} with respect to $\mu_i$, with $R = 2 \leq 9/4$. Indeed, fix any $x \in \{0,1\}^{\leq S}$. If $x \npreceq v_i$ and $x \npreceq v_S$, then $\mu_i(x) = \hat \mu_i(x) = 0$. If $x$ is a prefix of exactly one of $v_i, v_S$, then $\mu_i(x) = \hat \mu_i(x) = 1/2$. Finally, if $x$ is a prefix of both $v_i$ and $v_S$, then necessarily $x = 1^t$ for some $t \leq i$, and thus $\mu_i(x) = 1$ while $\hat \mu_i(x) = 1/2$. Therefore, whenever $\hat \mu_i(x) > 0$, we have
\begin{align}
1 \leq \frac{\mu_i(x)}{\hat \mu_i(x)} \leq 2 \leq 9/4.\nonumber
\end{align}
Moreover, for every leaf $x \in \{0,1\}^S$, the definition gives $\hat \mu_i(x) = \mu_i(x)$, as required.

It remains to show that any randomized algorithm requires $\Omega(S)$ queries. Fix any randomized algorithm making at most $q$ queries, and draw $I \sim \mathrm{Unif}([S-1])$. We will show that if $q \leq cS$ for a sufficiently small constant $c > 0$, then the algorithm cannot output a distribution within total variation distance $1/16$ of $\mu_I$.

Fix a realization $\rho$ of the algorithm's internal randomness, and consider the resulting deterministic algorithm. Let $\hat \mu_\circ$ denote the reference oracle defined by
\begin{align}
\hat \mu_\circ(x) := \frac{1}{2} \cdot \One{x \preceq v_S}.\nonumber
\end{align}
Run the deterministic algorithm against $\hat \mu_\circ$, and let $x_1^{\rho}, \ldots, x_q^{\rho}$ denote the resulting query sequence, with output $Y_{\rho} \in \{0,1\}^S$. Define
\begin{align}
T_{\rho} := \{ i \in [S-1] \mid \exists j \in [q] \mbox{ such that } x_j^{\rho} \preceq v_i \mbox{ and } x_j^{\rho} \npreceq v_S\}.\nonumber
\end{align}
For any fixed query $x$, there is at most one index $i \in [S-1]$ such that $x \preceq v_i$ and $x \npreceq v_S$: indeed, once a prefix leaves the all-ones path leading to $v_S$, it can lie on at most one of the paths to $v_1, \ldots, v_{S-1}$. Consequently, $|T_{\rho}| \leq q$.

Now fix any $i \in [S-1] \setminus T_{\rho}$. By definition of $T_{\rho}$, every query $x_j^{\rho}$ receives the same answer from $\hat \mu_i$ as from $\hat \mu_\circ$: if $x_j^{\rho} \preceq v_S$, then both oracles return $1/2$, and otherwise $x_j^{\rho} \npreceq v_i$, so both return $0$. Thus, when run against $\hat \mu_i$, the deterministic algorithm follows the exact same execution trace as it does against $\hat \mu_\circ$, and hence outputs the same string $Y_{\rho}$. Therefore,
\begin{align}
\BP_{I \sim \mathrm{Unif}([S-1])}(Y_{\rho} = v_I)
\leq \BP(I \in T_{\rho}) + \BP(Y_{\rho} = v_I,\ I \notin T_{\rho})
\leq \frac{|T_{\rho}|}{S-1} + \frac{1}{S-1}
\leq \frac{q+1}{S-1}.\nonumber
\end{align}
Since this bound holds for every realization $\rho$ of the algorithm's internal randomness, averaging over $\rho$ yields
\begin{align}
\BP(X = v_I) \leq \frac{q+1}{S-1},\label{eq:small-s-success-bound}
\end{align}
where $X$ denotes the output of the original randomized algorithm.

On the other hand, if the algorithm were to satisfy the guarantee of \cref{thm:size-sensitive-lb}, then for every $i \in [S-1]$ its output distribution $\hat \nu_i$ under oracle $\hat \mu_i$ would satisfy $\tvd{\hat \nu_i}{\mu_i} \leq 1/16$. Since $\mu_i(v_i) = 1/2$, this implies
\begin{align}
\BP(X = v_i \mid I = i) = \hat \nu_i(v_i) \geq \mu_i(v_i) - \tvd{\hat \nu_i}{\mu_i} \geq \frac{1}{2} - \frac{1}{16} = \frac{7}{16}.\nonumber
\end{align}
Averaging over $I \sim \mathrm{Unif}([S-1])$, we obtain $\BP(X = v_I) \geq 7/16$, which contradicts \cref{eq:small-s-success-bound} as long as $q \leq cS$ for a sufficiently small universal constant $c > 0$.

Finally, we may view $\MF_S$ as a subset of $\Delta(\{0,1\}^n)$ by padding the last $n-S$ coordinates with zeros deterministically. Since $|\MF_S| = S-1 \leq S$, this establishes the desired $\Omega(S) = \Omega(\min\{S,n\})$ lower bound in the small-$S$ regime, and completes the proof.
\end{proof}

We next show a variant of the construction which yields a stronger lower bound, in the following sense: any algorithm which makes strictly subquadratically many queries to $\hat \mu$ cannot sample from any distribution whose \emph{KL divergence} from $\mu$ is strictly sublinear (in $n$).

Fix positive integers $n,k,q$ as above, as well as a positive integer $s$.  Recall the definition of the class $\MH_{n,k,q,\ep}$ in \cref{def:hash-functions}, as well as, for $\tau = (\vinit, \phi, \psi) \in \MH_{n,k,q,\ep}$, the distribution $\mu_\tau \in \Delta(\{0,1\}^n)$ and mapping $\hat \mu_\tau : \{0,1\}^{\leq n} \to \BR_{\geq 0}$ defined in \cref{eq:def-muhat-zhat}. 

Write $\bar{n} := ns$. We will often interpret elements of $\{0,1\}^{\bar{n}}$ by partitioning their coordinates into $s$ blocks of size $n$, as follows: for $x \in \{0,1\}^{\bar{n}}$, we write $x = (x\^1, \ldots, x\^s)$, where each $x\^i \in \{0,1\}^n$ is the $i$th block of $x$ consisting of the coordinates $(x_{(i-1)n+1}, \ldots, x_{in})$.
Next, for $\tau\^1, \ldots, \tau\^s \in \MH_{n,k,q,\ep}$, we define a distribution $\mu_{\tau\^{1:s}} \in \Delta(\{0,1\}^{\bar{n}})$, and mapping $\hat \mu_{\tau\^{1:s}} : \{0,1\}^{\leq \bar{n}} \to \BR_{\geq 0}$ as follows: for $x = (x\^1, \ldots, x\^s) \in \{0,1\}^{\bar{n}}$ with each $x\^i \in \{0,1\}^n$, we set
\begin{align}
\mu_{\tau\^{1:s}}(x) := \prod_{i=1}^s \mu_{\tau\^i}(x\^i)\nonumber.
\end{align}
Further, for $x = (x\^1, \ldots, x\^a) \in \{0,1\}^{\leq \bar{n}}$ with $x\^1, \ldots, x\^{a-1} \in \{0,1\}^n$ and $x\^a \in \{0,1\}^{\leq n}$, we set
\begin{align}
\hat \mu_{\tau\^{1:s}}(x) := \prod_{i=1}^{a} \hat \mu_{\tau\^i}(x\^i).\nonumber
\end{align}

\begin{lemma}
\label{lem:gamma-close}
For any $\tau\^{1:s} \in \MH_{n,k,q,\ep}^s$ and $x \in \{0,1\}^{\leq \bar{n}}$, we have $\frac{1}{(1+\gamma)^{2}} \leq \frac{\mu_{\tau\^{1:s}}(x)}{\hat \mu_{\tau\^{1:s}}(x)} \leq (1+\gamma)^{2}$ for all $x \in \{0,1\}^{\leq \bar{n}}$.
\end{lemma}
\begin{proof}
  Let us write $x = (x\^1, \ldots, x\^a) \in \{0,1\}^{\leq \bar{n}}$ with $x\^1, \ldots, x\^{a-1} \in \{0,1\}^n$ and $x\^a \in \{0,1\}^{\leq n}$. 
By definition of $\mu_{\tau\^i}, \hat \mu_{\tau\^i}$, for $i < a$, since we have $x\^i \in \{0,1\}^n$, it holds that $\mu_{\tau\^i}(x\^i) = \hat \mu_{\tau\^i}(x\^i)$. Thus
\begin{align}
\frac{\mu_{\tau\^{1:s}}(x)}{\hat \mu_{\tau\^{1:s}}(x)} = \frac{\prod_{i=1}^a \mu_{\tau\^i}(x\^i)}{\prod_{i=1}^a \hat \mu_{\tau\^i}(x\^i)} = \frac{\mu_{\tau\^a}(x\^a)}{\hat \mu_{\tau\^a}(x\^a)} \in [(1+\gamma)^{-2}, (1+\gamma)^2],\nonumber
\end{align}
where the final inclusion uses \cref{lem:mu-muhat-ratio}. (Above, we interpret $0/0 = 1$.)
\end{proof}

\begin{theorem}
  \label{thm:strong-lb}
  There is a constant $c > 0$ so that the following holds, for any sufficiently large $\bar{n} \in \BN$, parameter $\gamma \in (1/\bar{n},1)$, and $\KLparam \geq \gamma^2$. There is no randomized algorithm $\Alg$ which enjoys the following guarantee: for any distribution $\mu^\star \in \Delta(\{0,1\}^{\bar{n}})$ and $\hat \mu : \{0,1\}^{\leq \bar{n}} \to \BR_{\geq 0}$ satisfying \cref{asm:linfty-close-gen} with $R = (1+\gamma)^2$, $\Alg$ makes $o\left( \frac{\bar{n}^2 \gamma^6}{\KLparam^2 \log^2(\bar{n}) \cdot \log^2(2/\gamma)} \right)$ queries to $\hat \mu$ and outputs a (random) vector $X \in \{0,1\}^{\bar{n}}$, according to some distribution $\hat \nu$, satisfying $\kld{\mu^\star}{\hat \nu} \leq \KLparam$.
\end{theorem}
\begin{proof}
If $\frac{800 \KLparam}{\gamma^2} > \frac{\bar{n}}{ 100 \log \bar{n}}$, then the claimed query bound $c' \bar{n}^2\gamma^6 / (\KLparam^2 \log^2 \bar{n} \cdot \log^2(2/\gamma))$ is $O(1)$, so the result is trivial. Thus we may assume $\frac{800 \KLparam}{\gamma^2} \leq \frac{\bar{n}}{ 100 \log \bar{n}}$. Take
\begin{align}
s := \left\lceil \frac{800}{\gamma^2} \KLparam \right\rceil,
\qquad
k := \lceil 10 \log_2 \bar{n} \rceil,
\qquad
r := \left\lfloor \frac{\lfloor \bar{n}/s \rfloor - 1}{2k} \right\rfloor,\nonumber\\
n := 1 + 2k \cdot r,
\qquad
\ep := \frac{32 \log(96/\gamma) \cdot k}{n},
\qquad
\bar{n}' :=ns \nonumber.
\end{align}
Since $\frac{800 \KLparam}{\gamma^2} \leq \frac{\bar{n}}{ 100 \log \bar{n}}$, we have $s \leq \frac{\bar n}{50 \log \bar n}$ and hence $r \geq 1$. Since $\bar{n}' \leq \bar{n}$, it suffices to prove the lower bound with $\bar{n}'$ in place of $\bar{n}$ (we can simply pad the last $\bar{n}-\bar{n}'$ coordinates deterministically). 

Suppose to the contrary that such $\Alg$ exists, and consider the execution of $\Alg$ with respect to the distribution $\mu^\star = \mu_{\tau\^{1:s}}$ and mapping $\hat \mu_{\tau\^{1:s}}$, where $\tau\^{1:s}  \sim \MH_{n,k,q,\ep}^s$ is drawn uniformly at random. Note that $\mu_{\tau\^{1:s}} \in \Delta(\{0,1\}^{\bar{n}'})$, and that the domain of $\hat\mu_{\tau\^{1:s}}$ is $\{0,1\}^{\leq \bar{n}'}$.  First, \cref{lem:gamma-close} ensures that we indeed have that $\frac{1}{(1+\gamma)^2} \leq \frac{\mu_{\tau\^{1:s}}(x)}{\hat \mu_{\tau\^{1:s}}(x)} \leq (1+\gamma)^2$ for all $x \in \{0,1\}^{\leq \bar{n}'}$ with probability $1$. Further, the definitions of $\mu_{\tau\^{1:s}}$ and $\hat \mu_{\tau\^{1:s}}$ ensure that $\mu_{\tau\^{1:s}}(x) = \hat \mu_{\tau\^{1:s}}(x)$ for all $x \in \{0,1\}^{\bar{n}}$. Hence \cref{asm:linfty-close-gen} is satisfied. 

For any instantiation of $\tau\^{1:s}$, let the distribution of the output $X = (X\^1, \ldots, X\^s) \in \{0,1\}^{\bar{n}}$ of $\Alg$ be denoted $\hat \nu_{\tau\^{1:s}}$, and the distribution of the first bits of each block, namely $(X\^1_1, \ldots, X\^s_1) \in \{0,1\}^s$, be denoted $\hat \beta_{\tau\^{1:s}}$. 

For $\tau = (\vinit, \phi, \psi) \in \MH_{n,k,q,\ep}$, let $\alpha_\tau \in \Delta(\{0,1\})$ be the distribution which puts mass $\frac{1+\gamma}{2+\gamma}$ on $\vinit$ and mass $\frac{1}{2+\gamma}$ on $1-\vinit$; in other words, $\alpha_\tau$ is the distribution of $X_1$, for $X = (X_1, \ldots, X_n) \sim \mu_\tau$. Let $\alpha_{\tau\^{1:s}} := \alpha_{\tau\^1} \times \cdots \times \alpha_{\tau\^s} \in \Delta(\{0,1\}^s)$. Next, we may write
\begin{align}
\KLparam \geq & \E_{\tau\^{1:s} \sim \MH_{n,k,q,\ep}^s} \left[ \kld{\mu_{\tau\^{1:s}}}{\hat \nu_{\tau\^{1:s}}} \right] 
\geq  \E_{\tau\^{1:s} \sim \MH_{n,k,q,\ep}^s} \left[ \kld{\alpha_{\tau\^{1:s}}}{\hat \beta_{\tau\^{1:s}}} \right] \label{eq:exp-tau-alphabeta},
\end{align}
where the first inequality uses the assumption on $\Alg$ and the second inequality uses the data processing inequality. Next, for any $\tau\^{1:s} \in \MH_{n,k,q,\ep}^s$, the chain rule for KL divergence gives
\begin{align}
\kld{\alpha_{\tau\^{1:s}}}{\hat \beta_{\tau\^{1:s}}} =&  \sum_{i=1}^s \E_{(X\^1, \ldots, X\^{i-1}) \sim \alpha_{\tau\^{1:i-1}}}\left[\kld{\alpha_{\tau\^i}}{\hat \beta_{\tau\^{1:s}}(X\^i = \cdot \mid X\^1, \ldots, X\^{i-1})}\right]\label{eq:alphabeta-chain}. 
\end{align}
Combining \cref{eq:exp-tau-alphabeta,eq:alphabeta-chain} as well as convexity of KL divergence gives that there exists some $i \in [s]$ such that
\begin{align}
  & \E_{\tau\^{1:s} \sim \MH_{n,k,q,\ep}^s}\left[ \kld{\alpha_{\tau\^i}}{\hat \beta_{\tau\^{1:s}}(X\^i = \cdot)}\right]\nonumber\\
\leq & \E_{\tau\^{1:s} \sim \MH_{n,k,q,\ep}^{s}}\E_{(X\^1, \ldots, X\^{i-1}) \sim \alpha_{\tau\^{1:i-1}}}\left[\kld{\alpha_{\tau\^i}}{\hat \beta_{\tau\^{1:s}}(X\^i = \cdot \mid X\^1, \ldots, X\^{i-1})}\right] \leq  \frac{\KLparam}{s}.\label{eq:algi-output}
\end{align}
This contradicts \cref{lem:technical-random-lb}, however: consider the algorithm $\Alg'$ which functions as follows, given access to $\hat \mu_{\tau\^i}$, for $\tau\^i \sim \MH_{n,k,q,\ep}$. It draws $\tau\^j \sim \MH_{n,k,q,\ep}$ uniformly for each $j \in [s]$ with $j \neq i$, and then simulates $\Alg$, using $\hat \mu_{\tau\^i}$ together with $\hat \mu_{\tau\^j}$ (where $\tau\^j$ are known for $j \neq i$) to simulate access to the oracle $\hat \mu_{\tau\^{1:s}}$. The output of $\Alg'$ is the first bit of the $i$th block of the output of $\Alg$. By \cref{eq:algi-output}, $\Alg'$ outputs a distribution whose KL divergence from $\alpha_{\tau\^i}$ is at most $\frac{\KLparam}{s}$; in particular, by Pinsker's inequality, the probability that the output bit of $\Alg'$ is equal to $(\vinit)\^i$ is at least
\begin{align}
\BP_{X \sim \mu_{\tau\^i}}(X_1 = (\vinit)\^i) - \sqrt{\frac{\KLparam}{2s}} \geq \frac{1+\gamma}{2+\gamma} - \sqrt{\frac{\KLparam}{2s}} - e^{-\ep r} \geq \frac{1+\gamma}{2+\gamma} - 2 \cdot \frac{\gamma}{20} \geq \frac 12 + \frac{\gamma}{15}.\nonumber
\end{align}
Above, the first inequality uses \cref{lem:mu-muideal}. The second inequality uses that $\sqrt{\KLparam/(2s)} \leq \gamma/20$ (from the definition of $s$) and that $e^{-\ep r} \leq \gamma/96$ (from the definition of $\ep$ and $r$ and the fact that $2kr \geq 1$).

This contradicts \cref{lem:technical-random-lb} as long as the number of queries $q$ made by $\Alg$ satisfies $q \leq c'\gamma^2/ \ep^2$, for a sufficiently small constant $c' > 0$ (note that $q \leq c'\gamma^2/\ep^2$ also ensures that $q \cdot 2^{-k} \leq q/\bar{n}^{10} \leq \bar{n}^{-5} \leq \gamma/100$ for sufficiently small $\bar{n}$). In turn, the condition $q \leq c'\gamma^2/\ep^2$ is satisfied as long as $q\leq c \cdot \frac{\bar{n}^2 \gamma^6}{\KLparam^2 \log^2(\bar{n}) \cdot \log^2(1/\gamma)}$ for a sufficiently small constant $c$. 
\end{proof}

\section*{Acknowledgments}
AM is supported in part by a Microsoft Trustworthy AI Grant, NSF award CCF-2430381, ONR grant N00014-22-1-2339, and a David and Lucile Packard Fellowship. DR is supported by NSF awards CCF-2430381 and DMS-2022448, and ONR grant N00014-22-1-2339. 

\section*{AI Disclosure}
GPT-5.4 and GPT-5.2 were used to assist in the writing of this paper. Their principal use was in suggesting ``first drafts'' for several of the proofs of the lemmas. The final writing of the paper, including the proofs, was done primarily by the authors with some fragments suggested by GPT. Correctness of all proofs was verified by the authors.

\newpage
\appendix
\section{Relationship between counting-to-sampling reduction and LLM inference}
\begin{proposition}
  \label{prop:muhat-vhat}
  Fix $\Sigma, n, \delta, R$. 
Suppose that $\Alg$ is a randomized algorithm which has query access to some $\hat \mu : \Sigma^{\leq n} \to \BR_{\geq 0}$ satisfying \cref{eq:r-mult-intro} with respect to some $\mu^\star \in \Delta(\Sigma^n)$ and the ratio $R$. Suppose that:
\begin{itemize}
  \item $\Alg$ makes at most $q$ queries to $\hat \mu$, after which it outputs some $X \in \Sigma^n$ distributed according to some distribution $\hat \nu$ with $\tvd{\hat \nu}{\mu^\star} \leq \delta$.
\item  The distribution of the output of $\Alg$ is unchanged if every answer returned by the oracle $\hat \mu$ is multiplied by the same positive constant.
\end{itemize}

Then there is a randomized algorithm $\Alg'$, which is given query access to $\piref \in \Delta(\Sigma^n)$ and $\hat V : \Sigma^{\leq n} \to \BR$ as in \cref{sec:inference-llms}, satisfying $\| \hat V - V^\star\|_\infty \leq \log R$ and $V^\star(x) = \hat V(x)$ for $x \in \Sigma^n$. The algorithm $\Alg'$ makes at most $q$ queries to $\piref, \hat V$ and outputs some $X \in \Sigma^n$ distributed according to some distribution $\hat \nu$ with $\tvd{\hat \nu}{\mu^\star} \leq \delta$.
\end{proposition}
\begin{proof}
The algorithm $\Alg'$ simply simulates $\Alg$, answering queries to $\hat \mu(x_{1:i})$ by returning $\hat \mu(x_{1:i}) = \piref(x_{1:i}) \cdot \exp(\hat V(x_{1:i}))$. By scale-invariance of $\Alg$, the distribution of the output of $\Alg'$ is the same as if it had instead answered queries by returning $\frac{1}{Z} \cdot \hat \mu(x_{1:i})$, which is an $R$-multiplicative approximation of $\mu^\star$ by the assumption that $\| \hat V - V^\star\|_\infty \leq \log R$. (Further, we have $\frac{1}{Z} \hat \mu(x) = \mu^\star(x)$ for $x \in \Sigma^n$.) Thus, the output distribution of $\Alg'$ is $\delta$-close to $\mu^\star$, as desired.
\end{proof}
\section{Concentration inequalities} 
\begin{lemma}[Freedman's inequality]
  \label{lem:freedman}
  Suppose that $(X_i)_{i \in [n]}$ is a sequence of random variables adapted to a filtration $(\SF_i)_{i \in [n]}$. Then for any $\delta \in (0,1)$, with probability at least $1-\delta$,
    \begin{align}
\sum_{i=1}^n X_i \leq \sum_{i=1}^n \log \E[\exp( X_i)\mid \SF_{i-1}] + \log(1/\delta)\nonumber.
    \end{align}
\end{lemma}

\bibliographystyle{alpha}
\bibliography{refs}

\end{document}